\documentclass[12pt]{article}
\usepackage[T1]{fontenc}
\usepackage[dvips]{graphicx}
\graphicspath{{images/}}
\usepackage{verbatim}
\usepackage{amsmath,amsthm}
\usepackage{xspace}
\usepackage{pifont}
\usepackage{graphicx}
\usepackage{amssymb}
\usepackage{epic, eepic}
\usepackage{dsfont}
\usepackage{amssymb}
\usepackage{makeidx}
\usepackage{mathrsfs}
\usepackage{exscale}
\usepackage{color} 
\usepackage{overpic} 
\usepackage{bm}
\usepackage{bbm}
\usepackage{booktabs} 
\usepackage{color, colortbl}
\usepackage{subcaption}
\usepackage[numbers]{natbib}

\RequirePackage[colorlinks,citecolor=blue,urlcolor=blue]{hyperref}

\definecolor{Gray}{gray}{0.9}

\usepackage{amsmath,afterpage}
\usepackage{epsf}
\usepackage{graphics,color}

\def\0{\mathbf{0}}

\def\eps{\varepsilon}

\def\lam{\lambda}
\def\rr{\rightarrow}

\def \< {\langle}
\def \> {\rangle}

\def\ol{\overline}

\def\beqa{\begin{eqnarray}}
\def\eeqa{\end{eqnarray}}
\def\beqas{\begin{eqnarray*}}
\def\eeqas{\end{eqnarray*}}

\newtheorem{theorem}{Theorem}[section]
\newtheorem{lemma}[theorem]{Lemma}

\newtheorem{proposition}[theorem]{Proposition}

\newtheorem{corollary}[theorem]{Corollary}
\newtheorem{cor}[theorem]{Corollary}

\newtheorem{remark}[theorem]{Remark}

\newtheorem{definition}[theorem]{Definition}
\newtheorem{assumption}[theorem]{Assumption}
\numberwithin{equation}{section}
\newcommand{\hatd}[1]{{}}

\newcommand{\bd}{\begin{displaymath}}
\newcommand{\ed}{\end{displaymath}}
\newcommand{\be}{\begin{equation}}
\newcommand{\ee}{\end{equation}}
\newcommand{\bq}{\begin{eqnarray}}
\newcommand{\eq}{\end{eqnarray}}
\newcommand{\bn}{\begin{eqnarray*}}
\newcommand{\en}{\end{eqnarray*}}

\newcommand{\re}{\mathds{R}}

\def\wt{\widetilde}

\def\P{\mathbb{P}}

\usepackage{authblk}

\title{Trading with the Crowd}
\author[1]{Eyal Neuman}
\author[2]{Moritz Vo\ss  }
\affil[1]{Department of Mathematics, Imperial College London}
\affil[2]{Department of Mathematics,
University of California, Los Angeles}
 
\begin{document}

 \vspace{-0.5cm}
\maketitle

\begin{abstract}
We formulate and solve a multi-player stochastic differential game between financial agents who seek to cost-efficiently liquidate their position in a risky asset in the presence of jointly aggregated transient price impact,  
along with taking into account a common general price predicting signal. 
The unique Nash-equilibrium strategies reveal how each agent's liquidation policy adjusts the predictive trading signal to the aggregated transient price impact induced by all other agents. This unfolds a quantitative relation between trading signals and the order flow in crowded markets. We also formulate and solve the corresponding mean field game in the limit of infinitely many agents. We prove that the equilibrium trading speed and the value function of an agent in the finite $N$-player game converges to the corresponding trading speed and value function in the mean field game at rate $O(N^{-2})$.  In addition, we prove that the mean field optimal strategy provides an approximate Nash-equilibrium for the finite-player game. 
\end{abstract} 


\begin{description}
\item[Mathematics Subject Classification (2010):] 49N80, 49N90, 93E20, 60H30 
\item[JEL Classification:] C73, C02, C61, G11
\item[Keywords:] crowding, optimal portfolio liquidation, price
  impact, mean field games,
  optimal stochastic control, predictive signals
\end{description}

\bigskip

\section{Introduction}
 The phenomenon of \textit{crowding} in financial markets has gained
 an increasing attention from both academics and  financial
 institutions over the past couple of decades. It is a subject of
 numerous research works studying both theoretical and empirical
 aspects of the topic, including
 \cite{Cont2000,volpati2020zooming,Bucci2020,Barroso2017,Caccioli2015,Caccioli2014,Khandani2008}
 among others. Crowding is often considered to be an explanation for
 sub-par performances of investments as well as the development of
 systemic risk in financial markets. The presence of largely
 overlapping portfolios comes at the expense of portfolio managers' profits, also in terms of transaction costs, as affine positions usually lead to similar trades.

The existing literature on crowding concentrates both on analytic models that explain some aspects of crowded markets behaviour and on data driven statistical models.  In the first class, Cont and Bouchaud \cite{Cont2000} proposed a simple mathematical model in which the communication structure between agents gives rise to heavy tailed distribution for stock returns. This established a theoretical connection between crowding and stock markets shortfall. The aforementioned portfolios overlap was shown to be a considerable factor in the August 2007 Quant Meltdown. Cont and Wagalath \cite{Cont-Wag13} proposed a simple multi-period model of price impact from trading in a market with multiple assets. Their model illustrated how feedback effects due to distressed selling and short selling lead to endogenous correlations between asset classes and it provided a quantitative framework to evaluate strategy crowding as a risk factor (see also \cite{Cont-wag16}). 

Within the class of statistical models, Khandani and Lo \cite{Khandani2008} used simulated returns of overlapping equity portfolios and showed that combined effects of portfolio deleveraging followed by a temporary withdrawal of market-making risk capital was one of the main drivers of the 2007 Quant Meltdown.  Caccioli et al. \cite{Caccioli2015,Caccioli2014} developed a mathematical model for a network of different banks holding overlapping portfolios. They investigated the circumstances under which systemic instabilities may occur as a result of various parameters, such as market crowding and price impact.  Volpati et al. \cite{volpati2020zooming} measured significant levels of crowding in U.S. equity markets for momentum signals  as well as for Fama-French factors signals. In \cite{Mich-Neum20} an index reconstruction methodology was developed in order to measure the crowding effect on Russell indexes around reconstitutions events.  

A new approach that connects optimal execution of multiple agents to
crowding phenomena was proposed by \citet{CardaliaguetLehalle:18}.
More specifically, a mean field game with an infinite number of agents, where each agent executes a large order on the same risky asset was considered. In their model, the aggregated permanent price impact created by all players' transactions is modeled as an exogenous process, which satisfies a consistency condition, while temporary price impact influences each agent individually. The solution of the game showed qualitatively that in an infinite player setting, the optimal trading speed is deterministic and that it is optimal to follow the crowd but not too fast, as this could create additional trading costs.  

In this work we further extend and develop the model proposed in
\cite{CardaliaguetLehalle:18} in order to reveal new properties and
insights on crowding that arise in optimal execution framework. More
precisely, we formulate and solve a multi-player stochastic differential game between traders who execute large orders on one risky asset in the presence of individual temporary price impact and jointly aggregated transient price impact. We also assume that traders are observing a common exogenous price predicting signal. The unique Nash-equilibrium strategies show how each agent's liquidation strategy adjusts the predictive trading signal for the accumulated transient price distortion induced by all other agents' price impact. This unfolds a qualitative relation between trading signals and the agents' aggregated order flow in crowded markets. We refer to Section \ref{sec:finitePlayer} for the model setup and to Theorem \ref{thm:main-finite} for the solution of the game. One can observe from the explicit solution to the game in \eqref{eq:opt_ui} that the aggregated price distortion $Y^{u^N}$ and the exogenous trading signal $A$ are coupled. We therefore conclude that the distortion is acting as an endogenous signal, created by the trading strategies of all $N$ players. In the infinite player game we show that a similar coupling holds (see \eqref{eq:opt_uiInf}), however the price distortion $\wt Y^{\tilde \nu }$ is determined by an independent equation \eqref{eq:opt_ubar-MeanField_infAggregated}, then it is plugged into the optimal trading speed of the individual agent in \eqref{eq:opt_uiInf}, essentially acting as an exogenous signal. See Remarks \ref{rem-sig} and \ref{rem-mg-conc} for further details. From the explicit solutions which were mentioned above, we notice that both in the finite-player and infinite-player games the signal and the price distortion can follow a similar direction at least for a while. This is demonstrated in figure \ref{fig:ill3}, where a decreasing signal amplifies the sell strategies and as a result, a negative price distortion is created. Note that at some point in time the inventory penalties force the agents to close their position and the price impact turns in an opposite direction to the signal.  

Here we summarise the main financial interpretation of our analysis: 
\begin{itemize} 
\item[\textbf{(i)}] \emph{The consideration of a decaying price distortion in the finite player game points out that the cumulative order-flow is in fact an endogenous signal, which is observed and used by all traders in the game. }
\item[\textbf{(ii)}] \emph{In various scenarios the order-flow amplifies the effect of the exogenous price predicting signal on the price process and on the traders' execution strategies at equilibrium.}
\end{itemize} 
 These interesting and surprising results cannot be derived by the
 \citet{CardaliaguetLehalle:18} model, as their model assumes that the
 permanent price impact is an exogenous signal which satisfies a consistency condition. Indeed in the mean
 field setting the contribution of each agent to this signal is
 infinitesimal. Also Cardaliaguet and Lehalle did not incorporate a
 price predicting exogenous signal in their model, so their mean
 field optimal strategies are in fact deterministic and the order-flow
 amplification effect does not appear. In fact our results show that
 for a sell strategy for example, a negative price predicting signal
 will motivate the traders to sell more aggressively and their
 associated order-flow will drop the price down even further and
 create an endogenous signal in the same direction of the predictive
 signal. The only reason that the system remains in equilibrium is due
 to the penalties of holding inventory during and at the end of the
 trading period in the agents' performance functional. See Section \ref{sec:illustrations} for a qualitative analysis and illustrations of the results. 

In this work we also formulate and solve the corresponding mean field
game, which describes the limit of infinitely many agents (see Section
\ref{sec:infinitePlayer} and Theorem \ref{thm:main-infMFG}). We prove
in Theorem \ref{thm-strat-con} and Corollary \ref{corr-val} that the
equilibrium trading speed and value function of an agent in the finite
$N$-player game converges to the corresponding trading speed and value
function of an agent in the mean field game at rate $O(N^{-2})$. Similar convergence result with relaxed assumptions but without the convergence rate  is given in and Theorem \ref{thm-con-long}.
This
concludes that the aggregated order flow, which appears in the model
as the cumulative transient price impact of all agents, becomes an
exogenous signal as the number of agents tends to infinity. As a
result, we justify the a priori assumption of the exogenous price impact process in the simplified model of Cardaliaguet and Lehalle. 
Finally we prove in Theorem \ref{thm-eps-nash} that the mean field optimal strategy provides an approximate Nash-equilibrium for the finite-player game.  

Some additional papers on optimal execution in multiplayer and
infinite player games have appeared recently, without specific
reference to crowding. We will describe these results in short and
explain how this paper improves and extends them from the mathematical
point of view.  We will start with papers which solve only the mean
field game (i.e., the infinite player setting).  Casgrain and
Jaimungal \cite{, CasgrainJaimungal:18, CasgrainJaimungal:20} studied
a mean field game in which each agent is executing a large order while
creating both temporary and permanent price impact. Their model
generalises the basic model of 
\citet{CardaliaguetLehalle:18}, as it assumes that traders may have
differing beliefs or partial information on the price
process. \citet{HuangJaimungalNourian:19} also extended the mean field
model in \cite{CardaliaguetLehalle:18} by introducing three classes of
traders: a large agent, small high frequency traders and noise
traders. Finally, \citet{FuGraeweHorstPopier:20} extended the model to
liquidation under asymmetric information. Our results improve and
extend these papers, since we also solve the corresponding $N$-player
game, which is known to be less tractable. We further assume that the
traders create in addition a transient price impact that depends on
the entire trading paths of all players. Lastly, we prove the
convergence of the optimal strategy of the $N$-player game to the mean
field game equilibrium.

A few recent papers deal with finite player execution games for the
case that traders create transient price impact.  \citet{Strehle:17}
and \citet{SchiedStrehelZhang:17} worked on the problem in continuous
time, while \citet{SchiedZhang:19} and \citet{LuoSchied:20} studied
the discrete time setting. These references describe a special case of
our model as they do not include a predictive signal, which is
responsible for the randomness of the equilibrium strategies, and they do not prove convergence results to a mean field limit, as done in this paper.  

Finite player price impact games and mean field games with permanent
price impact were studied by     \citet{EvangelistaThamsten:20},
\citet{DrapeauLuoSchiedXiong:19} and \citet{FeronTankovTinsi:20},
where a special attention is given to the last two papers, in which
the convergence of the finite player equilibrium to the mean field equilibrium was derived. We remark that the convergence results in \cite{DrapeauLuoSchiedXiong:19} and \cite{FeronTankovTinsi:20} do not derive  the convergence rate, and of course their model did not include transient price impact which plays a crucial role in our model. Moreover, in these papers the convergence proof uses particular features of the model, which do not apply to the transient price impact case. In Section \ref{sec:convergence} we develop a method which not only provides the rate of convergence but could be adapted to a more general class of models, which translate at equilibrium to systems of FBSDEs. 

Finally some additional convergence results on a finite player game to a mean field game were derived recently for liquidation games with self-exciting order flow by \citet{fu2020}, which did not include a price predicting signal. These convergence results are quite different than the convergence results in this paper, as they derive the convergence of a game with stochastic i.i.d. noise in the transient price impact coefficients. In our setting the source of randomness is a common noise, which is the exogenous signal.  Moreover the convergence results in \citet{fu2020} do not derive the convergence rate. 

\paragraph{Structure of the paper:} In Section \ref{sec:finitePlayer} we define the finite player game and derive the Nash equilibrium. In Section \ref{sec:infinitePlayer} we present the corresponding mean field game and derive its equilibrium. In Section \ref{sec:approxNash} we present our convergence results and approximated Nash equilibrium. Section \ref{sec:illustrations} is dedicated to illustrations of the equilibrium strategies. Finally, Section \ref{sec:convergence}--\ref{subsec:proofs-matrixexponentials} are dedicated to the proofs of the main results.

\section{A Finite Player Game} \label{sec:finitePlayer}
\subsection{Model Setup} 
Our first goal is to adopt and extend the single-agent, signal-adaptive optimal execution problem with transient price impact from~\cite{N-V19} to a finite $N$-player stochastic differential price impact game. 

As usual we begin with fixing a finite deterministic time horizon $T>0$ and a filtered probability space $(\Omega, \mathcal F,(\mathcal F_t)_{0 \leq t\leq T}, \P)$ satisfying the usual conditions of right continuity and completeness. We further denote by $\mathcal H^2$ the class of all (special) semimartingales~$P=(P_t)_{0 \leq t\leq T}$ allowing for a canonical decomposition $P = L + A$ into a (local) martingale $L=(L_t)_{0 \leq t\leq T}$ and a predictable finite-variation process $A=(A_t)_{0 \leq t\leq T}$ such that
\be \label{ass:P} 
E \left[ \langle  L  \rangle_T \right] + E\left[\left( \int_0^T |dA_s| \right)^2 \right] < \infty.  
\ee
Next, we introduce a class of $N \in \mathbb{N}$ agents. Each agent $i \in \{1,\ldots,N\}$ has an initial position of $x^{i,N} \in \re $ shares in a risky asset and their number of shares held at time $t\in [0,T]$ is given by 
\begin{equation} \label{def:Xi}
X^{u^{i,N}}_t \triangleq x^{i,N}-\int_0^t u^{i,N}_s ds
\end{equation}
with a selling rate $(u^{i,N
}_s)_{0 \leq s \leq T}$ chosen from a set of admissible strategies
\be \label{def:admissset} 
\mathcal A \triangleq \left\{ v \, : \, v \textrm{ progressively measurable s.t. } E\left[ \int_0^T v_s^2 ds \right] <\infty \right\}.
\ee
We denote 
\begin{equation*}
u^N \triangleq (u^{1,N}, \ldots, u^{N,N}) \in \mathcal{A}^N,
\end{equation*}
and assume similar to the single-agent case in~\cite{N-V19} that the agents' collective trading activity $u^N$ induces a common transient price impact on the risky asset's execution price. Specifically, all agents' orders are filled at prices
\be \label{def:S}
S^{u^N}_{t} \triangleq P_{t}  - \kappa Y^{u^N}_t \qquad (0 \leq t \leq T), 
\ee
where $P$ denotes some unaffected price process in $\mathcal H^2$ and
\be \label{def:Y} 
Y^{u^N}_t \triangleq e^{-\rho t} y + \gamma \int_0^t e^{-\rho (t-s)} \left( \frac{1}{N}\sum_{i=1}^{N} u^{i,N}_s \right) ds \qquad (0 \leq t \leq T),
\ee
captures an aggregated linear and exponentially decaying price distortion from the unaffected level $P$ with some constants $\kappa, \gamma >0$, resilience rate $\rho >0$, and some initial value $y > 0$. In addition, we also assume that each agent's trading incurs an individual slippage cost $\lambda >0$ which is levied on their respective quadratic turnover rate and accumulates to       
\begin{equation*}
\lambda \int_0^T (u^{i,N}_{t})^2 dt, 
\end{equation*}
up to a terminal time $T$. For each agent $i \in \{1,\ldots,N\}$ we denote by
\begin{equation*}
u^{-i,N} \triangleq (u^{1,N},\ldots, u^{i-1,N}, u^{i+1,N}, \ldots u^{N,N}) \in \mathcal{A}^{N-1},
\end{equation*}
the other agents' trading activities. The $i$-th agent's objective is to optimally unwind her initial position $x^{i,N} \in \re$ by time $T$. The execution is done while taking into account both the interaction with all other agents' strategies $u^{-i,N}$ through the jointly generated transient price impact $Y^{u^N}$ in~\eqref{def:Y}, as well as the risky asset's price signal $A$. These considerations are accounted by maximizing the performance functional
\begin{equation} \label{def:FPGobjective}
\begin{aligned}
J^{i,N}(u^{i,N};u^{-i,N}) \triangleq & \, E \Bigg[ \int_0^T S_t^{u^N} u^{i,N}_t dt - \lambda \int_0^T (u^{i,N}_{t})^2 dt -\phi \int_0^T (X_t^{u^{i,N}})^2 dt \bigg. \\
& \hspace{22pt} \bigg. + X_T^{u^{i,N}} (P_T - \varrho X_T^{u^{i,N}}) \Bigg],
\end{aligned}
\end{equation}
over admissible rates $u^{i,N} \in \mathcal A$. Observe that $J^{i,N}(u^{i,N};u^{-i,N}) < \infty$ for any set of admissible strategies $u^N \in \mathcal A^N$ because the objective is a linear quadratic functional of the state processes and the controls. As in the single-agent case in~\cite{N-V19} the parameters $\phi > 0$ and $\varrho > 0$ implement, respectively, an additional penalty on the $i$-th agent's running inventory and her terminal position at final time~$T$. For further motivations of the objective in~\eqref{def:FPGobjective} we refer to, e.g.,~\cite{N-V19} in the single-agent case $N=1$ and the references therein.  

Our first goal in this paper is to solve simultaneously for each agent $i \in \{1,\ldots,N\}$ their individual optimal stochastic control problems 
\begin{equation} \label{def:FPGoptimization}
J^{i,N}(u^{i,N};u^{-i,N}) \rightarrow \max_{u^{i,N} \in \mathcal A}.
\end{equation}
This solution will establish a Nash equilibrium for this stochastic differential price impact game in the following usual sense.

\begin{definition} \label{def:Nash}
A set of strategies $\hat{u}^N = (\hat{u}^{1,N},\ldots,\hat{u}^{N,N}) \in \mathcal{A}^N$ is called an open-loop Nash equilibrium if for all $i \in \{1,\ldots, N\}$ and for all admissible strategies $v \in \mathcal{A}$ it holds that
\begin{equation*}
    J^{i,N}(\hat{u}^{i,N};\hat{u}^{-i,N}) \geq J^{i,N}(v;\hat{u}^{-i,N}). 
\end{equation*}
\end{definition}

\begin{remark} \label{rem:model}
When the signal is set to zero, that is, $A\equiv 0$, our $N$-player model reduces to the models that were studied by \citet{Strehle:17} and \citet{SchiedStrehelZhang:17}. The optimal strategy in these models is deterministic while in our setting an optimal strategy will be signal-adaptive. 
\end{remark}
\begin{remark} 
 We remark that the transient impact scaling with $1/N$ in
 \eqref{def:Y} is not essential at this point, but is crucial for the
 scaling limit of the system, as number of agents tends to
 infinity. See the discussion in Section \ref{sec:infinitePlayer}. 
\end{remark}
\begin{remark} \label{rem:terminalPosition}
Observe that agent $i$'s terminal position $X_T^{u^{i,N}}$ in~\eqref{def:FPGobjective} is valued with respect to $P_T$ and not $S^{u^N}_T$. This is to ensure that the functional $u^{i,N} \mapsto J^{i,N}(u^{i,N};u^{-i,N})$ is strictly concave in $u^{i,N} \in \mathcal{A}$; see Lemma~\ref{lem:concave_finite} below. Interestingly, strict concavity is in general not guaranteed if $P_T$ is replaced by $S^{u^N}_T$ due to the arising mixed product $X_T^{u^{i,N}} Y^{u^N}_T$. Note, however, that we implicitly assume a large penalty parameter $\varrho >0$ on agent $i$'s outstanding inventory $X_T^{u^{i,N}}$ in order to virtually enforce a liquidation constraint at terminal time $T>0$ and to stabilize the competitive game between the agents. In this regard, the valuation of the final position in the risky asset is of minor relevance because agent $i$'s terminal inventory will be very close to zero anyways. 
\end{remark}

\subsection{An FBSDE Characterization of the Finite Player Nash Equilibrium} \label{subsec:FBSDEchar}
 
In the spirit of Pontryagin's stochastic maximum principle and along the lines of the corresponding single-agent problem in~\cite{N-V19}, a probabilistic and convex analytic calculus of variations approach can be readily employed here to derive a system of coupled linear FBSDEs. This system characterizes a unique open-loop Nash equilibrium for our multi-agent price impact game. 

\begin{lemma} \label{thm:NASHFBSDE}
A set of controls $(u^{i,N})_{i \in \{1,\ldots, N\}} \subset \mathcal A$ yields the unique Nash equilibrium in the sense of Definition~\ref{def:Nash} if and only if the processes $(X^{u^{i,N}},Y^{u^N},u^{i,N},Z^{u^{i,N}})$, $i = 1,\ldots,N$, satisfy the following coupled linear forward backward SDE system
\begin{equation} \label{eq:NASHFBSDE}
\left\{
\begin{aligned}
    dX^{u^{i,N}}_t = & \, - u^{i,N}_t dt, \quad X^{u^{i,N}}_0 = x^{i,N},\\
    dY^{u^N}_t = & \, -\rho Y_t^{u^N} dt + \frac{\gamma}{N} \sum_{i=1}^{N} u^{i,N}_t dt, \quad Y^{u^N}_0 = y ,\\
    du^{i,N}_t = & \, \frac{dP_t}{2\lambda}  + \frac{\kappa\rho Y^{u^N}_t }{2\lambda} dt -\frac{\gamma\kappa}{2\lambda N} \sum_{j\not =i} u_{t}^{j,N} dt - \frac{\phi X^{u^{i,N}}_t}{\lambda} dt + \frac{\rho Z^{u^{i,N}}_t }{2\lambda} dt + dM^{i,N}_t,  \\
    & \hspace{215pt} u^{i,N}_T = \frac{\varrho}{\lambda} X^{u^{i,N}}_T -
      \frac{\kappa}{2\lambda} Y_T^{u^N} \\ 
    dZ^{u^{i,N}}_t = & \,\rho Z^{u^{i,N}}_t dt + \frac{\gamma\kappa}{N}u^{i,N}_t dt + dN^{i,N}_t, 
    \quad  Z^{u^{i,N}}_T = 0 
\end{aligned}
\right.
\end{equation}
for suitable square integrable martingales $M^{i,N}=(M^{i,N}_t)_{0\leq t \leq T}$ and $N^{i,N}=(N^{i,N}_t)_{0\leq t\leq T}$, $i = 1,\ldots,N$. In particular, the system in~\eqref{eq:NASHFBSDE} has a unique solution. 
\end{lemma} 

The proof of Lemma \ref{thm:NASHFBSDE} is given in Section~\ref{subsec:proofs-finite}. In order to decouple the system in Theorem~\ref{thm:NASHFBSDE} it is very natural to first average over all $N$ FBSDEs in~\eqref{eq:NASHFBSDE} and to introduce an auxiliary aggregated FBSDE system; see also~\citet{DrapeauLuoSchiedXiong:19}. More precisely, we let $(\ol X^{\bar u^N}, \ol Y^{\bar u^N}, \bar u^N, \ol Z^{\bar u^N})$ with $\bar u^N \in \mathcal{A}$ denote the unique solution to the linear FBSDE system
\begin{equation} \label{eq:mean-FBSDE}
\left\{
\begin{aligned}
    d\ol X_t^{\bar u^N} = & \, - \ol u_t^N dt, \quad \ol X_0 = \ol x^N,\\
    d\ol Y_t^{\bar u^N}  = & \, -\rho \ol Y_t^{\bar u^N}  dt + \gamma  \bar u_{t}^N dt, \quad \ol Y_0^{\bar u^N} = y, \\
    d\bar u_t^N = & \,  \frac{dP_t}{2 \lambda}  + \frac{\kappa\rho \ol Y_t^{\bar u^N}}{2\lambda}  dt  -\frac{\gamma\kappa (N-1) \bar u_{t}^N}{2 \lambda N} dt  - \frac{\phi \ol X_{t}^{\bar u^N}}{\lambda} dt + \frac{\rho \ol Z_{t}^{\bar u^N}}{2\lambda} dt + d\ol M_t^N, \\
    & \hspace{215pt} \bar u_T^N = \frac{\varrho}{\lambda} \ol X_T^{\bar u^N} -
    \frac{\kappa}{2\lambda} \ol Y_T^{\bar u^N} \\ 
    d\ol Z_t^{\bar u^N} = & \,\rho \ol Z_t^{\bar u^N} dt + \frac{\gamma\kappa}{N}\bar u_t^N dt
    + d\ol N_t^N, \quad  \ol Z_T^{\bar u^N} = 0 
\end{aligned}
\right.
\end{equation}
for two suitable square integrable martingales $\ol M^N = (\ol M^N_t)_{0 \leq t \leq T} $ and $\ol N^N = (\ol N^N_t)_{0 \leq t \leq T}$ where $\bar x^N \triangleq \frac{1}{N} \sum_{i=1}^{N} x^{i,N}$. Then we obtain the following

\begin{cor} \label{cor:NASHFBSDE}
Let $(\ol X^{\bar u^N}, \ol Y^{\bar u^N}, \bar u^N, \ol Z^{\bar u^N})$, $\bar u^N \in \mathcal{A}$, be the unique solution to the linear FBSDE system in~\eqref{eq:mean-FBSDE}. Moreover, for each $i \in \{1, \ldots, N\}$ let $(X^{u^{i,N}}, u^{i,N}, Z^{u^{i,N}})$  with $u^{i,N} \in \mathcal{A}$ be the unique solution to 
\begin{equation} \label{eq:NASHFBSDE*}
\left\{
\begin{aligned}
    dX^{u^{i,N}}_t = & \, - u^{i,N}_t dt, \quad X^{u^{i,N}}_0 = x^{i,N}\\
    du^{i,N}_t = & \, \frac{1}{2\lambda} (dP_t - \kappa d\ol Y_t^{\bar u^N}) + \frac{\gamma\kappa}{2\lambda N} u_t^{i,N} dt - \frac{\phi X^{u^{i,N}}_t}{\lambda} dt + \frac{\rho Z^{u^{i,N}}_t }{2\lambda} dt + dM^{i,N}_t,  \\
    & \hspace{215pt} u^{i,N}_T = \frac{\varrho}{\lambda}
      X^{u^{i,N}}_T-\frac{\kappa}{2\lambda}  \ol Y_T^{\bar u^N} \\
    dZ^{u^{i,N}}_t = & \,\rho Z^{u^{i,N}}_t dt +
    \frac{\gamma\kappa}{N}u^{i,N}_t dt + dN^{i,N}_t, \quad Z^{u^{i,N}}_T =
    0 
\end{aligned}
\right.
\end{equation}
for suitable square integrable martingales $M^{i,N}=(M^{i,N}_t)_{0\leq t \leq T}$ and $N^{i,N}=(N^{i,N}_t)_{0\leq t\leq T}$. Then it holds that 
\begin{equation} \label{eq-rrr} 
    \begin{aligned}
    \bar u_t^N = & \, \frac{1}{N} \sum_{i=1}^N u^{i,N}_t, & \ol X_t^{\bar u^N} = & \, \frac{1}{N}\sum_{i=1}^N X^{u^{i,N}}_t = \bar x - \int_0^t \bar u_s^N ds, & \ol Y_t^{\bar u^N} = & \, Y^{u^N}_t, \\
    \ol Z_t^{\bar u^N} = & \, \frac{1}{N} \sum_{i=1}^N Z^{i,N}_t, & \ol M_t^N = & \, \frac{1}{N} \sum_{i=1}^N M^{i,N}_t, \quad  \ol N_t^N = \, \frac{1}{N} \sum_{i=1}^N N^{i,N}_t,
    \end{aligned}
\end{equation}
and $(X^{u^{i,N}}, \ol Y^{\bar u^N}, u^{i,N}, Z^{u^{i,N}})$, $i=1,\ldots,N$, satisfy the system in~\eqref{eq:NASHFBSDE}. In particular, the set of controls $(u^{i,N})_{i \in \{1, \ldots, N\}} \subset \mathcal A$ yields the unique Nash equilibrium in the sense of Definition~\ref{def:Nash}.
\end{cor}

\begin{proof}
Summing up all four equations in~\eqref{eq:NASHFBSDE} over $i$ and multiplying them with $1/N$ gives the FBSDE system in~\eqref{eq:mean-FBSDE}. Next, replacing in~\eqref{eq:NASHFBSDE} in the second BSDE for $u^{i,N}$ the term $\sum_{j \neq i} u^{j,N}$ by $N \ol u^N - u^{i,N}$ and noting that $\ol Y^{\bar u^N} = Y^{u^N}$ yields the claim.  
\end{proof}

Observe that for each agent $i \in \{1,\ldots,N\}$ the FBSDEs
in~\eqref{eq:NASHFBSDE*} are fully decoupled and that the jointly
created but autonomous transient price distortion $\ol Y^{\bar u^N}$ computed
from the FBSDE in~\eqref{eq:mean-FBSDE} feeds into the system
in~\eqref{eq:NASHFBSDE*} through adding on to the dynamics of the
unaffected price process via $P - \kappa \ol Y^{\bar u^N}$. In fact, the
process~$\ol Y^{\bar u^N}$ can be viewed as an endogenous signal, which is observed by the traders, in addition to the exogenous signal $A$.

\begin{remark} \label{rem:SingelAgentFBSDE}
In the case $N=1$ the FBSDE system in~\eqref{eq:NASHFBSDE} or, equivalently, in~\eqref{eq:mean-FBSDE} or, equivalently, in~\eqref{eq:NASHFBSDE*}, corresponds to the one derived in~\cite[Lemma 5.2]{N-V19} for the single-agent problem.
\end{remark}

\subsection{Solving the Finite Player Game} \label{sec:finitePlayerSol}

To compute explicitly the unique Nash equilibrium for our $N$-player stochastic differential game we need to solve successively two linear FBSDE systems. First, the aggregated system in~\eqref{eq:mean-FBSDE} and then, separately for each agent $i \in \{1,\ldots,N\}$ the system in~\eqref{eq:NASHFBSDE*}. Both can be achieved in terms of matrix exponentials. Specifically, let 
\begin{equation} \label{def:matrixExpEbar}
\ol Q(t) \triangleq \exp( \ol F^N \cdot t) = (\ol Q_{ij}(t))_{1 \leq i,j, \leq 4} \in \mathbb{R}^{4 \times 4},
\end{equation}
denote the matrix exponential of the matrix
\be \label{def:Fbar} 
\ol F^N \triangleq \begin{pmatrix}
    0 & 0 & -1 & 0 \\ 0 & -\rho & \gamma & 0 \\ - \frac{\phi }{\lambda}  &  \frac{\kappa\rho }{2 \lambda} & -\frac{\kappa \gamma (N-1)}{2 \lambda N}  & \frac{\rho}{2 \lambda}\\ 0 & 0 & \frac{\kappa\gamma}{N} & \rho
    \end{pmatrix} \in \mathbb{R}^{4 \times 4}. 
\ee
Moreover, define $\ol G(t) = ( \ol G_i(t))_{1 \leq i \leq 4} \in \mathbb{R}^{4}$ as   
\begin{equation} \label{def:Gbar}
\ol G(t) \triangleq  \left(\frac{\varrho}{\lambda}, - \frac{\kappa}{2\lambda}, -1,0 \right)  \ol Q(t) \qquad (t\geq 0), 
\end{equation}
and $\ol H(t) = ( \ol H_i(t))_{1 \leq i \leq 4} \in \mathbb{R}^{4}$ as   
\begin{equation} \label{def:Hbar}
\ol H(t) \triangleq \left(0,0,0,1 \right)  \ol Q(t) \qquad (t\geq 0),
\end{equation}
and let 
\begin{equation}
\begin{aligned} \label{def:vsbar}
\bar v_0(t) \triangleq \left( 1- \frac{\ol G_4(t)}{\ol G_3(t)} \frac{\ol H_{3}(t)}{\ol H_{4}(t)} \right)^{-1},  \qquad
&& \bar v_1(t) \triangleq & \;
          \frac{\ol G_4(t)}{\ol G_3(t)}\frac{\ol H_{1}(t)}{\ol H_{4}(t)}
          -\frac{\ol G_1(t)}{\ol G_3(t)}, \\
\bar v_2(t) \triangleq 
          \frac{\ol G_4(t)}{\ol G_3(t)}\frac{\ol H_{2}(t)}{\ol H_{4}(t)}
         -\frac{\ol G_2(t)}{\ol G_3(t)},  \qquad
&& \bar v_3(t) \triangleq & \; \frac{\ol G_4(t)}{\ol G_3(t)},
\end{aligned}
\end{equation} 
for all $t \in [0,\infty)$. For simplicity we make the following assumption. 

\begin{assumption} \label{assump:1}
We assume that the set of parameters $\xi \triangleq (\lambda,\gamma, \kappa, \rho, \varrho, \phi, T) \in \re^7_+$ are chosen such that
\be \label{g-s-cond-1} 
\inf_{t \in [0,T]} \left| \ol G_3(t) \ol H_{4}(t) - \ol G_4(t) \ol H_{3}(t) \right|>0 
 \ee
and $ \inf_{t \in [0,T]} | \ol G_3(t)| > 0$, $\inf_{t \in [0,T]}| \ol H_4(t) | >0$; as well as that the eigenvalues $\bar\nu_1, \bar\nu_2, \bar\nu_3, \bar\nu_4$ of the matrix $\ol F^N$ in~\eqref{def:Fbar} are real-valued and distinct.  
\end{assumption}

Then denoting by $E_{t}$ the conditional expectation with respect to $\mathcal F_t$ for all $t\in [0,T]$, we obtain the following feedback solution for the FBSDE system in~\eqref{eq:mean-FBSDE}.

\begin{proposition} \label{prop:sol-mean-FBSDE} 
Under Assumption~\ref{assump:1} the unique solution $\bar u^N \in \mathcal A$ satisfying~\eqref{eq:mean-FBSDE} is given in a linear feedback form via
\be \label{eq:opt_ubar} 
\begin{aligned}
\bar {u}_{t}^N = & \; \bar v_0(T-t) \Bigg( \bar v_1(T-t) \ol X_{t}^{\bar u^N} + \bar v_2(T-t) \ol Y_{t}^{\bar u^N} \Bigg. \\
& + \frac{1}{2\lam} \Bigg. \bigg( \bar v_3(T-t)E_t\left[
   \int_t^T  \frac{\ol H_{3}(T-s)}{\ol H_{4}(T-t)}dA_s \right] - E_{t}\left[
   \int_{t}^T\frac{\ol G_3(T-s)}{\ol G_{3}(T-t)}dA_{s} \right] \bigg) \Bigg),
\end{aligned}
\ee
for all $t \in (0,T)$.
\end{proposition}  

The proof of Proposition \ref{prop:sol-mean-FBSDE} is given in Section~\ref{subsec:proofs-finite}.

\begin{remark}
In the case $N=1$ Proposition~\ref{prop:sol-mean-FBSDE} retrieves the single-agent optimal strategy from~\cite[Theorem 3.2]{N-V19}.
\end{remark}

\begin{remark} 
  We refer to Section \ref{sec-mat-fin} for the computation of the
  matrix exponential $\ol Q(t)$ in~\eqref{def:matrixExpEbar} via
  diagonalization, as well as the functions $\ol G(t), \ol H(t)$
  in~\eqref{def:Gbar}, \eqref{def:Hbar}. Symbolic computation of the
  eigenvalues $\bar\nu_1, \bar\nu_2, \bar\nu_3, \bar\nu_4$ of
  $\ol F^N$ in~\eqref{def:Fbar} is very cumbersome as it involves
  roots of a quartic equation. We hence omit it. Instead, we require
  the additional property of the eigenvalues in Assumption~\ref{assump:1}
  to guarantee that $\ol G_i(t)$ and $\ol H_i(t)$ which are explicitly
  given in~\eqref{def:Gbar1}--\eqref{def:Hbar4} are well-defined and
  bounded for all $i \in \{1,\ldots,4\}$. Note, however, that
  $\det(\ol F^N) = \rho^2 \phi/\lambda > 0$ so that all eigenvalues are
  different from zero. Moreover, observe that for a given set of
  parameters $\xi$ the conditions in Assumption~\ref{assump:1} can
  also be easily verified using the expressions derived in
  Section~\ref{sec-mat-fin}.
\end{remark}

Next having at hand the solution $(\ol X^{\bar u^N}, \ol Y^{\bar u^N}, \ol u^N)$ from Proposition~\ref{prop:sol-mean-FBSDE} for the aggregated FBSDE system~\eqref{eq:mean-FBSDE} we can insert $\ol Y^{\bar u^N}$ into~\eqref{eq:NASHFBSDE*} and solve this linear FBSDE system in $(X^{u^{i,N}},u^{i,N},Z^{u^{i,N}})$ for each $i=1,\ldots,N$ separately. More precisely, let 
\begin{equation} \label{def:matrixExpE}
Q(t) \triangleq \exp(  F^N \cdot t) = ( Q_{ij}(t))_{1 \leq i,j, \leq 3}  \in \mathbb{R}^{3 \times 3},
\end{equation}
denote the matrix exponential of the matrix
\be \label{def:F} 
F^N \triangleq \begin{pmatrix}
    0 & -1 & 0 \\ -\frac{\phi}{\lambda} & \frac{\kappa \gamma}{2 \lambda N} & \frac{\rho}{2\lambda} \\  
    0 & \frac{\kappa\gamma}{N} & \rho
    \end{pmatrix} \in \mathbb{R}^{3 \times 3}. 
\ee
In addition, define $G(t) = ( G_i(t))_{1 \leq i \leq 3} \in \mathbb{R}^{3}$ as    
\begin{equation} \label{def:G}
G(t) \triangleq  \left(\frac{\varrho}{\lambda}, -1, 0 \right) Q(t) \qquad (t\geq
0),
\end{equation}
and $H(t) = (H_i(t))_{1\leq i \leq 3} \in \mathbb{R}^{3}$ as
\begin{equation} \label{def:H}
H(t) \triangleq  \left(0, 0, 1 \right) Q(t)
\qquad (t\geq 0).
\end{equation}
Lastly, let 
\begin{equation}
\begin{aligned} \label{def:vs}
v_0(t) \triangleq \left(1 - \frac{G_3(t)}{G_2(t)} \frac{H_{2}(t)}{H_{3}(t)} \right)^{-1}, \;
v_1(t) \triangleq 
          \frac{G_3(t)}{G_2(t)}\frac{H_{1}(t)}{ H_{3}(t)}
          -\frac{G_1(t)}{G_2(t)}, 
\; v_2(t) \triangleq \frac{G_3(t)}{G_2(t)},
\end{aligned}
\end{equation} 
for all $t \in [0,\infty)$ and make the following assumption.

\begin{assumption} \label{assump:2}
We assume that the set of parameters $\xi = (\lambda,\gamma, \kappa, \rho, \varrho, \phi, T) \in \re^7_+$ are chosen such that 
\be \label{g-s-cond-2} 
  \inf_{t\in [0,T]}|G_2(t) H_{3}(t) - G_3(t) H_{2}(t)|>0, 
 \ee 
and that $ \inf_{t\in [0,T]} |G_2(t)| > 0$, $ \inf_{t\in [0,T]}|H_3(t)| > 0$.
\end{assumption}

Our main result of Section~\ref{sec:finitePlayer} can now be summarized as follows:

\begin{theorem} \label{thm:main-finite} 
Under Assumptions~\ref{assump:1} and~\ref{assump:2} let $\ol Y^{\bar u^N}$ denote the unique solution of~\eqref{eq:mean-FBSDE}. Then for each $i \in \{1,\ldots,N\}$ the unique solution $\hat{u}^{i,N} \in \mathcal A$ satisfying~\eqref{eq:NASHFBSDE*} is given in a linear feedback form via
\begin{equation} \label{eq:opt_ui}
\begin{aligned}
\hspace{-1pt} \hat{u}^{i,N}_{t} = & \; v_0(T-t) \Bigg( v_1(T-t) X^{\hat{u}^{i,N}}_{t} + \frac{v_2(T-t)}{2\lam} \Bigg. E_t\left[
  \int_t^T  \frac{H_{2}(T-s)}{H_{3}(T-t)} \left(dA_s - \kappa d\ol Y_s^{\bar u^N} \right) \right]  \\
  & \hspace{60pt}  - \frac{1}{2\lam} E_{t}\left[
  \int_{t}^T\frac{G_2(T-s)}{G_{2}(T-t)} \left(dA_s - \kappa d\ol Y_s^{\bar u^N}
  \right) -\frac{\kappa}{G_2(T-t)} \ol Y_T^{\bar u^N}\right] \Bigg),
\end{aligned}
\end{equation}
for all $t \in (0,T)$. In particular, the set of controls $(\hat{u}^{i,N})_{i \in \{1,\ldots,N\}} \subset \mathcal A$ in~\eqref{eq:opt_ui} provides the unique Nash equilibrium strategies in the sense of Definition~\ref{def:Nash}.
\end{theorem}  
The proof of Theorem \ref{thm:main-finite} is given in Section~\ref{subsec:proofs-finite}.

\begin{remark} 
We refer to Section \ref{sec-mat-fin} for the computation of the matrix exponential $Q(t)$ in~\eqref{def:matrixExpE} via diagonalization, as well as the functions $G(t), H(t)$ in~\eqref{def:G}, \eqref{def:H}. In contrast to the matrix $\ol F^N$ in~\eqref{def:Fbar}, the eigenvalues $\nu_1, \nu_2, \nu_3 \in \mathbb{R}$ of the matrix $F^N$ in~\eqref{def:F} can be easily computed explicitly. Also note that for a given set of parameters $\xi$ the conditions in Assumption~\ref{assump:2} can be readily checked with the expressions computed in Section~\ref{sec-mat-fin}.
\end{remark} 

\begin{remark}  \label{rem-sig} 
The optimal Nash equilibrium strategies $(\hat{u}^{i,N})_{i \in
  \{1,\ldots,N\}}$ in Theorem~\ref{thm:main-finite} reveal that in
equilibrium, the aggregated transient price impact $\ol Y^{\bar u^N}$ directly
feeds into the signal process~$A$ through the term $A - \kappa \ol
Y^{\bar u^N}$. In other words, each agent is adjusting the trading speed
according to the common exogenous price signal $A$, as well as to the
common price impact~$\ol Y^{\bar u^N}$, which acts as an endogenous signal. The
signal $\ol Y^{\bar u^N}$ can be interpreted as the aggregated order-flow of all agents with an exponential weighting (see \eqref{def:Y}). 
\end{remark} 

\begin{remark} 
Observe that in contrast to the single agent solution in~\cite[Theorem 3.2]{N-V19}, each agent $i$'s individual transient price distortion $Y^{u^{i,N}}$ (cf.~\eqref{def:Yui}) is not a state variable anymore which is taken into account in the feedback dynamics in~\eqref{eq:opt_ui} but it gets absorbed in the autonomous process $\ol Y^{\bar u^N}$.
\end{remark}

\section{An Infinite Player Mean Field Game} \label{sec:infinitePlayer}

Our second goal in this paper is to introduce and study the limiting mean field game of the finite player price impact game from Section~\ref{sec:finitePlayer} when the number of agents $N$ tends to infinity, and the price impact of a single agent becomes negligible.

To this end, suppose there are infinitely many agents indexed by $i \in \mathbb{N}$. As in~\eqref{def:Xi} the inventory of each agent $i \in \mathbb{N}$ is described by
\begin{equation} \label{def:Xi_infMFG} 
X^{v^i}_t \triangleq
  x^{i}-\int_0^t v^{i}_s ds \qquad (0 \leq t \leq T)
\end{equation}
with initial position $x^i \in \mathbb{R}$ and selling rate $v^i \in \mathcal{A}$.

\begin{assumption} \label{ass:initPosLimit}
We assume that the limit of the averages of all initial positions $(x^i)_{i \in \mathbb{N}} \subset \mathbb{R}$ exists and denote this limit as  
\begin{equation} \label{eq:initPosLimit}
    \tilde x = \lim_{N \rightarrow \infty} \frac{1}{N} \sum_{i=1}^N x^i \in \mathbb{R}.
\end{equation}
\end{assumption}

Moreover, we introduce a stochastic process $\nu \in \mathcal{A}$ which represents the limiting aggregated averaged trading (selling) speed of all agents. For such a given net trading flow $\nu$, the risky asset's execution price at which each agent $i$ is executing her trades is now prescribed as 
\begin{equation} \label{def:S_infMFG} 
S^{\nu}_{t} \triangleq P_{t} - \kappa Y^{\nu}_t \qquad (0 \leq t \leq T), 
\end{equation}
with unaffected price process $P \in \mathcal H^2$ and transient price distortion
\begin{equation} \label{def:Y_infMFG} 
Y^{\nu}_t \triangleq e^{-\rho t} y
  + \gamma \int_0^t e^{-\rho (t-s)} \nu_s ds \qquad (0 \leq t \leq T),
\end{equation}
for some $y > 0$. Note that in contrast to~\eqref{def:S}
and~\eqref{def:Y}, here agent $i$'s impact on the visible price process $S^\nu_t$ is neglected. 

The trader's objective functional from~\eqref{def:FPGobjective} modifies to
\begin{equation} \label{def:objective_infMFG}
J^{i,\infty}(v^{i};\nu) \triangleq E \Bigg[ \int_0^T S^{\nu}_t v^{i}_t dt - \lambda \int_0^T (v^{i}_{t})^2 dt - \phi \int_0^T
(X_t^{v^i})^2 dt + X_T^{v^i} (P_T - \varrho X_T^{v^i}) \Bigg].
\end{equation}
Observe that $J^{i,\infty}(v^{i};\nu) < \infty$ for any admissible $v^{i}, \nu  \in \mathcal A$ and any agent $i \in \mathbb{N}$.

The paradigm of the infinite player mean field game is now to first solve for each agent $i \in \mathbb{N}$ and for a given net trading flow $\nu \in \mathcal{A}$ the optimal stochastic control problems 
\begin{equation} \label{def:optimization_infMFG}
J^{i,\infty}(v^{i};\nu) \rightarrow \max_{v^{i} \in \mathcal A} \qquad (i \in \mathbb{N}), 
\end{equation}
and then to determine $\nu$ endogenously in the following sense; see, e.g., also~\cite{CasgrainJaimungal:18,CasgrainJaimungal:20}.

\begin{definition} \label{def:Nash_infMFG}
A collection of strategies $(\hat{v}^i)_{i \in \mathbb{N}} \subset \mathcal{A}$ is called a Nash equilibrium to the infinite player mean field game if for each agent $i\in\mathbb{N}$ the strategy~$\hat{v}^i$ solves the optimization problem in~\eqref{def:optimization_infMFG} with $\hat{\nu} \in \mathcal{A}$ satisfying
\begin{equation} \label{eq:Nash_infMFG}
    \hat{\nu}_t = \lim_{N \rightarrow \infty} \frac{1}{N} \sum_{i=1}^N \hat{v}^i_t \qquad d\mathbb{P}\otimes dt\text{-a.e. on } \Omega \times [0,T].
\end{equation}
\end{definition}

Put differently, since the optimal solution $\hat{v}^i$ from~\eqref{def:optimization_infMFG} depends on the net trading flow $\nu$, the aim for solving the infinite player mean field game is to determine $\hat{\nu}$ such that the fixed point equation in~\eqref{eq:Nash_infMFG} is satisfied.
\begin{remark} 
Note that Definition \ref{def:Nash_infMFG} extends the classical definition of mean field games by \citet{LasryLions} and \citet{CarmonaDelarue:18:1} (see Chapter 3.2.1), as it includes common noise and therefore the mean field trading speed $\hat \nu$ is random. Our setting is closer to the definition of the mean field game in \citet{car-16}, which permits common noise as well. However, unlike \cite{car-16}, in our case agents start with different initial conditions, therefore the fixed point condition in \cite{car-16} is replaced by \eqref{eq:Nash_infMFG}, as the agents' laws conditioned on the common noise are not identical. One can observe that under the assumption that all agents are identical, the two definitions coincide. We remark that Definition \ref{def:Nash_infMFG} is compatible with Assumption 3.1 of \cite{CasgrainJaimungal:20}, only here we prove that our solution satisfies \eqref{eq:Nash_infMFG} instead of assuming that, as was done in \cite{CasgrainJaimungal:20}.
 \end{remark} 
\begin{remark} \label{rem-l2-con} 
Our convergence results will give us precise information on the convergence rate of the limit in \eqref{eq:Nash_infMFG} under the $L^2$ norm. Specifically, in \eqref{consis-l2} we prove that
$$
 E\left[\sup_{t\in [0,T]} \left(\hat \nu_{t}-\frac{1}{N}\sum_{i=1}^N \hat v_t^i\right)^{2}\right]= O\left( C(N)^2\right), 
$$
where, using the notation of \eqref{eq:initPosLimit}, $C(N)$ is given by, 
\be \label{c-n}
C(N)=  \left| \tilde x -   \frac{1}{N} \sum_{i=1}^N x^i \right|.  
\ee
\end{remark} 

\begin{remark} \label{rem:MFG}
Our mean field model extends the models that were studied by Casgrain
and Jaimungal \cite{CasgrainJaimungal:18, CasgrainJaimungal:20} in
the sense that it incorporates a transient price impact into the mean field interaction. We also extend the model of \citet{fu2020}, which does incorporate a transient price impact, but does not include the exogenous price predictive signal.  
\end{remark}



\subsection{An FBSDE Characterization of the Infinite Player Nash
  Equilibrium} \label{sec:infinitePlayerFBSDEs}

Along the lines of Section~\ref{sec:finitePlayer}, a purely probabilistic approach can be implemented to solve the infinite player mean field game formulated in Definition~\ref{def:Nash_infMFG}. That is, due to the concave structure of our problem we can directly derive the characterizing system of equations again via calculus of variations arguments \`a la Pontryagin's stochastic maximum principle. The unique infinite player Nash equilibrium for the mean field game problem in Definition~\ref{def:Nash_infMFG} is then obtained by solving a coupled infinite system of FBSDEs.

\begin{lemma} \label{thm:MeanFieldFBSDE_inf}
  A collection of controls $(\hat{v}^i)_{i \in \mathbb{N}} \subset \mathcal{A}$ uniquely solves the infinite player mean field game in the sense of Definition~\ref{def:Nash_infMFG} if
  and only if for each $i\in\mathbb{N}$ the processes $(X^{\hat{v}^i},\hat{v}^i, Y^{\hat{\nu}},\hat{\nu})$
  satisfy the following coupled linear forward backward SDE
  systems
  \begin{equation} \label{eq:MeanFieldFBSDE_inf} \left\{
      \begin{aligned}
        dX^{\hat v^i} _t = & \, - \hat v^i_t dt, \quad  X^{\hat v^i}_0 = x^i \\
        dY^{\hat\nu}_t  = & \, -\rho Y^{\hat\nu}_t  dt + \gamma  \hat\nu_{t} dt, \quad Y^{\hat\nu}_0 = y \\
        d \hat v^i_t = & \, \frac{dP_t}{2\lambda} + \frac{\kappa\rho }{2
          \lambda} Y^{\hat\nu}_t dt - \frac{\kappa \gamma}{2\lambda}
        \hat\nu_t dt - \frac{\phi }{\lambda} X^{\hat v^i}_{t}dt + d
        L^i_t, \; \hat v^i_T = \frac{\varrho}{\lambda} X^{\hat
            v^i}_T - \frac{\kappa}{2\lambda} Y^{\hat\nu}_T 
      \end{aligned}
    \right.
  \end{equation}
  for a collection of suitable square-integrable martingales
  $L^i=(L^i_t)_{0 \leq t \leq T}$ where
  \begin{equation} \label{MeanFieldFBSDE_inf_consist}
    \hat\nu_t = \lim_{N \rightarrow \infty} \frac{1}{N} \sum_{i=1}^N \hat v^i_t, \qquad d\mathbb{P}\otimes dt\text{-a.e. on } \Omega \times [0,T].
\end{equation}
\end{lemma}
The proof of Lemma \ref{thm:MeanFieldFBSDE_inf} is given in
Section~\ref{subsec:proofs-infinite}.

Similar to the finite player game in Section~\ref{sec:finitePlayer}, in order to decouple and solve the infinite FBSDE system in~\eqref{eq:MeanFieldFBSDE_inf}, it is very natural to introduce again an aggregated auxiliary FBSDE system similar to~\eqref{eq:mean-FBSDE}. Specifically, let $(\tilde X^{\tilde \nu},\tilde Y^{\tilde \nu},\tilde \nu)$ with $\tilde \nu \in \mathcal{A}$ denote the unique solution to the linear FBSDE system
\begin{equation} \label{eq:MeanField_infAggregated} 
\left\{
\begin{aligned} 
d \tilde X_t^{\tilde \nu} = & \, -  \tilde\nu_t dt, \quad  \tilde X_0^{\tilde \nu} =  \tilde x, \\
d \tilde Y_t^{\tilde \nu}  = & \, -\rho \tilde Y_t^{\tilde \nu}  dt + \gamma \tilde\nu_{t} dt, \quad \tilde Y_0 = y, \\
d \tilde\nu_ t = & \, \frac{dP_t}{2\lambda} + \frac{\kappa\rho }{2
  \lambda}  \tilde Y_t^{\tilde \nu} dt - \frac{\kappa \gamma}{2\lambda}
\tilde\nu_{t} dt - \frac{\phi }{\lambda}  \tilde X_{t}^{\tilde \nu} dt + d\tilde
L_t, \; \tilde \nu_T = \frac{\varrho }{\lambda}  \tilde X_T^{\tilde \nu} -
\frac{\kappa}{2\lambda} \tilde Y_T^{\tilde \nu}  
\end{aligned}
\right.
\end{equation}
with $\tilde x \in \mathbb{R}$ from~\eqref{eq:initPosLimit} and a suitable square integrable martingale $\tilde L=(\tilde L_t)_{0 \leq t \leq T}$. Then we obtain following corollary. 

\begin{cor} \label{cor:MeanFieldFBSDE_inf}
  Let $(\tilde X^{\tilde \nu},\tilde Y^{\tilde \nu},\tilde\nu)$, $\tilde\nu \in \mathcal{A}$, be the unique solution to the linear FBSDE system in~\eqref{eq:MeanField_infAggregated}. Moreover, for each $i \in \mathbb{N}$ let $(X^{\hat{v}^i}, \hat{v}^i)$ with $\hat{v}^i \in \mathcal{A}$ be the unique solution to the linear FBSDE 
\begin{equation} \label{eq:MeanFieldGameFBSDE_inf} 
\left\{
\begin{aligned}
dX^{\hat v^i} _t = & \, - \hat v^i_t dt, \quad  X_0^{\hat v^i} = x^i \\
d \hat v^i_t = & \, \frac{1}{2\lambda} \left( dP_t - \kappa d \tilde Y_t^{\tilde \nu} \right)  - \frac{\phi}{\lambda} X^{\hat v^i}_{t}dt + d
L^i_t, \quad \hat v^i_T = \frac{\varrho }{\lambda} X^{\hat v^i}_T
- \frac{\kappa}{2\lambda} \tilde Y_T^{\tilde \nu} 
\end{aligned}
\right.
 \end{equation}
for a collection of suitable square-integrable martingales $L^i=(L^i_t)_{0 \leq t \leq T}$. Then it holds that for all $i \in \mathbb{N}$ the quadruples $(X^{\hat{v}^i},\hat{v}^i, \tilde Y^{\tilde \nu},\tilde\nu)$ satisfy the systems in~\eqref{eq:MeanFieldFBSDE_inf} with the consistency condition~\eqref{MeanFieldFBSDE_inf_consist} (with $\tilde \nu $ in the role of $\hat \nu$). In particular, $(\hat v^i)_{i \in \mathbb{N}}$ are the unique infinite player Nash equilibrium strategies in the sense of Definition~\ref{def:Nash_infMFG} and
\begin{equation}
    \tilde \nu_t = \lim_{N \rightarrow \infty} \frac{1}{N} \sum_{i=1}^N \hat v^i_t \qquad d\mathbb{P}\otimes dt\text{-a.e. on } \Omega \times [0,T].
\end{equation}
\end{cor}
The proof of Corollary  \ref{cor:MeanFieldFBSDE_inf} is given in Section~\ref{subsec:proofs-infinite}.

\begin{remark} \label{rem-mg-conc}
In view of Corollary~\ref{cor:MeanFieldFBSDE_inf} we refer to the processes $\tilde X^{\tilde \nu}$, $\tilde\nu$, and $\tilde Y^{\tilde \nu}$ satisfying~\eqref{eq:MeanField_infAggregated} as the mean field inventory and trading rate, as well as the mean field transient price impact. Observe that for each agent $i \in \mathbb{N}$ the FBSDEs in~\eqref{eq:MeanFieldGameFBSDE_inf} are now fully decoupled and that the mean field transient price impact $\tilde Y^{\tilde \nu}$ feeds into the system in~\eqref{eq:MeanFieldGameFBSDE_inf} by adding on to the dynamics of the unaffected price process via $P - \kappa \tilde Y^{\tilde \nu}$. Unlike the $N$-player game (see Remark \ref{rem-sig}), here the transient price impact $\tilde Y^{\tilde \nu}$ can be regarded as an exogenous signal, as the contribution of each agent to $\tilde Y^{\tilde \nu}$ is neglected. 
\end{remark} 
\begin{remark} 
Note that the mean field's three-dimensional FBSDE system in~\eqref{eq:MeanField_infAggregated} is considerably simpler than the finite player's aggregated four-dimensional FBSDE system in~\eqref{eq:mean-FBSDE}. This is due to the fact that in the infinite agent setup each agent $i$'s individual price impact on the execution price via the transient price distortion in~\eqref{def:Y_infMFG} is neglected, in contrast to the finite-agent game's deviation process in~\eqref{def:Y}. In other words, passing to the infinite player mean field limit simplifies the finite player game from Section~\ref{sec:finitePlayer} because the controlled state variable $\ol Y^{\bar u^N}$ therein turns into an uncontrolled factor process $\tilde{Y}^{\tilde \nu}$ in the limit.
\end{remark}

\begin{remark} \label{rem:infMFFBSDE}
 The FBSDE system in~\eqref{eq:MeanField_infAggregated} describing the mean field inventory $\tilde X^{\tilde \nu}$ and mean field trading rate $\tilde\nu$ can be considered as a generalization of the deterministic ordinary differential equation derived in~\citet{CardaliaguetLehalle:18}, equation (17). The latter describes the mean field inventory and trading rate in a framework without signal and transient price impact (i.e., $A \equiv 0$ and $\rho = y = 0$). In contrast, in our more general setup, the mean field trading rate must be stochastic and adapted due to the presence of the signal.    
 \end{remark}


\subsection{Solving the Infinite Player Mean Field Game} \label{sec:infinitePlayerSol}

Akin to the finite player game in Section~\ref{sec:finitePlayer} in order to compute the solution to the infinite player mean field game we again have to solve successively two linear systems. First, we solve the mean field FBSDE system in~\eqref{eq:MeanField_infAggregated} and then for each agent $i\in\mathbb{N}$ the FBSDE in~\eqref{eq:MeanFieldGameFBSDE_inf}. Regarding the former let 
\begin{equation} \label{def:matrixExpBtilde}
\tilde R(t) \triangleq \exp( \tilde B \cdot t) = (\tilde{R}_{ij}(t))_{1 \leq i,j, \leq 3} \in \mathbb{R}^{3 \times 3},
\end{equation}
denote the matrix exponential of the matrix
\be \label{def:Btilde} 
\tilde B \triangleq \begin{pmatrix}
    0 & 0 & -1 \\ 0 & -\rho & \gamma \\ -\frac{\phi}{\lambda} & \frac{\kappa\rho}{2 \lambda} & -\frac{\kappa \gamma}{2 \lambda}
    \end{pmatrix} \in \mathbb{R}^{3 \times 3}. 
\ee
Moreover, define $\tilde K(t) = ( \tilde K_i(t))_{1 \leq i \leq 3} \in \mathbb{R}^{3}$ as   
\begin{equation} \label{def:Ktilde}
\tilde K(t) \triangleq  \left( \frac{\varrho}{\lambda},
    -\frac{\kappa}{2\lambda} , -1 \right)
\tilde R(t) \qquad (t\geq 0)
\end{equation}
and let 
\begin{equation}
\begin{aligned} \label{def:wstilde}
\tilde w_1(t) \triangleq - \frac{\tilde K_1(t)}{\tilde K_3(t)}, \qquad \tilde w_2(t) \triangleq & \; -\frac{\tilde K_2(t)}{\tilde K_3(t)},
\end{aligned}
\end{equation} 
for all $t \in [0,\infty)$. 

We obtain the feedback form solution to the mean field trading speed. 
\begin{proposition} \label{prop:sol-MeanField_infAggregated} 
Assume that $\inf_{t \in [0,T]}|\tilde K_3(t)| > 0$. Then, the unique solution $\tilde\nu \in \mathcal{A}$ satisfying the FBSDE system in~\eqref{eq:MeanField_infAggregated} is given in linear feedback form via
\be \label{eq:opt_ubar-MeanField_infAggregated} 
\begin{aligned}
\tilde\nu_{t} = \tilde w_1(T-t) \tilde X_{t}^{\tilde \nu} + \tilde w_2(T-t) \tilde Y_{t}^{\tilde \nu} - \frac{1}{2\lam}
   E_t\left[ \int_t^T  \frac{\tilde K_{3}(T-s)}{\tilde K_{3}(T-t)} dA_s \right],
\end{aligned}
\ee
for all $t \in (0,T)$.
\end{proposition}  

The proof of Proposition \ref{prop:sol-MeanField_infAggregated} is given in Section~\ref{subsec:proofs-infinite}. 

\begin{remark} 
We refer to Section \ref{sec-mat-infin} for the computation of the matrix exponential $\tilde R(t)$ in~\eqref{def:matrixExpBtilde} via diagonalization, as well as the function $\tilde K(t)$ in~\eqref{def:Ktilde}. Note that similar to the matrix $F^N$ in~\eqref{def:F}, the eigenvalues $\tilde\nu_1, \tilde\nu_2, \tilde\nu_3 \in \mathbb{R}$ of the matrix $\tilde B$ in~\eqref{def:Btilde} can be easily computed explicitly. 
\end{remark}

\begin{remark}
Let us mention that the mean field strategy $\tilde \nu$ from Proposition~\ref{prop:sol-MeanField_infAggregated} does not coincide with the single-agent trading strategy with transient price impact from~\citet{N-V19}. Indeed, the FBSDE system in~\eqref{eq:MeanField_infAggregated} describing the latter is only three-dimensional, whereas the single agent's FBSDE system in~\cite[Lemma 5.2]{N-V19} is four-dimensional.
\end{remark}

Next, given the solution $(\tilde X^{\tilde \nu}, \tilde Y^{\tilde \nu}, \tilde\nu)$ from Proposition~\ref{prop:sol-MeanField_infAggregated} for system~\eqref{eq:MeanField_infAggregated} we can insert the process $\tilde Y^{\tilde \nu}$ into~\eqref{eq:MeanFieldGameFBSDE_inf} and solve for each agent $i \in \mathbb{N}$ the resulting linear FBSDE in $(X^{v^i},v^i)$. More precisely, introducing
\begin{equation} \label{def:R}
R(t) \triangleq \sqrt{\phi/\lambda} \cosh ( \sqrt{\phi/\lambda} \, t) + \varrho/\lambda \sinh ( \sqrt{\phi/\lambda} \, t),
\end{equation}
for all $t \in [0,\infty)$ we obtain our main result of Section~\ref{sec:infinitePlayer}: 

\begin{theorem} \label{thm:main-infMFG} 
Let $\tilde Y^{\tilde \nu}$ denote the unique solution satisfying~\eqref{eq:MeanField_infAggregated}. Then the unique solution $(\hat{v}^i)_{i \in \mathbb{N}} \subset \mathcal A$ to the infinite player mean field game in the sense of Definition~\ref{def:Nash_infMFG} is given in linear feedback form via
\be \label{eq:opt_uiInf} 
\begin{aligned}
\hat{v}^{i}_{t} = \frac{R'(T-t)}{R(T-t)} X^{\hat{v}^i}_{t} - \frac{1}{2\lam} E_{t}\left[
  \int_{t}^T\frac{R(T-s)}{R(T-t)} \left(dA_s - \kappa d\tilde Y_s^{\tilde \nu}
  \right) + \sqrt{\frac{\phi}{\lambda}}
    \frac{\kappa}{R(T-t)} \tilde Y_T^{\tilde \nu}  \right],
\end{aligned}
\ee
for all $t \in (0,T)$.
\end{theorem}  

The proof of Theorem \ref{thm:main-infMFG} is given in Section~\ref{subsec:proofs-infinite}.

\begin{remark} \label{rem:Belak}
Remarkably, given $\tilde Y^{\tilde \nu}$, the optimal solution~$(\hat{v}^i)_{i \in \mathbb{N}}$ in~\eqref{eq:opt_uiInf} of the mean field game in fact coincides with the optimal solution of the single-agent optimal signal-adaptive liquidation problem with temporary price impact and an exogenous signal given by $A - \kappa \tilde Y^{\tilde \nu}$ provided in~\citet{BMO:19}. The only minor difference is the appearance of the term involving $\tilde{Y}_T^{\tilde \nu}$. But this merely stems from the fact that the value of the terminal inventory $X^{v^i}_T$ in~\eqref{def:objective_infMFG} is measured in terms of $P_T$ and not $S^{\nu}_T$. The reason for this particular choice is Remark~\ref{rem:terminalPosition}. 
\end{remark}

\begin{remark} 
Observe that the form of the mean field game solution in~\eqref{eq:opt_uiInf} is similar to agent $i$'s optimal response~$\hat{u}^{i,N}$ in the finite player Nash equilibrium in Theorem~\ref{thm:main-finite}. The major difference is that the equilibrium's aggregated transient price impact $\bar Y^{\bar u^N}$ in~\eqref{eq:opt_ui} from the $N$ agents is now replaced by the limiting mean field's transient price distortion $\tilde Y^{\tilde \nu}$ in~\eqref{def:Y_infMFG}. 
\end{remark}

Finally, the optimal infinite player mean field game strategies in Theorem~\ref{thm:main-infMFG} can also be written in terms of the mean field trading rate $\tilde\nu$ and mean field inventory $\tilde X^{\tilde \nu}$ instead of the mean field transient price impact $\tilde Y^{\tilde \nu}$ in the following way:

\begin{cor} \label{cor:main-infMFG} 
The optimal solution $(\hat{v}^i)_{i\in\mathbb{N}} \subset \mathcal A$ from Theorem~\ref{thm:main-infMFG} can alternatively be represented as
\be \label{eq:opt_ustar_alt} 
\begin{aligned} 
\hat{v}^i_{t} = \tilde{\nu}_t - \frac{R'(T-t)}{R(T-t)} \left( \tilde{X}_t^{\tilde \nu} - X^{\hat{v}^{i}}_t \right)
\end{aligned}
\ee
for all $t \in (0,T)$ where $\tilde{\nu}$ and $\tilde{X}^{\tilde \nu}$ denote the unique solution from~\eqref{eq:MeanField_infAggregated}.
\end{cor} 
\begin{proof}
Using the dynamics of $\tilde{\nu}$ and $\tilde{Y}^{\tilde \nu}$ from~\eqref{eq:MeanField_infAggregated} we can rewrite $\hat{v}^{i}_{t}$ in~\eqref{eq:opt_uiInf} as  
\begin{align*}
\hat{v}^{i}_{t} = & \, \frac{R'(T-t)}{R(T-t)} X^{\hat{v}^i}_{t} \\
& - \frac{1}{R(T-t)} E_{t}\left[
  \int_{t}^T R(T-s) d\tilde{\nu}_s + \frac{\phi}{\lambda} \int_t^T R(T-s) \tilde{X}_s^{\tilde \nu} ds + \frac{\kappa}{2\lambda} \sqrt{\frac{\phi}{\lambda}}
     \tilde Y_T^{\tilde \nu}  \right].
\end{align*}
Next, performing integration by parts for both integrals yields
\begin{align*}
\hat{v}^{i}_{t} = & \, \frac{R'(T-t)}{R(T-t)} X^{\hat{v}^i}_{t} \\
& \, - \frac{1}{R(T-t)} E_{t}\left[
  R(0) \tilde{\nu}_T  - R(T-t) \tilde{\nu}_t + \int_t^T \tilde{\nu}_s  R'(T-s) ds  - \frac{\phi}{\lambda} \tilde{R}(0) \tilde{X}_T^{\tilde \nu}  \right. \\
  & \hspace{85pt} \left. + \frac{\phi}{\lambda} \tilde{R}(T-t) \tilde{X}_t^{\tilde \nu} - \frac{\phi}{\lambda} \int_t^T \tilde{R}(T-s) \tilde{\nu}_s ds + \frac{\kappa}{2\lambda} \sqrt{\frac{\phi}{\lambda}}
    \tilde Y_T^{\tilde \nu}  \right]
\end{align*}
where $\tilde{R}(t) \triangleq \sinh ( \sqrt{\phi/\lambda} \, t) + \varrho \sqrt{\lambda} /(\lambda\sqrt{\phi}) \cosh ( \sqrt{\phi/\lambda} \, t)$. Note that $R'(T-s) = \tilde{R}(T-s) \phi /\lambda$, $R(0) = \sqrt{\phi/\lambda}$ and $\tilde{R}(0)=\varrho\sqrt{\lambda}/(\lambda\sqrt{\phi})$. Hence, using the fact that $\tilde{\nu}_T = \varrho \tilde X_T^{\tilde \nu}/\lambda - \kappa \tilde Y_T^{\tilde \nu}/(2\lambda)$ yields the claim in~\eqref{eq:opt_ustar_alt}.
\end{proof} 

\begin{remark} \label{rem-crowd} 
The mean field game solution $(\hat{v}^i)_{i\in\mathbb{N}}$ in~\eqref{eq:opt_ustar_alt}, expressed in terms of the mean field inventory $\tilde X^{\tilde \nu}$ and trading rate $\tilde{\nu}$, has exactly the same feedback form as the mean field game solution computed in~\citet[equation (19)]{CardaliaguetLehalle:18} (i.e., $R'(T-t)/R(T-t)$ equals $h_2(t)/(2\kappa)$ from~\cite[equation (18)]{CardaliaguetLehalle:18}) without a signal ($A\equiv 0$) and where the net trading flow's price impact is permanent (i.e., $\rho=y=0$). The difference is of course that $\tilde{X}^{\tilde \nu}$ and $\tilde{\nu}$ satisfy different equations, namely~\eqref{eq:MeanField_infAggregated} instead of~\cite[equation (17)]{CardaliaguetLehalle:18}. As discussed in~\cite{CardaliaguetLehalle:18} the representation in~\eqref{eq:opt_ustar_alt} shows that the optimal mean field game strategies $(\hat{v}^i)_{i\in\mathbb{N}}$ follow the mean field trading rate $\tilde{\nu}$ and gradually push their inventories $X^{\hat{v}^i}$ towards the mean field inventory $\tilde{X}^{\tilde \nu}$ because $R'(T-t)/R(T-t) > 0$ (recall that $\hat{v}^i$ is expressed as a selling rate). We illustrate this phenomenon in Section \ref{sec:illustrations}. 
\end{remark} 
\begin{remark} \label{rem-sol}
Interestingly, the finite-player game solution in~\eqref{eq:opt_ui} does not allow for a similar representation as the one in~\eqref{eq:opt_ustar_alt} in terms of the $N$ agents' average inventory~$\ol X^{\bar u^N}$ and average  selling rate $\bar u^N$. The reason is that the finite-player's aggregated four-dimensional FBSDE system in~\eqref{eq:mean-FBSDE}, which describes the average selling rate~$\bar u^N$, depends on an additional adjoint process~$\ol Z^{\bar u^N}$ because of the additional state variable $\ol Y^{\bar u^N}$. In contrast, this adjoint process disappears in the mean field's three-dimensional FBSDE system in~\eqref{eq:MeanField_infAggregated} for the mean field selling rate $\tilde{\nu}$ since $\tilde Y^{\tilde \nu}$ is not a state variable anymore. In fact, this makes the infinite-player mean field optimal policies significantly simpler to compute numerically than the $N$-agent Nash equilibrium strategies in~\eqref{eq:opt_ui}. More precisely, by combining~\eqref{eq:opt_ustar_alt}, \eqref{eq:opt_ubar-MeanField_infAggregated}, and the forward equation for $\tilde{Y}^{\tilde \nu}$ in~\eqref{eq:MeanField_infAggregated}, each agent $i$'s optimal control $\hat{v}^i$ can easily be computed by solving numerically a (random) three-dimensional linear forward ODE system in $(\tilde{X}^{\tilde \nu}, \tilde{Y}^{\tilde \nu}, X^{\hat{v}^{i}})$. This has been done to obtain the illustrations in Section~\ref{sec:illustrations} below. In contrast, the optimal policy~$\hat{u}^{i,N}$ in the finite-player game hinges on the representation in~\eqref{eq:opt_ui} which requires the computation of conditional expectations of integrals with respect to the future evolution of the $N$ agents' aggregated transient price distortion~$\ol Y^{\bar u^N}$.
\end{remark}

\section{Convergence Results and Approximations } \label{sec:approxNash}

In this section we present the main theoretical results of this
paper. We first prove that the optimal trading speed $\hat u^{i,N}$ in
the $N$-player game Nash equilibrium converges to the optimal trading
speed $\hat v^{i}$ of the mean field game in the $L^2$ norm and we
determine the convergence rate. Throughout this section we assume the
existence of $\hat u^{i,N}$ and $\hat v^{i}$, for which some
sufficient conditions were given in Sections \ref{sec:finitePlayer}
and \ref{sec:infinitePlayer}.

In order to state our result we assume that the agents in the finite player game and the agents in the mean field game start from similar initial inventories. More precisely, we introduce the following assumption on the initial inventories in \eqref{def:Xi} and \eqref{def:Xi_infMFG}:  
\be \label{init-assum} 
\sup_{i\geq 1} |x^{i}|  <\infty,
\ee
where for any $N\geq 2$ we set $x^{i,N}= x^i$, $i=1,\ldots,N$, in
\eqref{def:Xi}. Moreover, recall that we assume~\eqref{eq:initPosLimit}. 

We denote by  
\be \label{c-1} 
C_{\eqref{c-1}}  = 16\Big(\frac{\max\{\varrho, \kappa, \rho \kappa, \kappa \gamma, \phi, \rho  \}}{\lambda  }  \Big)^2.
\ee
We further make the following assumption on the constants: 
\be \label{a-const} 
20C_{\eqref{c-1}}(T^2\vee 1)T^2<1.
\ee
  
Our first convergence result is given in the following theorem. 
\begin{theorem} \label{thm-strat-con} 
For any $N\geq 2$, let $\hat{u}^{i,N}$ be the Nash equilibrium strategy of
player~$i$ in the $N$-player game in the sense of Definition
\ref{def:Nash}. Let $\hat v^i$ and $\hat \nu$ be the equilibrium
strategy of player $i$ and the mean field trading speed in the mean field
game in the sense of Definition~\ref{def:Nash_infMFG}. Under assumptions \eqref{init-assum} and \eqref{a-const} we have
$$
 \sup_{0\leq i\leq N}\sup_{0\leq s\leq T}E\big[(\hat u_s^{i,N}-\hat v^{i}_s)^2\big]  + \sup_{0\leq s\leq T} E\left[\left(\hat \nu_{s}-\frac{1}{N}\sum_{i=1}^N\hat u_s^{i,N}\right)^{2}\right]=O(N^{-2}). 
$$
\end{theorem} 
The proof of Theorem \ref{thm-strat-con} is given in Section \ref{sec:convergence}.
\begin{remark} 
Some convergence results on finite player equilibrium towards a mean field equilibrium, in the context of optimal execution, were derived recently by \citet{DrapeauLuoSchiedXiong:19}, see Theorem 4.4 therein. Theorem \ref{thm-strat-con} extends the results of \cite{DrapeauLuoSchiedXiong:19} as it provides additionally the convergence rate. Moreover, our model includes the transient price impact effects which were not considered in  \cite{DrapeauLuoSchiedXiong:19}. The proof of convergence of Drapeau et al. uses particular features of the model, which do not apply to our transient price impact case. In the proof of Theorem \ref{thm-strat-con} we develop a method which not only provides the rate of convergence but could also be adapted to more general models which translate at equilibrium to systems of FBSDEs. Theorem \ref{thm-strat-con}  also generalises Proposition 10 in \citet{FeronTankovTinsi:20} in a similar manner. 
\end{remark} 
\begin{remark} 
Our assumption in Theorem \ref{thm-strat-con} requires that $T$ is
bounded by a constant which depends on the model's parameters
$(\lambda,\gamma, \kappa, \rho, \varrho, \phi)$. The main reason for
this restriction arises from the fact that the corresponding systems
of FBSDEs \eqref{eq:NASHFBSDE} and \eqref{eq:MeanFieldFBSDE_inf} are
both degenerate in their forward components and do not satisfy
monotonicity assumptions (see, e.g., the terminal condition for
$u^{i,N}$ in  \eqref{eq:NASHFBSDE}). Moreover, a priori boundedness of $u^{i,N}$ uniformly in $N$ is unknown in this case. Therefore, standard FBSDE arguments for boundedness and uniqueness could not be adjusted to obtain convergence results as was done in \cite{Djete21,DrapeauLuoSchiedXiong:19,Lacker21,Tangpi}.   
\end{remark} 
In light of Lemmas \ref{thm:NASHFBSDE}, \ref{thm:MeanFieldFBSDE_inf} and Corollary \ref{cor:MeanFieldFBSDE_inf}, it is enough to prove the convergence of the solutions to corresponding FBSDE systems, which coincide with $\hat u^{i,N}$ and $(\hat v^i, \hat \nu)$. One of the ingredients in the proof of Theorem \ref{thm-strat-con} is the following proposition that derives a uniform bound on the solutions to these FBSDE systems.  
\begin{proposition} \label{lem-bnd-ui}
\begin{itemize} 
\item[\textbf{(i)}]Let $\hat u^{i,N}$ be the solution to \eqref{eq:NASHFBSDE}. Under assumptions \eqref{init-assum} and \eqref{a-const} we have 
$$
\limsup_{N \rightarrow \infty} \sup_{1\leq i \leq N} \sup_{t\in [0,T] }E\left[(\hat u^{i,N}_t)^2\right] <\infty. 
$$
\item[\textbf{(ii)}] Let $\hat v^i$ be the solution to \eqref{eq:MeanFieldFBSDE_inf}. Under assumption \eqref{init-assum} and \eqref{a-const} we have 

$$
\limsup_{N \rightarrow \infty} \sup_{1\leq i \leq N} \sup_{t\in [0,T] } E\left[(\hat v^i_t)^2\right] <\infty. 
$$
\end{itemize} 
 \end{proposition} 
The proof of Proposition \ref{lem-bnd-ui} is given in Section \ref{sec-pf-bnd}. 

From Theorem \ref{thm-strat-con} and Proposition \ref{lem-bnd-ui} we deduce the following result on the convergence of the value functions of the two games. 
\begin{corollary} \label{corr-val}
Let $J^{i,N}(\hat u^{i,N}; \hat u^{-i,N})$ be as in Definition \ref{def:Nash} and $J^{i,\infty}(\hat v^{i}; \hat \nu)$ be as in Definition~\ref{def:Nash_infMFG}. Then under assumptions \eqref{init-assum} and \eqref{a-const} we have   
$$
\sup_{1\leq i\leq N}|J^{i,N}(\hat u^{i,N}; \hat u^{-i,N}) - J^{i,\infty}(\hat v^{i}; \hat \nu) | =O(N^{-2}). 
$$
\end{corollary} 

The next theorem complements Theorem \ref{thm-strat-con} as it derives
almost sure convergence of the Nash equilibrium strategies from the
$N$-player game to the optimal strategies in the mean-field game
without assuming \eqref{a-const}, but at the cost of losing track of
the convergence rate.

To state our result, we first formulate suitable assumptions. Recall
that the matrix~$\ol F^N$ was defined in \eqref{def:Fbar}.  Moreover,
let \be \label{def:Fhat} \hat F^N = \begin{pmatrix}
  0 & 0 & -1 & 0 \\
  0 & 0 & 0 & 0 \\ -\frac{\phi}{\lambda} & 0 & \frac{\kappa \gamma}{2 \lambda N} & \frac{\rho}{2\lambda} \\
  0 & 0 & \frac{\kappa\gamma}{N} & \rho
    \end{pmatrix} \in \mathbb{R}^{4 \times 4} 
\ee
and let
\begin{equation} \label{b-00} 
    \ol F_{00} =
    \begin{pmatrix}
    \frac{\varrho}{\lambda} & -\frac{\kappa}{2\lambda} \\ 
    0 & 0
    \end{pmatrix}, \quad    \hat F_{00} =
    \begin{pmatrix}
    \frac{\varrho}{\lambda} & 0 \\
    0 & 0
    \end{pmatrix}.
\end{equation}
 We also introduce the matrices in $\mathbb{R}^{4\times 2}$
\begin{equation} \label{k-mat-exp}
    \ol K^N(t) = \exp(-\ol F^N \cdot t) \cdot \begin{pmatrix} I \\
      \ol F_{00} \end{pmatrix}, \quad   \hat K^N(t) = \exp(- \hat F^N \cdot t)
    \cdot \begin{pmatrix} I \\ \hat F_{00} \end{pmatrix},
  \end{equation}
  where $I$ represents the identity matrix in
$\mathbb{R}^{2\times 2}$.  Denote by
 $$
 \ol K^{N,(2)}(t) = \left(\ol K^N_{i,j}(t) \right)_{i,j=1,2}, \quad
 \hat  K^{N,(2)}(t) = \left(\hat K^N_{i,j}(t) \right)_{i,j=1,2} \quad (0\leq t \leq T);
 $$ 
 that is, $\ol K^{N,(2)}(t)$ and $\hat K^{N,(2)}(t)$ are matrices in $\mathbb{R}^{2\times 2}$ constructed from the first two rows of $\ol K^N(t)$ and $\hat K^N(t)$, respectively. 

In the following we assume that the set of parameters $\xi \triangleq (\lambda,\gamma,
\kappa, \rho, \varrho, \phi, T) \in \re^7_+$ are chosen such,
\be \label{assump-k}
\liminf_{N \geq 1}\inf_{t \in [0,T]}\left |\det \left( \ol K^{N,(2)}(t) \right)\right | > 0, \quad \liminf_{N \geq 1}\inf_{t \in [0,T]}\left |\det \left( \hat K^{N,(2)}(t)  \right) \right| > 0. 
\ee

  \begin{remark} \label{rem:assump:long1} Note that
    assumption~\eqref{assump-k} can be further simplified by
    explicitly computing the matrix exponentials
    in~\eqref{k-mat-exp} very
    similar to the computations performed in
    Section~\ref{subsec:proofs-matrixexponentials}. This will in turn
    lead to assumptions akin to those formulated in
    Sections~\ref{sec:finitePlayerSol} and Section~\ref{sec:infinitePlayerSol}, i.e.,
    Assumptions~\ref{assump:1} and~\ref{assump:2}, hence we omit the
    details. 
\end{remark} 

We are now ready to state our second convergence result.
 
\begin{theorem} \label{thm-con-long} For any $N\geq 2$, let
  $\hat{u}^{i,N}$ be the Nash equilibrium strategy of player~$i$ in
  the $N$-player game in the sense of Definition \ref{def:Nash}. Let
  $\hat v^i$ and $\hat \nu$ be the equilibrium strategy of player $i$
  and the mean field trading speed in the mean field game in the sense of
  Definition \ref{def:Nash_infMFG}. Under assumptions
  \eqref{eq:initPosLimit},  
  \eqref{init-assum} and \eqref{assump-k} we have
\begin{itemize} 
\item [\textbf{(i)}] 
$$
 \lim_{N\rr \infty}  \sup_{0\leq s\leq T}  \left| \hat
   \nu_{s}-\frac{1}{N}\sum_{i=1}^N\hat u_s^{i,N} \right| =0 \qquad \textrm{a.s.},
$$
\item [\textbf{(ii)}] 
 $$
\lim_{N\rr \infty}   \sup_{0\leq s\leq T} \big|\hat u_s^{i,N}-\hat
v^{i}_s\big|  =0  \qquad \textrm{a.s.}
$$
\end{itemize} 
\end{theorem} 

The proof of Theorem \ref{thm-con-long} is given in Section
\ref{sec-pf-long}.  

  \begin{remark}   
    As a byproduct of the proof of Theorem \ref{thm-con-long} we
    provide alternative representations of the solutions to the FBSDE
    systems
    in~\eqref{eq:mean-FBSDE},~\eqref{eq:NASHFBSDE*},~\eqref{eq:MeanField_infAggregated},
    \eqref{eq:MeanFieldGameFBSDE_inf}
    presented in Sections~\ref{sec:finitePlayerSol}
    and~\ref{sec:infinitePlayerSol}; and hence of the equilibrium
    strategies in both games; see Section \ref{sec-pf-long}
    below. Therefore, assumption \eqref{assump-k} which gives
    necessary conditions for the existence of these solutions can replace
    Assumptions~\ref{assump:1} and~\ref{assump:2}.
  \end{remark}

In our next result we prove an $\eps$-Nash equilibrium for an agent in
the $N$-player game, who is executing according to the mean field
optimal strategy. More precisely, we show that agent $i$ may improve
her performance by at most $O(N^{-1})$ when she deviates from the mean
field strategy $\hat{v}^i$. To this end, for any $u\in \mathcal A$ we
introduce the following norm
$$
\|u\|_{2,T}= \Big(\int_0^TE[u_t^2]dt  \Big)^{1/2}. 
$$

We first recall the definition of an $\eps$-Nash equilibrium from
\citet[Section 4]{CasgrainJaimungal:18}. 
\begin{definition}
Let $\mathcal A$ denote a class of admissible controls and fix
$\eps>0$. A set of controls $\{w^j \in \mathcal A : j =1,\ldots,N\}$
forms an $\eps$-Nash equilibrium with respect to a collection of
objective functionals $\{J^j(\cdot \,;\cdot): j=1,\ldots,N\}$ if it satisfies 
$$
J^{j}(w^{j}; w^{-j}) \leq \sup_{w \in \mathcal A }J^{j}(w; w^{-j}) \leq J^{j}(w^j; w^{-j}) + \eps \quad \textrm{for all } j=1,...,N.   
$$
\end{definition} 

\begin{theorem} \label{thm-eps-nash}
Let $J^{i,N}$ be the performance functional of the $N$-player game in
\eqref{def:FPGobjective}. Recall that $\{\hat v^i  : \, i \in
\mathbb{N}\}$ are the equilibrium strategies of the mean field game maximizing \eqref{def:objective_infMFG}. Then under assumptions \eqref{init-assum} and \eqref{a-const} there exists $C>0$ independent of $N$ and $u$ such that for all $u \in \mathcal A$ and $i \in \mathbb{N}$ we have 
$$
J^{i,N}(\hat{v}^{i}; \hat{v}^{-i}) \leq J^{i,N}(u; \hat{v}^{-i}) \leq J^{i,N}(\hat{v}^{i}; \hat{v}^{-i})+ C \|u\|^2_{2,T}(1\vee \|u\|^2_{2,T}) \left(\frac{1}{N} \right). 
$$
\end{theorem} 
The proof of Theorem \ref{thm-eps-nash} is given in Section \ref{sec-eps-pf}. 
\begin{remark} 
We can get a similar result as in Theorem \ref{thm-eps-nash} by replacing assumption \eqref{a-const} with the assumptions of Proposition \ref{prop:sol-MeanField_infAggregated} and Theorem \ref{thm:main-infMFG}, at the price of not having the precise convergence rate $N^{-1}$. This is done by deriving a uniform bound as in Proposition \ref{lem-bnd-ui} on $\{\hat{v}^{i}\}_{i\in \mathbb{N}}$, using the explicit solution in Theorem \ref{thm:main-infMFG} and a Gronwall-type argument as in the proof of~\cite[Theorem 3.2]{N-V19}, step 2. Then, one needs to repeat the same steps as in Lemma \ref{lem-j-dif} using the bound in Remark \ref{rem-l2-con} instead of Lemma \ref{lemma-con-mf}. The rest of the proof is similar to the proof of Theorem \ref{thm-eps-nash}, so we leave the details to the reader. 
\end{remark} 

From Proposition \ref{lem-bnd-ui} it follows that we can define a class $\mathcal A_{b}$ of admissible strategies that includes $\{(\hat u^{1,N},\ldots ,\hat u^{N,N}) : \ N\in \mathbb{N} \}$ and $\{\hat v^{i}\}_{i \in \mathbb{N}}$ such that 
$$
\sup_{u \in\mathcal A_{b} }  \|u\|^2_{2,T} <\infty. 
$$
Then the following corollary follows immediately from Theorem \ref{thm-eps-nash}. 
\begin{corollary} [$\eps$-Nash equilibrium] \label{thm-eps-nash2} Let
  $J^{i,N}$ be the performance functional of the $N$-player game in
  \eqref{def:FPGobjective}. Recall that
  $\{\hat v^i : \, i \in \mathbb{N}\}$ are the equilibrium strategies
  of the mean field game maximizing
  \eqref{def:objective_infMFG}. Then under assumptions \eqref{init-assum} and \eqref{a-const} there exists $C>0$ independent
  of $N$ such that for all $u \in \mathcal A$ and $i \in \mathbb{N}$
  we have
$$
J^{i,N}(\hat{v}^{i}; \hat{v}^{-i}) \leq \sup_{u \in \mathcal A_{b}} J^{i,N}(u; \hat{v}^{-i}) \leq J^{i,N}(\hat{v}^{i}; \hat{v}^{-i}) + O\left(\frac{1}{N} \right). 
$$
\end{corollary} 

\begin{remark} 
An $\eps$-Nash equilibrium result for execution games with partial information was derived by \citet{CasgrainJaimungal:18}, for the setting of permanent and temporary price impact. Corollary \ref{thm-eps-nash2} extends this result for the transient price impact case. 
\end{remark} 

\section{Illustrations} \label{sec:illustrations}

In this section we illustrate the agents' optimal inventories
for the mean field game, which were derived in Theorem~\ref{thm:main-infMFG} (or, equivalently,
Corollary~\ref{cor:main-infMFG}).  As
in~\cite{Lehalle-Neum18},~\cite{N-V19} we consider the case where the
exogenous signal process $A$ is given by
\begin{equation*} \label{a-i}
A_t = \int_0^t I_s ds
\qquad (t \geq 0),
\end{equation*}
with $I=(I_t)_{t \geq 0}$ following an autonomous Ornstein-Uhlenbeck
process with dynamics that are given by
\begin{equation*}
\begin{aligned} \label{I-OU} 
I_{0}&=\iota, \quad dI_{t} = -\beta I_{t}\, dt +\sigma \, dW_{t} \qquad (t\geq 0).
\end{aligned} 
\end{equation*}
Here, $W=(W_t)_{t \geq 0}$ denotes a standard Brownian motion, which is
defined on the underlying filtered probability space and
$\beta, \sigma>0$ are some constants. We fix the values of the
parameters as
\begin{equation} \label{model-par}
    T = 10,\quad  \kappa = 1 ,\quad \gamma = 1,\quad \rho = 1, \quad
    \lambda = 0.5, \quad \phi = 0.1, \quad \varrho = 10, 
\end{equation}
as well as
\begin{equation} \label{signal-par}
    \iota = 1, \quad \beta = 0.1,\quad \sigma = 0.5.
\end{equation}
We compute and plot the equilibrium inventories $\hat{X}^i \triangleq X^{\hat v^i}$ of five different agents $i=1,...,5$ using the representation in equation~\eqref{eq:opt_ustar_alt} of Corollary~\ref{cor:main-infMFG} together with the representation of $\tilde\nu$ in~\eqref{eq:opt_ubar-MeanField_infAggregated} and the forward equation for $\tilde{Y}^{\tilde \nu}$ from~\eqref{eq:MeanField_infAggregated}. Note that the offset term in~\eqref{eq:opt_ubar-MeanField_infAggregated}  can be easily computed when the signal process $A$ is given by an integrated Ornstein-Uhlenbeck process; cf. also the single-player case in~\cite{N-V19}. To wit, each agent $i$'s optimal inventory $\hat{X}^i$ is computed by solving numerically a (random) three-dimensional linear forward ODE
system in $(\tilde{X}^{\tilde\nu},\tilde{Y}^{\tilde\nu},\hat{X^i})$. We also present the mean field price distortion $\tilde Y \triangleq \tilde Y^{\tilde \nu}$, the mean field inventory $\tilde X \triangleq \tilde X^{\tilde \nu}$ and the amplified signal $A -\kappa\tilde Y$ in the following cases:
\begin{itemize} 
\item [\textbf{(i)}] without additional exogenous signal, $A\equiv 0$ (deterministic case),  in figure \ref{fig:ill1}, 
\item  [\textbf{(ii)}] with increasing (positive) signal in figure \ref{fig:ill2}, 
\item  [\textbf{(iii)}]  with decreasing (negative) signal in figure \ref{fig:ill3}.  
\end{itemize} 
Each of these plots show the agents' and the mean field inventories for
various cases of initial mean field inventories $\tilde X_0$, where in
figure \ref{fig:ill1} ($A\equiv 0$) the case of $\tilde X_0=0$ is
omitted as the mean field inventory is $\tilde X \equiv 0$. 

From these figures we conclude a few interesting observations:
\begin{itemize} 
\item [\textbf{(i)}] They show that the optimal mean field game strategies $\hat{v}^i$ follow the mean field trading rate $\tilde{\nu}$ and gradually push the inventory $\hat{X}^i$ towards the mean field inventory $\tilde{X}$, as pointed out in Remark \ref{rem-crowd}. 
\item [\textbf{(ii)}] We observe that the aggregated order-flow (i.e., the price distortion) amplifies the effect of the exogenous price predicting signal on the trading speeds $\hat v^{i}$, at least when the trading is far from termination. When approaching the end of the time horizon, the traders tend to close their positions due to inventory penalties, and the price distortion often has an opposite direction to the signal.    
\item [\textbf{(iii)}] In the cases of increasing positive signal in figure \ref{fig:ill2} and decreasing negative signal in figure \ref{fig:ill3}, when the mean field initial inventory is set to $\tilde X_0 = 0$ (top panel), we observe that the mean field strategy forms a round-trip, which is triggered by the signal and reinforcing it for the individual agent $i$'s trading.   
\end{itemize}

\begin{figure}[htbp!] 
  \begin{center}
    \includegraphics[scale=.6]{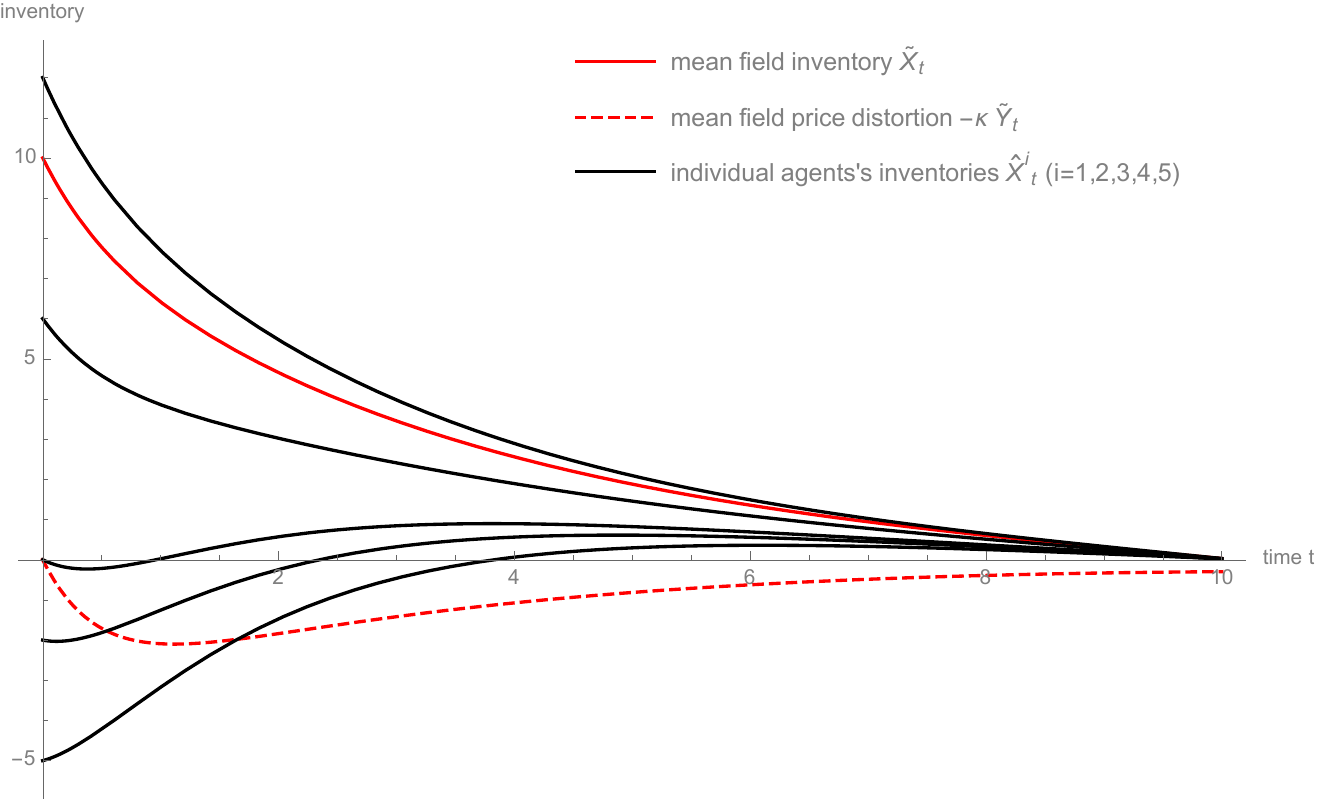}
    \par \vspace{1em} 
    \includegraphics[scale=.6]{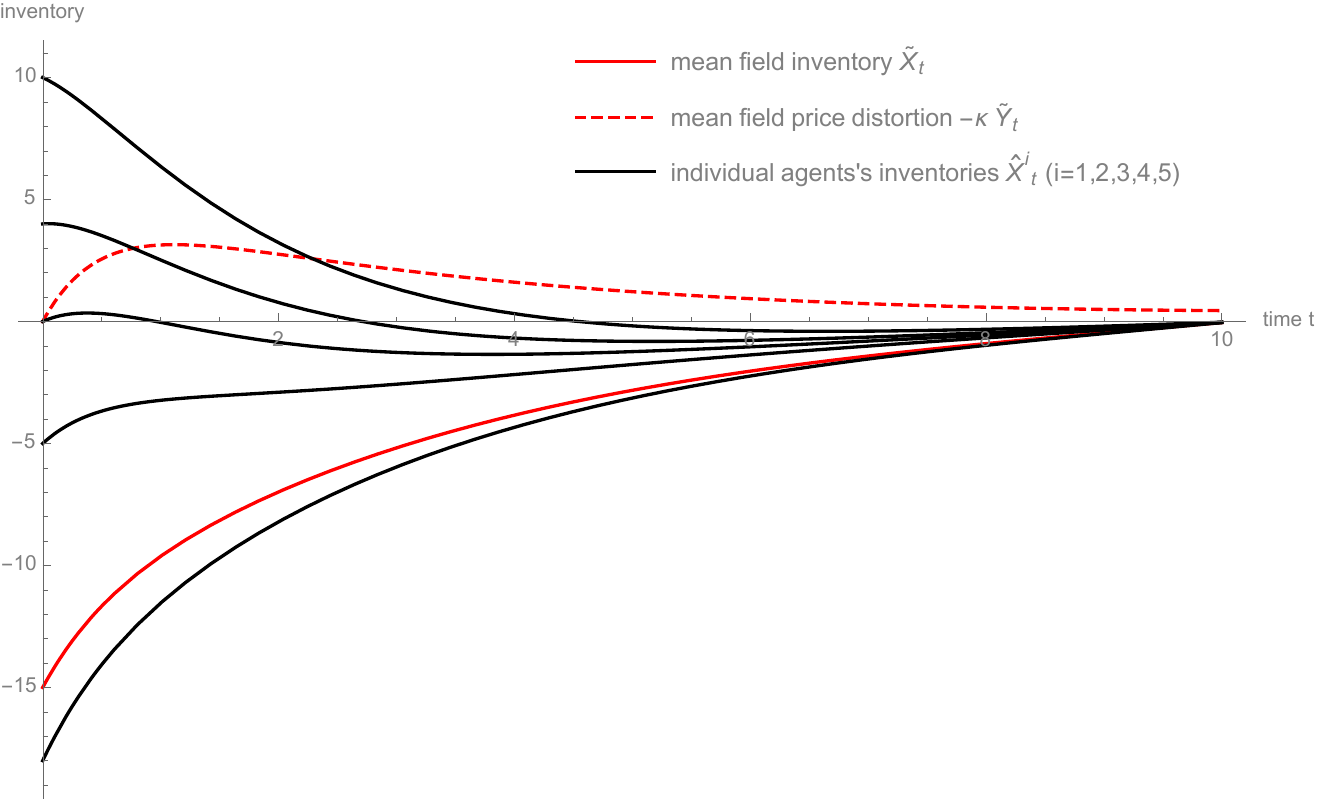}
    \caption{Case without additional exogenous signal $A\equiv 0$;
      Top: $\tilde X_0 =10$; bottom: $\tilde X_0 =-15$; all strategies
      are deterministic; mean field transient price distortion
      $-\kappa \tilde Y$ becomes the trading signal for the individual
      agents; note also the round-trip agent with zero initial
      inventory who is only exploiting the mean field's induced signal
      $-\kappa \tilde Y$.}
    \label{fig:ill1}
  \end{center}
\end{figure}

\begin{figure}[htbp!] 
  \begin{center}
    \includegraphics[scale=.38]{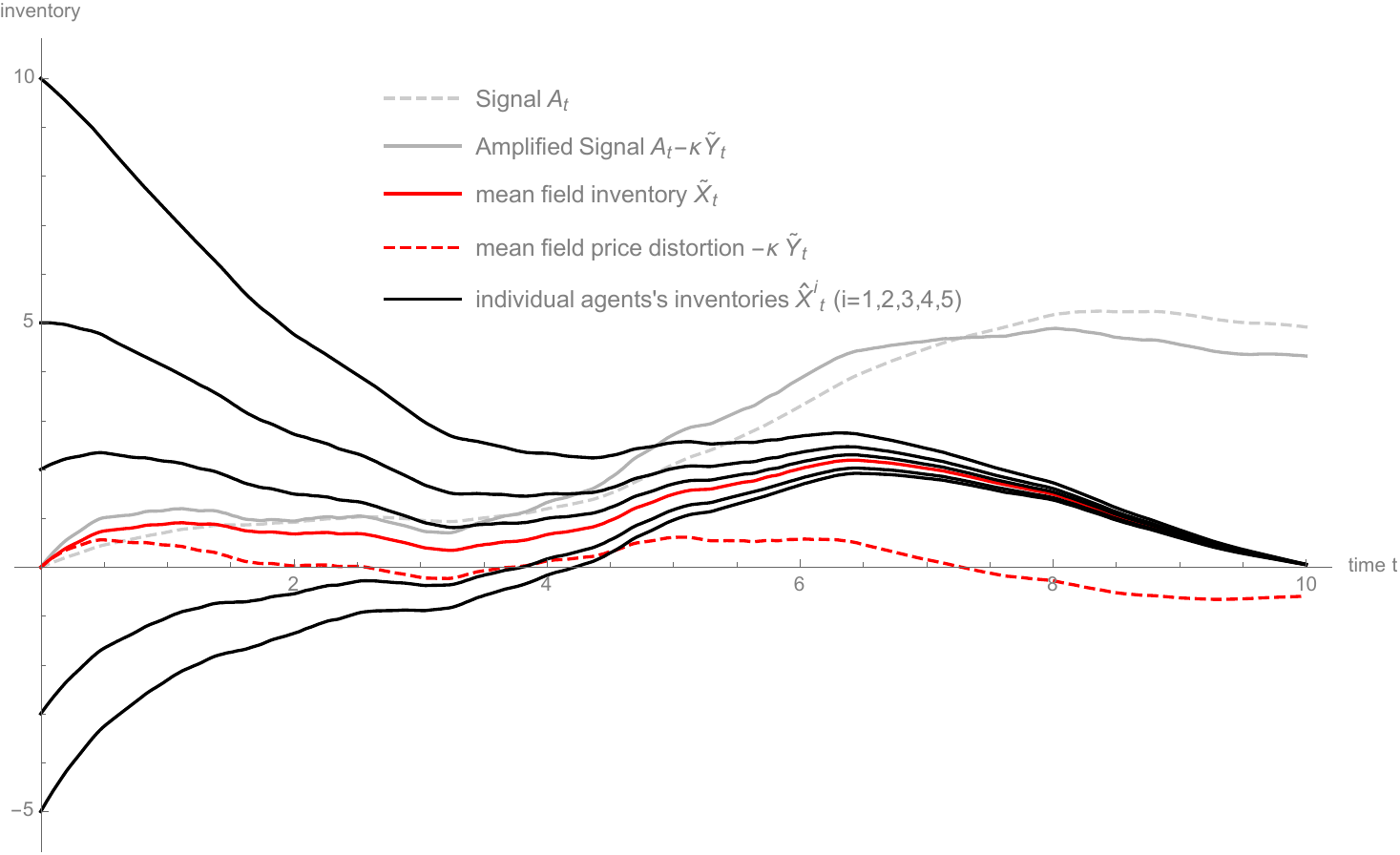}
    \par \vspace{1em} 
    \includegraphics[scale=.38]{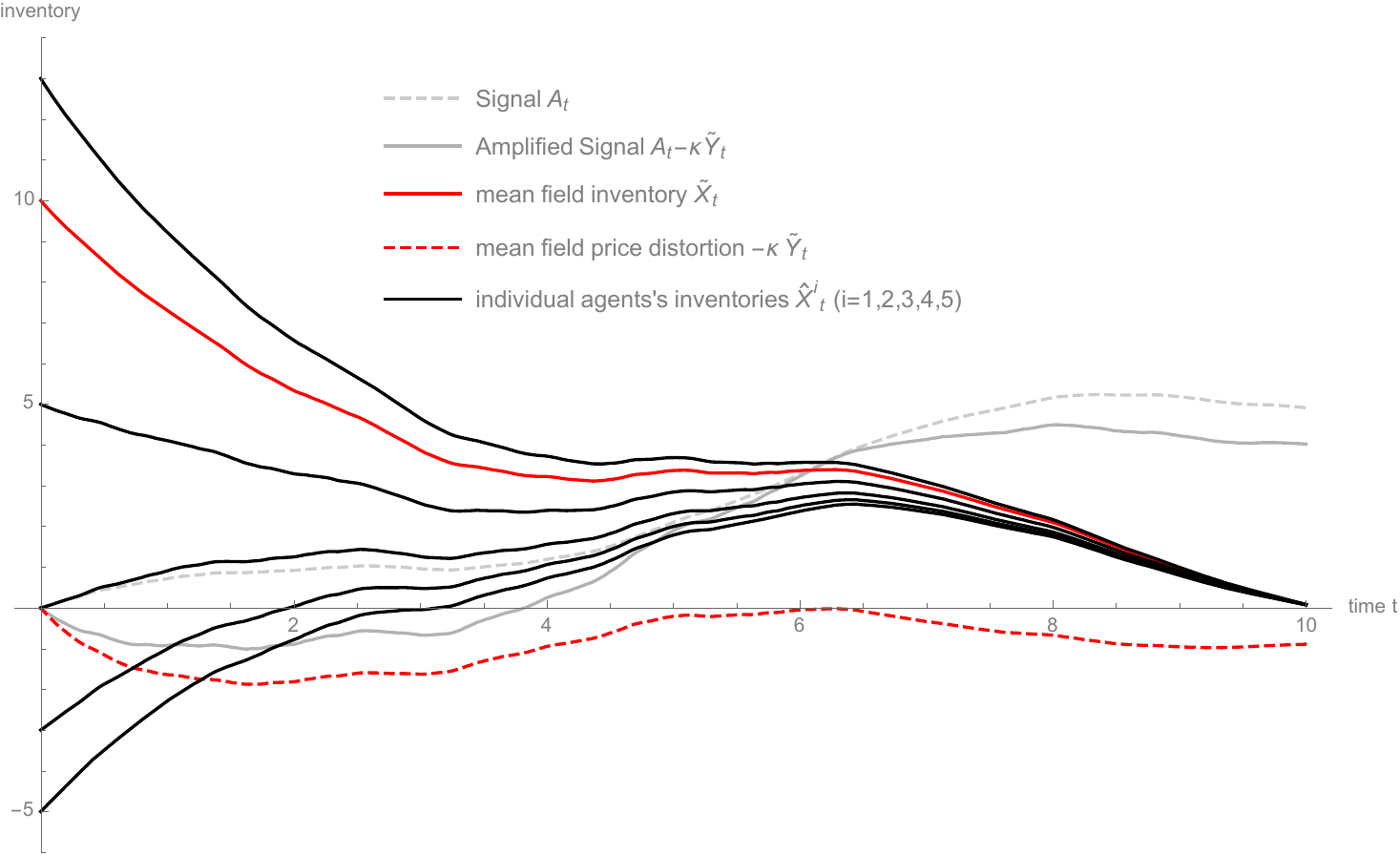}
        \par \vspace{1em} 
    \includegraphics[scale=.38]{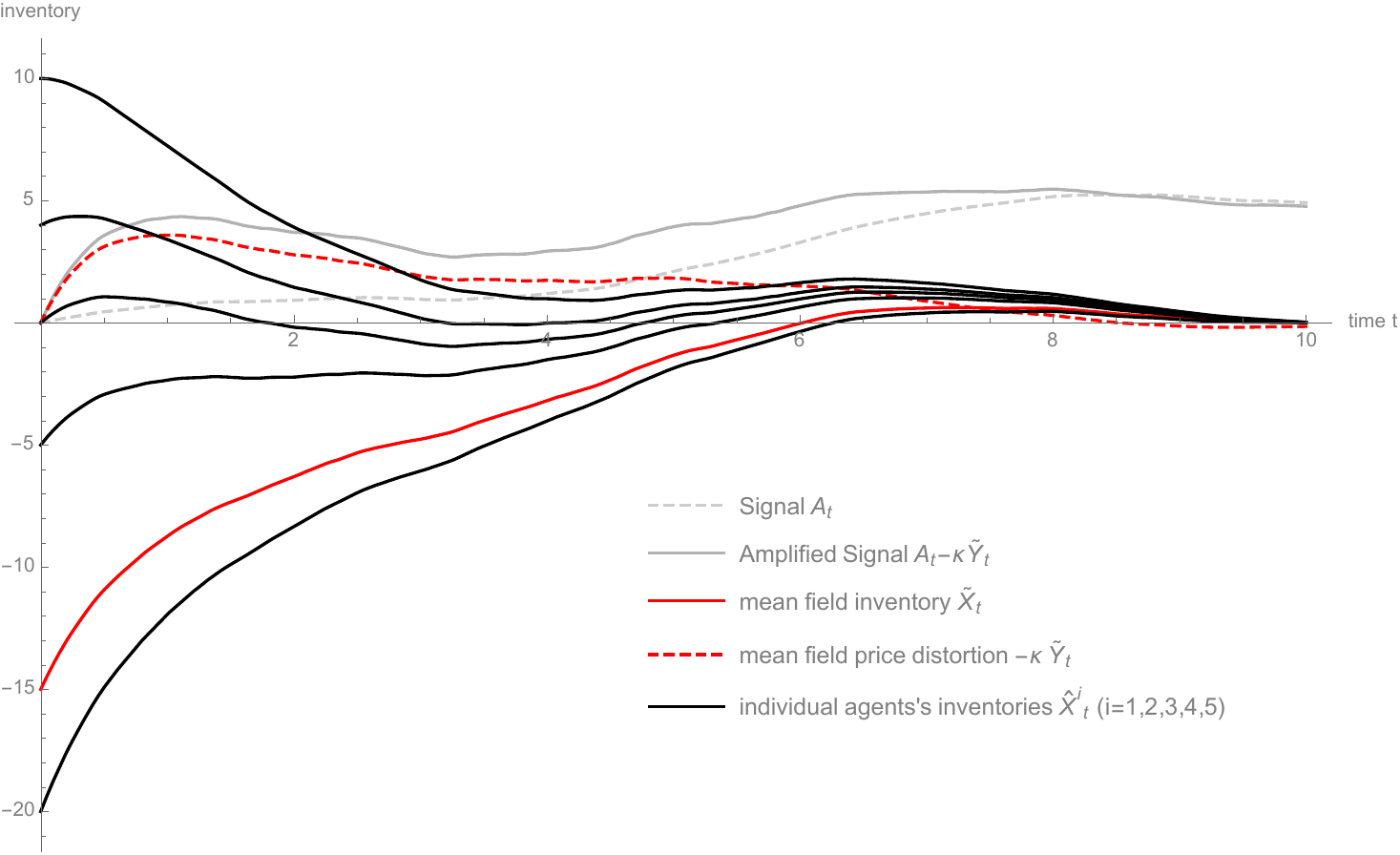}
    \caption{Case with increasing (positive) signal (same scenario in
      all three plots). Top: $\tilde X_0 =0$
      (round-trip); middle: $\tilde X_0 =10$; bottom: $\tilde X_0 =-15$. We present the agents' inventories $\hat X^i$ (black), mean field inventory $\tilde X$ (solid red), mean field price distortion $-\kappa \tilde Y$ (dashed red), signal $A$ (dashed grey) and amplified signal $A-\kappa \tilde Y$ (grey).}
    \label{fig:ill2}
  \end{center}
\end{figure}

\begin{figure}[htbp!] 
  \begin{center}
    \includegraphics[scale=.39]{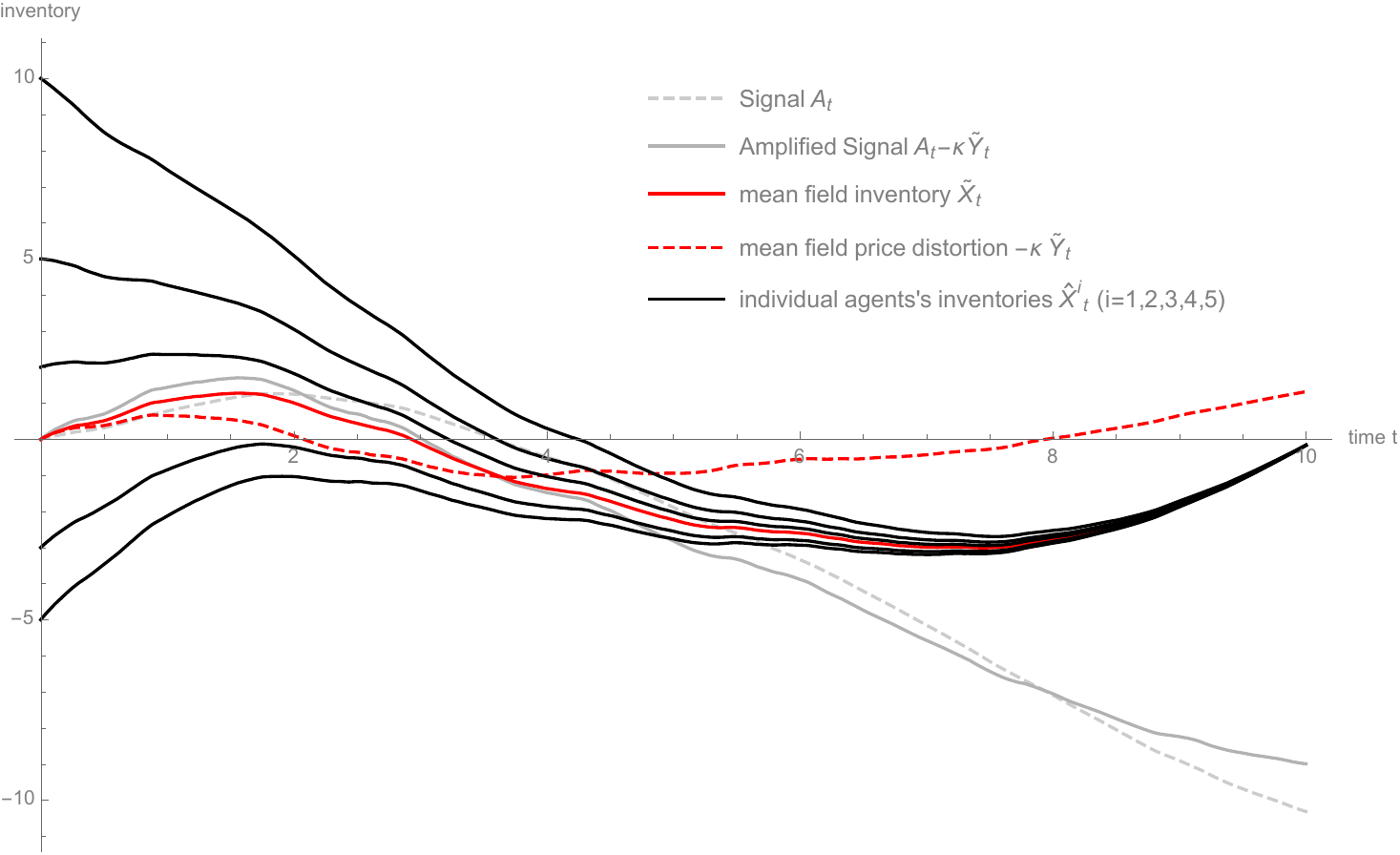}
    \par \vspace{1em} 
    \includegraphics[scale=.39]{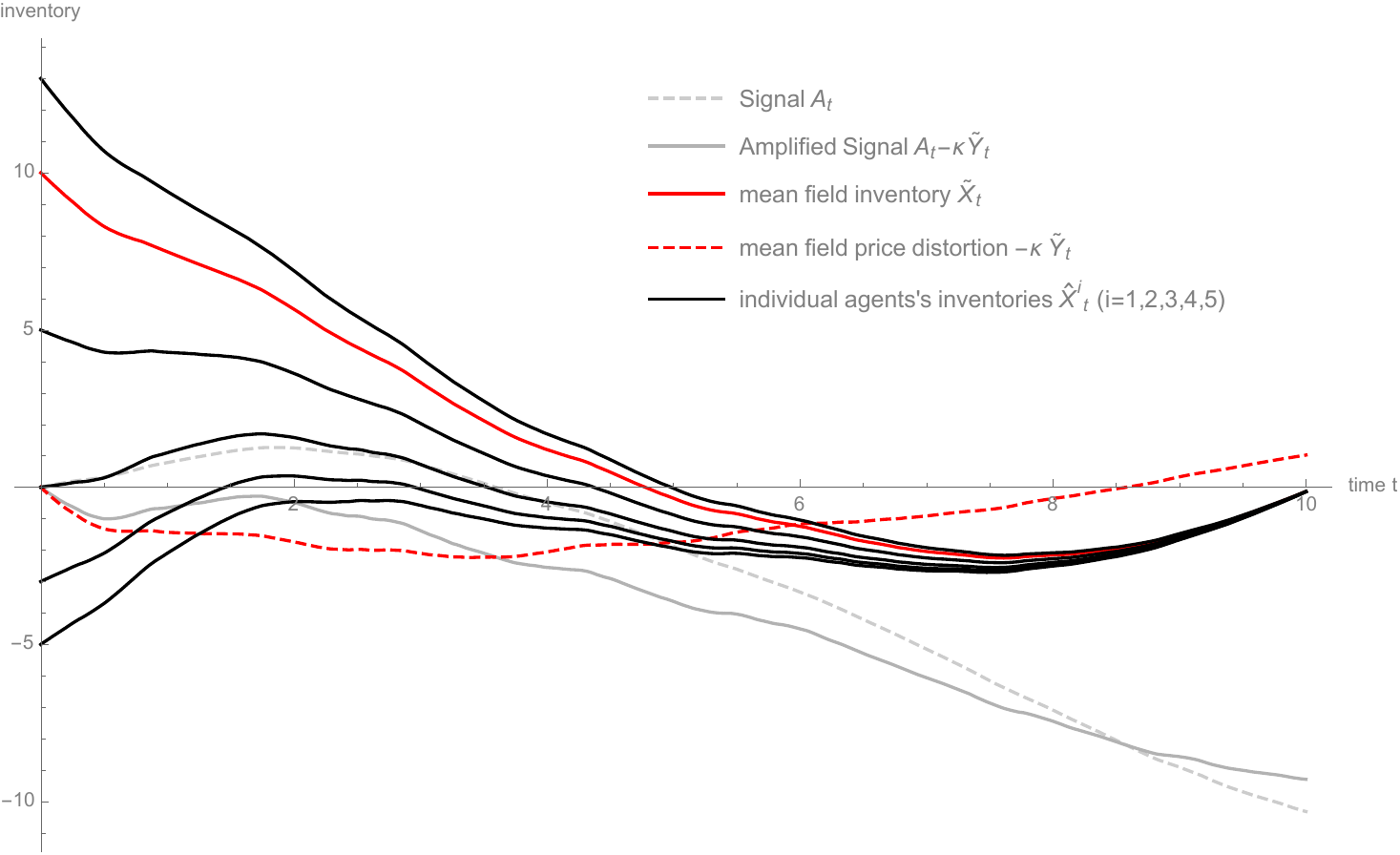}
    \par \vspace{1em}
    \includegraphics[scale=.39]{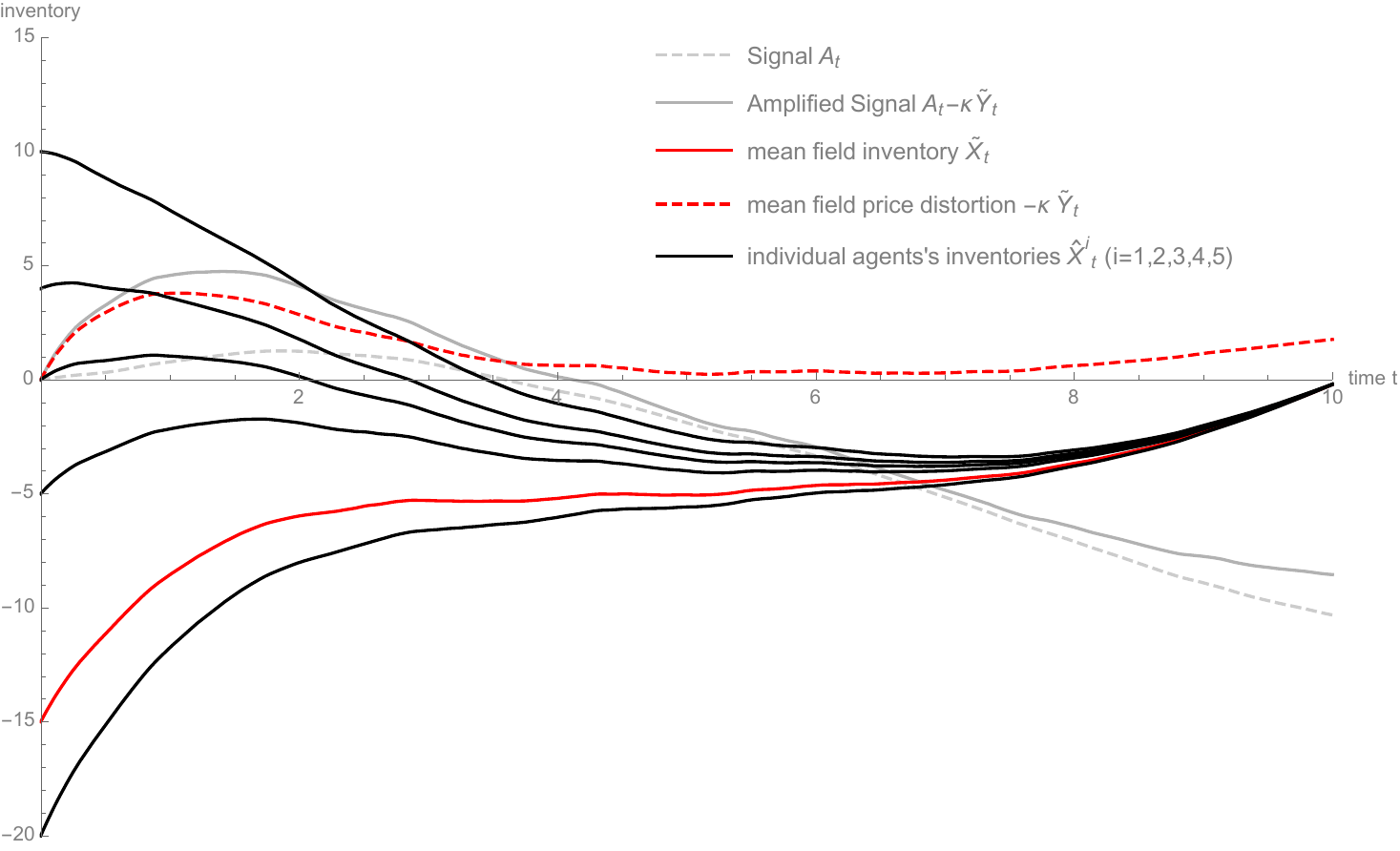}
    \caption{Case with decreasing (negative) signal (same scenario in
      all three plots): Top: $\tilde X_0 =0$, middle: $\tilde X_0 =10$;  bottom: $\tilde X_0 =-15$. We present the agents' inventories $\hat X^i$ (black), mean field inventory $\tilde X$ (solid red), mean field price distortion $-\kappa \tilde Y$ (dashed red), signal $A$ (dashed grey) and amplified signal $A-\kappa \tilde Y$ (grey).}
    \label{fig:ill3}
  \end{center}
\end{figure}


\section{Proof of Theorem \ref{thm-strat-con}} \label{sec:convergence}

\begin{remark}[Strategy of the proof]\label{strategy} 
We will show that the solutions $\{\hat u^{i,N}\}_{i=1}^{N}$ to \eqref{eq:NASHFBSDE} along with their average $\bar u^N$ converge in $L^{2}$ to the solutions $\{\hat v^{i}\}_{i=1}^{N}$ to \eqref{eq:MeanFieldGameFBSDE_inf}  and to the solution $\tilde \nu$ of \eqref{eq:MeanField_infAggregated}, respectively. Then from Lemmas \ref{thm:NASHFBSDE}, \ref{thm:MeanFieldFBSDE_inf} and from Corollary \ref{cor:MeanFieldFBSDE_inf} the convergence of $(\hat u^{i,N}, \hat u^{N})$ to $(\hat v^i, \hat \nu)$ would follow. 
\end{remark} 
In order to prove Theorem \ref{thm-strat-con} we introduce a few auxiliary lemmas, that concern the solutions to these FBSDE systems.  
\begin{lemma} \label{lem-z-con} 
Let $Z_t^{\hat u^{i,N}}$ as in \eqref{eq:NASHFBSDE} and assume that \eqref{a-const} holds. Then there exists a constant $C>0$ not depending on $(t,T,N,i)$ such that 
$$
\sup_{1\leq i\leq N} \sup_{0\leq t \leq T} E[ (Z_t^{\hat u^{i,N}})^2] \leq  CT N^{-2}.
$$
\end{lemma} 
\begin{proof}
We can write $Z_t^{\hat u^{i,N}}$ as follows 
\be \label{z-int} 
Z^{\hat u^{i,N}}_t = \kappa e^{\rho t}\left( \frac{ \gamma}{N} \int_{0}^{t}e^{-\rho s}\hat u^{i,N}_{s}ds + N^{i,N}_t \right), 
\ee
where the martingales $N^{i,N}$ in~\eqref{eq:NASHFBSDE} are given by (see Section \ref{subsec:proofs-finite} for the derivation), 
\begin{equation} \label{eq:Ni-martingale}
 N^{i,N}_t = -  \frac{ \gamma}{N}    E_{t}\left[\int_{0}^{T}e^{-\rho s}\hat u^{i,N}_{s}ds \right]  .
\end{equation}
Using the conditional Jensen's inequality and the tower property we get 
\bn
  E[(N^{i,N}_t)^{2}]& \leq&  \frac{ \gamma^2 }{N^{2}} E\left[ \left(   E_{t}\left[\int_{0}^{T}e^{-\rho s}\hat u^{i,N}_{s}ds \right] \right)^{2} \right] \\
 & \leq&  \frac{ \gamma^2 }{N^{2}} T^2 E\left[ \left(   E_{t}\left[\int_{0}^{T}e^{-\rho s}\hat u^{i,N}_{s}\frac{ds}{T} \right] \right)^{2} \right] \\
      &\leq &\frac{ \gamma^2 }{N^{2}} T E\left[ \int_{0}^{T} (\hat u^{i,N}_{s})^{2}ds  \right].\en
Together with Proposition \ref{lem-bnd-ui}(i) we get uniformly in $i=1,\ldots,N$ for $T<C_{\eqref{c-1}}$,  
$$
  E[(N^{i,N}_t)^{2}] =T\cdot O(N^{-2}). 
$$
Similarly we have uniformly in $i=1,...,N$, 

$$
E\left[ \left( \frac{ 1}{N} \int_{0}^{t}e^{-\rho s} \hat u^{i,N}_{s}ds  \right)^{2}\right]  = T\cdot O(N^{-2}), 
$$
and we conclude the result. 
\end{proof} 
 
Recall that $C_{\eqref{c-1}}$ was defined in \eqref{c-1}. 

\begin{lemma} \label{lemma-mar-bnd}
Let $(\hat u^{i,N},M^{i,N})$ as in \eqref{eq:NASHFBSDE} and $(\hat
v^i,L^i)$ as in~\eqref{eq:MeanFieldGameFBSDE_inf} and assume that \eqref{a-const} holds. Then there exists a constant $C>0$ not depending on $(N,i,t,T)$ such that for all $0\leq t\leq T$, 
\bn
&&E\left[\big(M^{i,N}_t -M_T^{i,N} - (L^{i}_t- L_T^i)\big)^2\right] \\
&&\leq C_{\eqref{c-1}}  T(T\vee 1)\Big( T^2\sup_{t\leq s \leq T}E\left[(\hat u_{s}^{i,N}-\hat v^i_s)^{2}\right]  + E[(X^{\hat u^{i,N}}_t-X^{\hat{v}^i}_t)^2] \Big) +T^2e^{2\rho T}C\frac{1}{N^2}. 
\en
\end{lemma} 
\begin{proof} 
The martingale $M^{i,N}$ in~\eqref{eq:NASHFBSDE} is given by, 
\be \label{eq:Mi-martingale}
    M^{i,N}_t =  \; \frac{1}{2\lambda} \tilde{M}^{i,N}_t - \frac{\gamma\kappa}{2\lambda N} \int_0^t e^{\rho s} d\tilde{N}^{i,N}_s, 
\ee
where
\begin{equation*} \label{eq:tildeNiLi-martingale}
    \tilde{N}^{i,N}_t \triangleq E_t \left[ \int_0^T e^{-\rho s} \hat u^{i,N}_s ds \right], 
\end{equation*}
 and 
 \begin{equation*} \label{eq:tildeMi-martingale}
    \tilde{M}^{i,N}_t \triangleq E_t \left[ 2\phi \int_0^T X^{\hat u^{i,N}}_s ds + 2 \varrho X^{\hat u^{i,N}}_T - P_T \right].
\end{equation*} 
The martingale $L^i$ from~\eqref{eq:MeanFieldGameFBSDE_inf} is given by 
\begin{equation} \label{eq:Li-martingale22}
    L^i_t \triangleq  \frac{1}{2\lambda}  E_t \left[ 2\phi \int_0^T X^{\hat{v}^i}_s ds + 2 \varrho X^{\hat{v}^i}_T - P_T \right].
\end{equation}
 See Sections \ref{subsec:proofs-finite} and \ref{subsec:proofs-infinite} for additional details on the derivation of the martingales in the finite player game and mean field settings, respectively. 

From the proof of Lemma \ref{lem-z-con} it follows that 
\be \label{t-N-bnd} 
\sup_{1\leq i\leq N} \sup_{0\leq t \leq T} E[ (\tilde N_t^{i,N})^2 ] = T\cdot O(1).
\ee
Hence from \eqref{eq:Mi-martingale} it follows that we only need to bound $ E\left[(\tilde M^{i,N}_t -\tilde M^{i,N}_T - (L^{i}_t-L^{i}_T)^2)\right] $.
Note that 
\bn
\tilde M^{i,N}_t -\tilde M^{i,N}_T &=& E_t \left[ 2\phi \int_t^T X^{\hat u^{i,N}}_s ds \right]-2\phi \int_t^T X^{\hat u^{i,N}}_s ds   \\
&&\quad + 2 \varrho \big(E_t\left[X^{ \hat u^{i,N}}_T\right]-X^{\hat u^{i,N}}_T\big)   - \big(E_t[P_T] -P_T \big). 
\en
For convenience we denote 
\be \label{L-hat} 
\hat L^{i}_t = 2\lambda L^{i}_t, \quad 0\leq t\leq T. 
\ee
Hence from \eqref{eq:Li-martingale22} we have, 
\bn 
\hat L^{i}_t-\hat L^{i}_T = E_t \left[ 2\phi \int_t^T X^{\hat v^i}_s ds \right]-2\phi \int_t^T X^{\hat v^i}_s ds  + 2 \varrho \big(E_t\left[X^{\hat v^i}_T\right]-X^{\hat v^i}_T\big)   - \big(E_t[P_T] -P_T \big). 
\en
 Using the following bound for any $n$ real numbers $a_i \in \mathbb {R}$, $i=1,...,n$, 
\be \label{2-p} 
\Big(\sum_{i=1}^N a_i \Big)^2 \leq n \sum_{i=1}^N a_i^2, 
\ee
together with the conditional Jensen inequality and the tower property we get  
 \be \label{jk1}
 \begin{aligned} 
& E\left[(\tilde M^{i,N}_t -\tilde M^{i,N}_T - (\hat L^{i}_t- \hat L^{i}_T)^2\right]  \\
& \leq  4(\phi^2\vee \varrho^2)\Bigg( E\left[\left( E_t \left[ \int_t^T X^{\hat u^{i,N}}_s ds \right]  - E_t \left[  \int_t^T X^{\hat{v}^i}_s ds \right]  \right)^2 \right] \\
&\quad+ E\left[\left(  \int_t^T X^{\hat u^{i,N}}_s ds   -   \int_t^T X^{\hat{v}^i}_s ds    \right)^2 \right] 
 +  E\left[\big( E_t\left[X^{\hat u^{i,N}}_T\right] - E_t\left[X^{\hat{v}^i}_T\right] \big)^2 \right] \\
 & \quad + E\left[\big(  X^{\hat u^{i,N}}_T  -  X^{\hat{v}^i}_T \big)^2 \right]\Bigg) \\
& \leq  8(\phi^2\vee \varrho^2)\bigg( T  E\left[ \int_t^T (X^{\hat u^{i,N}}_s -X^{\hat{v}^i}_s)^2 ds \right]    +    E\left[\big(  X^{\hat u^{i,N}}_T  -  X^{\hat{v}^i}_T \big)^2 \right]\bigg).
 \end{aligned} 
 \ee
Here again the factor $T$ appears due to normalization, since we have used Jensen inequality over the interval $[0,T]$.

 From \eqref{def:Xi} and \eqref{def:Xi_infMFG} we have 
\be \label{x-br} 
X^{\hat u^{i,N}}_s =X^{\hat u^{i,N}}_t -\int_t^s\hat u_r^{i,N}dr, \quad \textrm{for all } t\leq s \leq T,
\ee 
\be \label{x-br-h} 
X^{\hat v^i}_s =X^{\hat v^i}_t -\int_t^s \hat v_r^idr, \quad \textrm{for all } t\leq s \leq T.
\ee 
From \eqref{x-br} and \eqref{x-br-h} and Jensen inequality it follows that there exists $C>0$ not depending on $(i,N,t,T)$ such that,
\be  \label{x-bb-b-4}
\begin{aligned}
 E\bigg[ \int_t^T\big(X_s^{\hat u^{i,N}} -X_s^{\hat v^i} \big)^2ds  \bigg] &\leq   2\left( T E[(X^{\hat u^{i,N}}_t-X_t^{\hat v^i} )^2] +  T^3\sup_{t\leq s\leq T} E\left[(\hat u_{s}^{i,N}-\hat v_{s}^{i})^{2}\right] \right). 
 \end{aligned} 
\ee
By a similar argument we have 
\be  \label{x-bb-b2-4}
\begin{aligned}
 E\big[ (X^{\hat u^{i,N}}_T-X^{\hat{v}^i}_T)^2  \big] 
 &\leq  2 \left(E[(X^{\hat u^{i,N}}_t-X^{\hat{v}^i}_t)^2]  + T^2\sup_{t\leq s \leq T}E\left[(\hat u_{s}^{i,N}-\hat v^i_s)^{2}\right] \right).
 \end{aligned} 
\ee

By plugging \eqref{x-bb-b-4} and \eqref{x-bb-b2-4} into \eqref{jk1} and then using \eqref{eq:Mi-martingale}, \eqref{t-N-bnd} and \eqref{L-hat}  we get the result.  
\end{proof}

\begin{lemma} \label{lemma-mar-bnd-avr}
Let $(\bar u^N, \ol M^N)$ as in \eqref{eq:mean-FBSDE} and $(\tilde
\nu,\tilde L)$ as in~\eqref{eq:MeanField_infAggregated} and assume \eqref{a-const}. Then there exists a constant $C>0$ not depending on $(N,t,T)$ such that for all $0\leq t\leq T$,
\bn
&&E\left[\big(\ol M^N_t - \ol M^N_T - (\tilde L_t- \tilde L_T)\big)^2\right] \\
&&\leq C_{\eqref{c-1}}  T(T\vee 1)\Big( T^2\sup_{t\leq s \leq T}E\left[(\bar u^N_{s}-\tilde \nu_s)^{2}\right] + E[(\ol X^{\bar u^N}_t- X^{\tilde \nu}_t)^2]  \Big) +T^2e^{2\rho T}C\frac{1}{N^2}.    
\en
\end{lemma} 
The proof of Lemma \ref{lemma-mar-bnd-avr} is similar to the proof of Lemma \ref{lemma-mar-bnd} hence it is omitted. 


Recall the notation: 
\be \label{avr-def}
\bar u^N_t = \frac{1}{N}\sum_{i=1}^N \hat u_t^{i,N}, \quad t\geq 0. 
\ee
\begin{lemma} \label{lemma-converge-ui} 
Under Assumption \eqref{a-const} we have for all $0\leq t \leq T$, 
\bd  
\begin{aligned} 
&\sup_{t\leq s \leq T} E\left[(\hat u_{s}^{i,N}-\hat v_{s}^{i})^{2}\right]  \\ 
&\leq 10 C_{\eqref{c-1}} (T^2\vee 1) \Big(T^2 \sup_{t\leq s \leq T} E\left[(\hat u_{s}^{i,N}-\hat v_{s}^{i})^{2}\right] +T^2 \sup_{t\leq s \leq T} E\left[\left( \bar u^N_s -\tilde \nu_s\right)^{2}\right]  \\
&\quad +   E\big[(X^{\hat u^{i,N}}_t- X_t^{\hat v^i})]^2 +E\big[(Y^{\hat u^N}_t- Y_t^{\tilde \nu})]^2 \big] \Big)  + T^2e^{2\rho T}  O(N^{-2}). 
\end{aligned} 
\ed
\end{lemma} 
\begin{proof}

From \eqref{eq:NASHFBSDE}, \eqref{eq:MeanField_infAggregated} and \eqref{eq:MeanFieldGameFBSDE_inf} we have 
\bd
\begin{aligned} 
\hat u_t^{i,N}- \hat v^{i}_t& = -\int_t^T\left(\frac{\rho \kappa}{2\lambda} \big(Y^{\hat u^N}_s-  Y^{\tilde \nu}\big) - \frac{\kappa \gamma}{2\lambda}\Big(\frac{1}{N} \sum_{j\not =i}\hat u_s^{j,N}- \tilde \nu_{s}\Big) -\frac{\phi}{\lambda}\big(X_s^{\hat u^{i,N}}-X_s^{\hat v^{i}}\big) \right)ds \\
&\quad- \frac{\rho}{2\lambda}\int_t^TZ_s^{\hat u_{i,N}}ds  + M^{i,N}_t-M^{i,N}_T-(L^{i}_t-L_T^i)  + \frac{\varrho}{\lambda} \left( X^{\hat u^{i,N}}_T - X_T^{\hat v^i}\right)\\
&\quad - \frac{\kappa}{2\lambda}\left( Y^{\hat u^N}_T- Y^{\tilde v}_T \right) . 
\end{aligned} 
\ed
Using \eqref{2-p} it follows that
\be \label{dif-eq}
\begin{aligned} 
&E\big[(\hat u_t^{i,N}-\hat v^{i}_t)^2\big] \\
&\leq 5\bigg(E\bigg[\bigg(\int_t^T\bigg(\frac{\rho \kappa}{2\lambda} \big(Y^{\hat u^N}_s-  Y^{\tilde \nu}_s\big) - \frac{\kappa \gamma}{2\lambda}\Big(\frac{1}{N} \sum_{i \not =j}\hat u_s^{i,N}- \tilde \nu_{s}\Big) -\frac{\phi}{\lambda}\big(X_s^{\hat u^{i,N}}-X_s^{\hat v^{i}}\big) \bigg)ds \bigg)^2\bigg] \\
&\quad+ E\bigg[\bigg(\frac{\rho}{2\lambda}\int_t^TZ_s^{\hat u^{i,N}}ds \bigg)^2\bigg]  \\
&\qquad +\frac{\varrho^2}{\lambda^2}  E[(X^{\hat u^{i,N}}_T-X_T^{\hat v^i})^2]+  \frac{\kappa^2}{4\lambda^2} E[(Y^{\hat u^N}_T- Y^{\tilde v}_T)^2] \bigg)
\\
&\leq  \frac{ C_{\eqref{c-1}}}{2}\bigg( E\bigg[\left(\int_t^T  \big(Y^{\hat u^N}_s-  Y^{\tilde \nu}_s\big)ds\right)^2\bigg] +E\bigg[ \left(\int_t^T \Big(\frac{1}{N} \sum_{i \not =j}\hat u_s^{i,N}- \tilde \nu_{s}\Big)ds \right)^2\bigg] \\
&\qquad + E\bigg[ \left(\int_t^T\big(X_s^{\hat u^{i,N}}-X_s^{\hat v^{i}}\big)ds \right)^2  \bigg] 
  + E[(X^{\hat u^{i,N}}_T-X_T^{\hat v^i})^2]  +E[(Y^{\hat u^N}_T- Y^{\tilde v}_T)^2]  \bigg)  \\
  &\qquad +5E\big[\big( M^{i,N}_t-M^i_T-(L^{i}_t-L_T^i)\big)^2\big] +T^2\cdot O(N^{-2}), 
\end{aligned} 
\ee
where $C_{\eqref{c-1}}$ is given by \eqref{c-1}. Note that we used Lemma \ref{lem-z-con} in the last inequality.

From \eqref{def:Y} and \eqref{def:Y_infMFG} with $\tilde \nu$ instead of $\hat \nu$ we have 
\be \label{y-br} 
Y^{\hat u^N}_s =Y^{\hat u^N}_t +\gamma \int_t^s e^{-\rho(s-r)}  \frac{1}{N}\sum_{i=1}^N\hat u_s^{i,N} dr, \quad \textrm{for all } t\leq s \leq T,
\ee 
\be \label{y-br-h} 
Y^{\tilde \nu }_s =Y_t^{\tilde \nu} +\gamma \int_t^s e^{-\rho(s-r)}  \tilde \nu_r dr, \quad \textrm{for all } t\leq s \leq T. 
\ee 

From \eqref{y-br}, \eqref{y-br-h}, \eqref{avr-def} and Jensen inequality we have 
\be \label{rrr1}
\begin{aligned}
E\bigg[\left(\int_t^T  \big(Y^{\hat u^N}_s-  Y^{\tilde \nu}_s\big)ds \right)^2 \bigg]&\leq 2T  E\bigg[\int_t^T \Big((Y^{\hat u^N}_t-Y_t^{\tilde \nu})^2+T\int_t^s (\bar u^N_r-\tilde \nu_r)^2dr\Big) ds  \bigg] \\
&\leq 2T^2  \left( E\big[(Y_t^{\hat u^N}-Y_t^{\tilde \nu})^2\big ]+T^2\sup_{t\leq s \leq T} E\big[(\bar u^N_s-\tilde \nu_s)^2\big] \right).
\end{aligned} 
\ee
Similarly we have 
\be \label{y-2nd} 
\begin{aligned} 
E[(Y^{\hat u^N}_T- Y^{\tilde v}_T)^2]& \leq 2 \left( E\big[(Y_t^{\hat u^N}-Y_t^{\tilde \nu})^2\big ]+ T\int_{t}^T E\big[(\bar u^N_s-\tilde \nu_s)^2\big] ds\right) \\
&\leq 2 \left( E\big[(Y_t^{\hat u^N}-Y_t^{\tilde \nu})^2\big ]+T^2\sup_{t\leq s \leq T} E\big[(\bar u^N_s-\tilde \nu_s)^2\big] \right).
\end{aligned} 
\ee
From Proposition \ref{lem-bnd-ui}, Jensen inequality and \eqref{avr-def} we have for all $0\leq t\leq T$,    
\be \label{rrrr3}
\begin{aligned}
&E\bigg[ \left(\int_t^T \Big(\frac{1}{N} \sum_{j \not =i}\hat u_s^{j,N}-\tilde \nu_{s}\Big)ds \right)^2 \bigg] \\
&\leq 2TE\bigg[ \int_t^T \Big(\frac{1}{N} \sum_{j =1}^N\hat u_s^{j,N}- \tilde \nu_{s}\Big)^2ds \bigg] +2 T\frac{1}{N^2} E\bigg[  \int_t^T \big( \hat u_s^{i,N} \big)^2ds \bigg]  \\
&\leq  2T E\bigg[ \int_t^T \big(\bar u^N_s- \tilde \nu_{s}\big)^2ds \bigg] + T^2O(N^{-2})    \\
&\leq  2 T^2\sup_{t\leq s\leq T} E\big[  \big(\bar u^N_s- \tilde \nu_{s}\big)^2  \big] +T^2O(N^{-2})   . 
\end{aligned} 
\ee

Apply \eqref{x-bb-b-4}, \eqref{x-bb-b2-4}, \eqref{rrr1}--\eqref{rrrr3} and Lemma \ref{lemma-mar-bnd} to \eqref{dif-eq} to get that for all $0\leq t \leq T$, 
\bd  
\begin{aligned} 
&\sup_{t\leq s \leq T} E\left[(\hat u_{s}^{i,N}-\hat v_{s}^{i})^{2}\right]  \\ 
&\leq 10 C_{\eqref{c-1}} (T^2\vee 1) \Big(T^2 \sup_{t\leq s \leq T} E\left[(\hat u_{s}^{i,N}-\hat v_{s}^{i})^{2}\right] +T^2 \sup_{t\leq s \leq T} E\left[\left( \bar u^N_s -\tilde \nu_s\right)^{2}\right]  \\
&\quad +   E\big[(X^{\hat u^{i,N}}_t- X_t^{\hat v^i})]^2 +E\big[(Y^{\hat u^N}_t- Y_t^{\tilde \nu})]^2 \big] \Big)  + T^2e^{2\rho T} O(N^{-2}). 
\end{aligned} 
\ed

\end{proof}

Before introducing the next lemma we recall that $\bar u_t$ was defined in \eqref{avr-def} and that $\ol X^{u^{N}}$ from \eqref{eq-rrr} is given by 
\be \label{x-ol}
 \ol X^{\hat u^{N}}_t =  \frac{1}{N}\sum_{i=1}^N X^{\hat u^{i,N}}_t.
 \ee
 Recall that $C_{\eqref{c-1}}$ was defined in \eqref{c-1}. 
\begin{lemma} \label{lemma-converge-avr} 
Assume that \eqref{a-const} holds. Then, for all $0\leq t \leq T$ we have
\bd  
\begin{aligned} 
\sup_{t\leq s \leq T} E\left[(\bar u^N_{s} - \tilde \nu_s)^{2}\right]  &\leq  10 C_{\eqref{c-1}} (T^2\vee 1)\Big(T^2  \sup_{t\leq s \leq T} E\left[\left( \bar u^N_s -\tilde \nu_s\right)^{2}\right]   \\
&\quad + E\big[(\ol X^{\hat u^{N}}_t - X^{\tilde \nu}_t)^2] + E\big[(Y^{\hat  u^N}_t -Y^{\tilde \nu}_t )^2 \big] \Big) + T^2e^{2\rho T}  O(N^{-2}). 
\end{aligned} 
\ed
\end{lemma} 
\begin{proof} 
The proof follows the same lines as the proof of Lemma
\ref{lemma-converge-ui}, so we only give the outline.

Recall from \eqref{eq-rrr} that $\ol Y^{\hat u^N} = Y^{\hat u^N}$. From \eqref{eq:mean-FBSDE} and \eqref{eq:MeanField_infAggregated} we have 
\bd
\begin{aligned} 
\bar u^N_t- \tilde \nu_t& = -\int_t^T\left(\frac{\rho \kappa}{2\lambda} \big(  Y^{\hat u^N}_s-  Y^{\tilde \nu}_s\big) - \frac{\kappa \gamma}{2\lambda}\Big(\frac{N-1}{N} \bar u^N_s- \tilde \nu_{s}\Big) -\frac{\phi}{\lambda}\big(\ol X^{\hat u^N}_s-X_s^{\tilde\nu}\big) \right)ds \\
&\quad- \frac{\rho}{2\lambda}\int_t^T \ol Z^{\bar u^N}_sds  + \ol M^N_t -\ol M^N_T  -( \tilde L_t- \tilde  L_T)\\
&\quad +\frac{\varrho}{\lambda} \left(   \ol X^{\hat u^N}_T- X_T^{\tilde \nu} \right)  - \frac{\kappa}{2\lambda}\left( Y^{\hat u^N}_T- Y^{\tilde v}_T \right) .
\end{aligned} 
\ed
Using Jensen inequality we get, 
\be \label{dif-eq-av}
\begin{aligned} 
&E\big[(\bar u_t- \tilde \nu_t)^2\big] \\
&\leq \frac{C_{\eqref{c-1}}}{2} \bigg(T E\bigg[\int_t^T  \big( Y^{\hat u^N}_s-  Y^{\tilde \nu}_s\big)^2ds\bigg] +T E\bigg[ \int_t^T \Big(\frac{N-1}{N} \bar u^N_s- \tilde \nu_{s}\Big)^2ds \bigg] \\
&\qquad + E\bigg[ T \int_t^T\big(\ol X^{\hat u^N}_s-X_s^{\tilde \nu}\big)^2ds  \bigg] 
+ T  E\bigg[ \int_t^T(\ol  Z^{\bar u^N}_s)^2ds  \bigg] \\
 &\qquad +E\big[\big(\ol X^{\hat u^N}_T- X_T^{\tilde \nu}\big)^2\big] +E[(Y^{\hat u^N}_T- Y^{\tilde v}_T)^2] \bigg) \\
&\qquad  +5 E\big[\big( \ol M^N_t -\ol M^N_T  -( \tilde L_t- \tilde  L_T) \big)^2\big] .
\end{aligned} 
\ee
From Proposition \ref{lem-bnd-ui}(ii) it follows that there exists $C>0$ not depending on $(N,t,T)$ such that 
\be \label{blaw1}
\begin{aligned}
E\bigg[ \int_t^T \Big(\frac{N-1}{N} \bar u^N_s- \tilde \nu_{s}\Big)^2ds \bigg]  
&\leq 2 E\Big[ \int_t^T (\bar u^N_s-\tilde \nu_s)^2ds \Big] + CT\frac{1}{N^{2}}  \\
&\leq 2 T \sup_{t\leq s\leq T}E\bigg[  (\bar u^N_s-\tilde \nu_s)^2  \bigg] + CT^2\frac{1}{N^{2}}.
\end{aligned} 
\ee
Similarly to \eqref{x-bb-b-4} and \eqref{x-bb-b2-4} it follows that  
\be  \label{ }
\begin{aligned}
 E\bigg[ \int_t^T\big(\ol X_s^{\hat u^N} -X_s^{\tilde \nu} \big)^2ds  \bigg] &\leq 2\left( T E[(\ol X^{\hat u^N}_t-X_t^{ \tilde \nu} )^2] +  T^3\sup_{t\leq s\leq T} E\left[(\bar u^N_{s}-\tilde \nu _{s})^{2}\right] \right),
 \end{aligned} 
\ee
and 
\be  \label{x-id-p}
\begin{aligned}
 E\big[ (\ol X^{ }_T-X^{ \tilde \nu }_T)^2  \big] 
 &\leq 2 \left(E[(\ol X^{\hat u^N}_t-X^{\tilde \nu}_t)^2]  + T^2\sup_{t\leq s \leq T}E\left[(\bar u^N_{s}-\tilde \nu_s )^{2}\right] \right).
 \end{aligned} 
 \ee
From \eqref{eq-rrr} and Lemma \ref{lem-z-con} we get that there exists a constant $C_2>0$, not depending on $(N,t,T)$ such that 
\be \label{rrr2} 
E\bigg[ \int_t^T\big(\ol  Z^{\bar u^N}_s)^2ds  \bigg]   \leq C_2T^2\frac{1}{N^{2}}.
\ee
By applying \eqref{blaw1}--\eqref{rrr2}, \eqref{rrr1}, \eqref{y-2nd} and Lemma \ref{lemma-mar-bnd-avr} to \eqref{dif-eq-av} we get the result.   
\end{proof}

Now we are ready to prove Theorem \ref{thm-strat-con}.  
\begin{proof} [Proof of Theorem \ref{thm-strat-con}]
 Recall that $C_{\eqref{c-1}}$ was defined in \eqref{c-1}. 
From Lemmas \ref{lemma-converge-ui} and \ref{lemma-converge-avr} it follows that there exists a constant $C_1>0$ not depending on $(N,T)$ such that  
\bd  
\begin{aligned} 
&\sup_{t\leq s \leq T} E\left[(\hat u_{s}^{i,N}-\hat v_{s}^{i})^{2}\right]  + \sup_{t\leq s \leq T} E\left[(\bar u^N_{s} - \tilde \nu_s)^{2}\right] \\
&\leq  C_1 e^{2\rho T} T^2  \frac{1}{N^2}+20 C_{\eqref{c-1}} (T^2 \vee1)\Big( T^2 \sup_{t\leq s \leq T} E\left[(\hat u_{s}^{i,N}-\hat v_{s}^{i})^{2}\right] +  T^2\sup_{t\leq s \leq T} E\left[\left( \bar u^N_s -\tilde \nu_s\right)^{2}\right]   \\
&\quad +   E\big[( X^{\hat u^{i,N}}_t -  X^{\hat v^i}_t)^2] + E\big[(\ol X^{\hat u^N}_t -  X^{\tilde \nu}_t)^2] +E\big[(Y^{\hat u^N}_t -Y^{\tilde \nu}_t )^2] \big] \Big) , \quad \textrm{for all } 0\leq t \leq T.
\end{aligned} 
\ed

By \eqref{a-const} we have  $\alpha(T) := 1-20 C_{\eqref{c-1}} T^2(T^2 \vee1)>0$ by assumption, it holds for all $0\leq t \leq T$,  
\be  \label{gfd} 
\begin{aligned} 
&\alpha(T)\left(\sup_{t\leq s \leq T} E\left[(\hat u_{s}^{i,N}-\hat v_{s}^{i})^{2}\right]  + \sup_{t\leq s \leq T} E\left[(\bar u^N_{s} - \tilde \nu_s)^{2}\right] \right) \\
&\leq   C_1(T)\frac{1}{N^2} 
 + C_2(T)\left(  E\big[( X^{\hat u^{i,N}}_t -  X^{\hat v^i}_t)^2] + E\big[(\ol X^{\hat u^N}_t -  X^{\tilde \nu}_t)^2] +E\big[(Y^{\hat u^N}_t -Y^{\tilde \nu}_t )^2] \big] \right) \\
&\leq   \hat C_1(T)  \frac{1}{N^2}  +  \hat C_2(T)\left(  \int_0^t   E\big[(\bar u^N_s-\tilde \nu_s)^2\big]ds  + \int_0^t E\big[( \hat u^{i,N}_s-\hat  v^{i}_s)^2\big]ds   \right)  \\
&\leq  \hat C_1(T)  \frac{1}{N^2}  +  \hat C_2(T)\left(  \int_0^t  \sup_{s\leq r \leq T} E\big[(\bar u^N_r-\tilde \nu_r)^2\big]ds  + \int_0^t \sup_{s\leq r \leq T} E\big[(\hat u^{i,N}_r-\hat  v^{i}_r)^2\big]ds   \right), 
\end{aligned} 
\ee
where we used \eqref{x-bb-b2-4}, \eqref{y-2nd} and \eqref{x-id-p} in the second inequality.
 
From Gronwall's inequality we get that there exist constants $C_i(T)>0$, $i=3,4$ not depending on $N$ such that 
\be \label{gfd1} 
\begin{aligned} 
\sup_{0\leq s \leq T} E\left[(\hat u_{s}^{i,N}-\hat v_{s}^{i})^{2}\right]  + \sup_{0\leq s \leq T} E\left[(\bar u^N_{s} - \tilde \nu_s)^{2}\right]  \leq  &  C_3(T) \frac{1}{N^2} e^{  T  C_4(T)}, \\
&\quad \textrm{for all } i=1,\ldots,N, \, N \geq 2, 
\end{aligned} 
\ee
and we conclude the result by Remark \ref{strategy}. 
\end{proof}


\section{Proof of Proposition \ref{lem-bnd-ui}} \label{sec-pf-bnd}
We will only prove part (i) of the lemma as the proof of part (ii) follows similar lines. We first introduce the following two lemmas. Recall that $\bar u^N_t$ was defined in \eqref{avr-def}. 

Following \eqref{ass:P} we fix a constant $C_{P}(T)>0$ such that 
\be \label{c-p} 
\sup_{t\in [0,T]}E[(P_t)^2] <C_{P}(T), 
\ee
which will be used throughout this section. 

Recall that $C_{\eqref{c-1}}$ was defined in \eqref{c-1}. 
\begin{lemma} \label{lemma-con-ui} 
There exists a positive constant $c_N = O(N^{-2})$ such that for all $0\leq t \leq T$, $i=1,..,N$, 
\bd  \begin{aligned} 
&\sup_{t\leq s \leq T} E\left[(\hat u_{s}^{i,N})^{2}\right]  \\&\leq \frac{8}{\lam^2} C_P(T)+ 10 (C_{\eqref{c-1}} +c_N) (T^2\vee 1)\Big( T^2 \sup_{t\leq s \leq T} E\left[(\hat u_{s}^{i,N})^{2}\right] +T^2 \sup_{t\leq s \leq T} E\left[\left( \bar u^N_s\right)^{2}\right]  \\
& \quad+  E\big[(X^{\hat{u}^{i,N}}_t)]^2 +E\big[(Y^{\hat u^N}_t)]^2 \big] +E\Big[  \int_0^t(\hat u_s^{i,N})^2 ds \Big]  \Big). 
\end{aligned} 
\ed
\end{lemma} 
\begin{proof} 
From \eqref{eq:NASHFBSDE} we have 
\be \label{u-eq}
\begin{aligned} 
\hat u_t^{i,N} & =  \frac{\varrho}{\lambda}  X^{\hat u^{i,N}}_T -
      \frac{\kappa}{2\lambda} Y_T^{\hat u^N }- \int_t^T\left(\frac{\rho \kappa}{2\lambda} Y^{\hat u^N}_s - \frac{\kappa \gamma}{2\lambda} \frac{1}{N} \sum_{i \not =j}\hat u_s^{i,N}  -\frac{\phi}{\lambda} X_s^{\hat u^{i,N}}  \right)ds \\
&\quad - \frac{\rho}{2\lambda}\int_t^TZ_s^{\hat u^{i,N}}ds  - M^{i,N}_T+M^{i,N}_t. 
\end{aligned} 
\ee
Using \eqref{2-p} we get that 
\be \label{dif-eq-b-b}
\begin{aligned} 
&E\big[(\hat u_t^{i,N})^2\big] \\
&\leq C_{\eqref{c-1}}\bigg( E\big[(X^{\hat u^{i,N}}_T)^2]+ E\big[ (Y_T^{\hat u^N })^2\big]+ TE\bigg[\int_t^T  \big(Y^{\hat u^N}_s)^2 ds\bigg] \\
&\quad +T E\bigg[ \int_t^T \Big(\frac{1}{N} \sum_{j\not =i} \hat u_s^{j,N}\Big)^2ds \bigg] + TE\bigg[ \int_t^T\big(X_s^{\hat u^{i,N}}\big)^2ds  \bigg] 
+ T E\bigg[ \int_t^T(Z_s^{\hat u^{i,N}})^2ds  \bigg]  \bigg)\\
&\quad +2 E\big[\big( M^{i,N}_t-M^{i,N}_T \big)^2\big] , 
\end{aligned} 
\ee
where we used Jensen's inequality with respect to the normalized Lebesgue measure on $[0,T]$, which added factors of $T$ above.  

From \eqref{y-br}, Jensen inequality and Fubini's theorem we have 
\be  \label{rrrr3-b}
\begin{aligned}
E\bigg[\int_t^T  \big(Y^{\hat u^N}_s \big)^2ds\bigg] &\leq 2 E\bigg[\int_t^T \Big((Y^{\hat u^N}_t)^2+T\int_t^s (\bar u^N_r)^2dr\Big) ds  \bigg] \\
&\leq 2T \left( E\big[(Y_t^{\hat u^N})^2\big ]+T^2\sup_{t\leq s \leq T} E\big[(\bar u^N_s)^2\big] \right).
\end{aligned} 
\ee
By a similar argument we have 
\be  \label{y-T-2nd}
\begin{aligned}
E\big[  \big(Y^{\hat u^N}_T \big)^2 \big]  
&\leq 2  \left( E\big[(Y_t^{\hat u^N})^2\big ]+T^2\sup_{t\leq s \leq T} E\big[(\bar u^N_s)^2\big] \right).
\end{aligned} 
\ee

Note that 
\bd
\begin{aligned}
E\bigg[   \Big(\frac{1}{N} \sum_{j \not =i}\hat u_r^{j,N}\Big)^2  \bigg] &\leq 2E\bigg[   \Big(\frac{1}{N} \sum_{j=1}^N\hat u_r^{j,N} \Big)^2  \bigg] +2 \frac{1}{N^2} E\big[    \big(  \hat u_r^{i,N} \big)^2  \big].
\end{aligned} 
\ed
So we get that
\be \label{rrr1-b}
 E\bigg[ \int_t^T  \Big(\frac{1}{N} \sum_{j \not =i}\hat u_r^{j,N}\Big)^2 dr \bigg]   
\leq 2T \left(  \sup_{t\leq s \leq T}E\big[   \big( \bar u^N_s\big)^2  \big] + \frac{1}{N^2}   \sup_{t\leq s \leq T} E\big[\big(  \hat u_s^{i,N} \big)^2  \big] ds \right) .
 \ee
From \eqref{x-br} and Jensen inequality it follows that,
\be  \label{x-bb-b}
\begin{aligned}
 E\bigg[ \int_t^T\big(X_s^{\hat u^{i,N}}\big)^2ds  \bigg] &\leq 2T \left( E[(X^{\hat u^{i,N}}_t)^{2}] +  T^2 \sup_{t\leq s\leq T} E\left[(\hat u_{s}^{i,N})^{2}\right] \right). 
 \end{aligned} 
\ee
By a similar argument we have 
\be  \label{x-bb-b2}
\begin{aligned}
 E\big[ (X^{\hat u^{i,N}}_T)^2  \big] 
 &\leq  2 \left(E[(X^{\hat u^{i,N}}_t)^2]  +T^2 \sup_{t\leq s \leq T}E\left[(\hat u_{s}^{i,N})^{2}\right] \right).
 \end{aligned} 
\ee
Define  
\be \label{c-2} 
C_{\eqref{c-2}}(T)=2\kappa^2 \gamma^4e^{2\rho T}.
\ee
From \eqref{eq:Mi-martingale}, \eqref{2-p} and Burkholder-Davis-Gundy inequality we get that 
 \be \label{erq}
\begin{aligned} 
&E\left[(M_T^{i,N}- M^{i,N}_t )^2\right] \\ 
&\leq 
\frac{1}{2\lambda^2}  E\left[\left(   2\phi \int_t^T X^{\hat u^{i,N}}_s ds + 2 \varrho X^{\hat u^{i,N}}_T - P_T  -E_t \left[ 2\phi \int_t^T X^{\hat u^{i,N}}_s ds + 2 \varrho X^{\hat u^{i,N}}_T - P_T \right] \right)^2 \right] \\
&\quad + 8\left(\frac{\kappa \gamma^2}{2\lambda N}\right)^2e^{2\rho T}E\left[ \left( \int_t^T E_s\Big[  \int_0^Te^{-\rho r} \hat u_r^{i,N} dr \Big] ds  \right)^2\right]  \\
 &\leq \frac{1}{\lam^2}E\left[\left(   2\phi \int_t^T X^{\hat u^{i,N}}_s ds + 2 \varrho \int_t^T\hat u_s^{i,N}ds  -E_t \left[ 2\phi \int_t^T X^{\hat u^{i,N}}_s ds + 2 \varrho \int_t^T\hat u_s^{i,N}ds\right] \right)^2 \right] \\
&\qquad  +  C_{\eqref{c-2}}  \frac{1}{\lam^2}\frac{1}{N^2} E\left[\left( \int_t^T E_s\Big[  \int_0^Te^{-\rho r} \hat u_r^{i,N} dr \Big] ds \right)^2 \right] +   \frac{1}{\lam^2} E\left[ (P_T- E_t[P_T])^2\right].
\end{aligned} 
\ee
Using the conditional Jensen inequality and \eqref{c-p} we get  
\be \label{bla}
\begin{aligned} 
&E\left[ (P_T- E_t[P_T])^2\right]  \leq 4 C_{P}(T). 
\end{aligned} 
\ee
Using both Jensen and conditional Jensen inequalities, the tower property and Fubini theorem give  
\be \label{bla2}
\begin{aligned} 
E\left[\left( \int_t^T E_s\Big[  \int_0^Te^{-\rho r} \hat u_r^{i,N} dr \Big] ds \right)^2 \right] &\leq T E\left[ \int_t^T \left( E_s\Big[  \int_0^Te^{-\rho r} \hat u_r^{i,N} dr  \Big] \right)^2 ds \right] \\
&\leq T^2 E\left[ \int_t^T  E_s\Big[  \int_0^Te^{-2\rho r} (\hat u_r^{i,N})^2 dr  \Big]  ds \right] \\
&\leq T^3   E\Big[  \int_0^T(\hat u_r^{i,N})^2 dr  \Big] \\
&\leq T^3\Big( T \sup_{t\leq s \leq T} E[(\hat u_s^{i,N})^2]+   E\Big[  \int_0^t(\hat u_s^{i,N})^2 ds \Big] \Big).
 \end{aligned} 
\ee
Using \eqref{x-bb-b} and the conditional Jensen inequality we get for all $0\leq t\leq T$, 
\be \label{bla2232}
\begin{aligned}
&E\left[\left(   2\phi \int_t^T X^{\hat u^{i,N}}_s ds + 2 \varrho \int_t^T\hat u_s^{i,N}ds  -E_t \left[ 2\phi \int_t^T X^{\hat u^{i,N}}_s ds + 2 \varrho \int_t^T\hat u_s^{i,N}ds\right] \right)^2 \right] \\
&\leq 16(\varrho \vee \phi)^2 E\left[\left(   \int_t^T X^{\hat u^{i,N}}_s ds +  \int_t^T\hat u_s^{i,N}ds \right)^2 \right] \\
&\leq 32(\varrho \vee \phi)^2 T \bigg( E\left[ \int_t^T (X^{\hat u^{i,N}}_s)^2 ds \right]  +  E\left[  \int_t^T(\hat u_s^{i,N})^2ds \right] \bigg) \\
&\leq 128 (\varrho \vee \phi)^2 T^2 \left( E[(X^{\hat u^{i,N}}_t)^{2}] +  T(T\vee1) \sup_{t\leq s\leq T} E\left[(\hat u_{s}^{i,N})^{2}\right] \right). 
 \end{aligned} 
\ee
By plugging in \eqref{bla}, \eqref{bla2} and \eqref{bla2232} to \eqref{erq} it follows that 
\be\label{rrrr4-d}
\begin{aligned} 
&  E\left[(M_T^{i,N}- M^{i,N}_t )^2\right] \\ 
&\leq \frac{4}{\lam^2}C_P(T) + 128 \left(\frac{\varrho \vee \phi}{\lam}\right)^2 T^2 \left( E[(X^{\hat u^{i,N}}_t)^{2}] +  T(T\vee1) \sup_{t\leq s\leq T} E\left[(\hat u_{s}^{i,N})^{2}\right] \right) \\
&\quad +C_{\eqref{c-2}}  \frac{1}{\lam^2}\frac{1}{N^2}T^3 \left(  T \sup_{t\leq s \leq T} E[(\hat u_s^{i,N})^2]+E\Big[  \int_0^t(\hat u_s^{i,N})^2 ds \Big] \right) .
\end{aligned} 
\ee

From \eqref{z-int} and \eqref{eq:Ni-martingale} and by following the same lines leading to \eqref{bla2}, we get for all $0\leq t\leq T$, 
\be \label{rrr2-b} 
E\bigg[ \int_t^T(Z_s^{\hat u^{i,N}})^2ds  \bigg]   \leq C_{\eqref{c-2}}  \frac{1}{\lam^2}\frac{1}{N^2} T^3\Big( T \sup_{t\leq s \leq T} E[(\hat u_s^{i,N})^2]+   E\Big[  \int_0^t(\hat u_s^{i,N})^2 ds \Big] \Big). 
\ee

Applying \eqref{rrrr3-b}--\eqref{x-bb-b2}, \eqref{rrrr4-d} and \eqref{rrr2-b} to \eqref{dif-eq-b-b}, we get that there exists $c_N = O(N^{-2})$ such that for all $0\leq t \leq T$,
\bd \label{d-ineq-b}
\begin{aligned} 
&\sup_{t\leq s \leq T} E\left[(\hat u_{s}^{i,N})^{2}\right]  \\&\leq \frac{8}{\lam^2} C_P(T)+ 10 (C_{\eqref{c-1}} +c_N) (T^2\vee 1) \Big(T^2 \sup_{t\leq s \leq T} E\left[(\hat u_{s}^{i,N})^{2}\right] +T^2 \sup_{t\leq s \leq T} E\left[\left( \bar u^N_s\right)^{2}\right]  \\
& \quad+  E\big[(X^{\hat{u}^{i,N}}_t)]^2 +E\big[(Y^{\hat u^N}_t)]^2 \big] +E\Big[  \int_0^t(\hat u_s^{i,N})^2 ds \Big]  \Big). 
 \end{aligned} 
\ed
\end{proof} 

Before we introduce the next lemma we recall that $\ol X^{\hat u^N}$ was defined in \eqref{x-ol}. 

\begin{lemma} \label{lemma-conv-bar} 
There exists a positive constant $c_N = O(N^{-2})$ such that 
\bd  
\begin{aligned} 
\sup_{t\leq s \leq T} E\left[(\bar u^N_{s})^{2}\right]  &\leq  \frac{8}{\lam^2} C_P(T)+ 10 (C_{\eqref{c-1}} +c_N) (T^2\vee 1) \Big(T^2\sup_{t\leq s \leq T} E\left[\left( \bar u_s^N\right)^{2}\right]   \\
&\quad +  E\big[(\ol X^{\hat{u}^N}_t)^2] +E\big[(Y^{\hat u^N}_t)^2\big]  \Big) , \quad \textrm{for all } 0\leq t \leq T.
\end{aligned} 
\ed
\end{lemma} 
The proof of Lemma \ref{lemma-conv-bar} follows the same lines as the proof of Lemma \ref{lemma-con-ui}, hence it is omitted. 

Now we are ready to prove Proposition \ref{lem-bnd-ui}. 

\begin{proof}[Proof of Proposition \ref{lem-bnd-ui}]
From Lemmas \ref{lemma-con-ui} and \ref{lemma-conv-bar} we get that there exists a constant $c_N = O(N^{-2})$ such that for all $0\leq t \leq T$ and $i=1,...,N$ we have 
\bd
\begin{aligned} 
&\sup_{t\leq s \leq T} E\left[(\hat u_{s}^{i,N})^{2}\right]  +\sup_{t\leq s \leq T} E\left[(\bar u^N_{s})^{2}\right]   \\
&\leq \frac{16}{\lam^2}C_P(T)+ 20 (C_{\eqref{c-1}} +c_N) (T^2\vee 1) \Big(T^2\sup_{t\leq s \leq T}  E\left[(\hat u_{s}^{i,N})^{2}\right] + T^2\sup_{t\leq s \leq T} E\left[\left( \bar u^N_s\right)^{2}\right]     \\
&\quad + E\big[(\ol X^{\hat{u}^N}_t)^2] +E\big[(X^{\hat u^{i,N}}_t)^2] +E\big[(Y^{\hat u^N}_t)^2] \big] +E\Big[  \int_0^t(\hat u_s^{i,N})^2 ds \Big]   \Big) . 
\end{aligned} 
\ed
By assumption \eqref{a-const} we can choose $N$ large enough so that $g(T,N):=20 (C_{\eqref{c-1}} +c_N) (T^2\vee 1)T^2<1$. Then using Jensen inequality we get for all $0\leq t \leq T$, $N\geq 1$ and $i=1,...,N$: 
\be  \label{bnd-s-t} 
\begin{aligned} 
&(1-g(T,N))\left(\sup_{t\leq s \leq T} E\left[(\hat u_{s}^{i,N})^{2}\right]  +\sup_{t\leq s \leq T} E\left[(\bar u^N_{s})^{2}\right] \right) \\
&\leq C_1(T)\left( E\big[(\ol X^N_t)^2] +E\big[(X^{\hat u^{i,N}}_t)^2] +E\big[(Y^{\hat u^N}_t)^2] +E\Big[  \int_0^t(\hat u_s^{i,N})^2 ds \Big]  \right) + C_2(T)  \\
&\leq    \wt C_1(T)\left( (\ol X^N_0)^2+ (  X^{\hat u^{i,N}}_0)^2+  \int_0^t   E\big[(\bar u^N_s)^2\big]ds  + \int_0^t E\big[( \hat u^{i,N}_s)^2\big]ds   \right)  +\wt C_2(T)\\
&\leq  \ol C_1(T)\left(  \int_0^t \sup_{s \leq r \leq T} E\big[(\bar u_r^N)^2\big]ds  + \int_0^t\sup_{s \leq r \leq T} E\big[( \hat u^{i,N}_r)^2\big]ds  \right) +  \ol C_2(T), 
\end{aligned} 
\ee
where we used \eqref{init-assum}, \eqref{x-br} and \eqref{y-br} in the second inequality and recall that $Y^{\hat u^N}_0=y$. 

Recall that $  \ol C_1(T), \ol C_2(T)$ are not depending on $(N,i,t)$, so we get from Gronwall's lemma that 
\bd  
\begin{aligned} 
&\sup_{N \geq 2} \sup_{i=1,...,N}\sup_{0\leq s \leq T} \big( E\left[(\hat{u}_{s}^{i,N})^{2}\right]  +  E\left[(\bar u^N_{s})^{2}\right]  \big) \\
&\leq   (1-g(T,N))^{-1} \ol C_2(T) e^{ (1-g(T,N))^{-1} \ol C_1(T) T}, 
\end{aligned} 
\ed
which concludes the proof.

\end{proof} 

\section{Proof of Theorem \ref{thm-con-long}} \label{sec-pf-long} 

In order to prove the desired convergence result stated in Theorem
\ref{thm-con-long}, we will derive \emph{alternative} representations
of the solutions to the FBSDE systems
in~\eqref{eq:mean-FBSDE},~\eqref{eq:MeanField_infAggregated},
\eqref{eq:NASHFBSDE*} and \eqref{eq:MeanFieldGameFBSDE_inf} presented
in Sections~\ref{sec:finitePlayerSol} and~\ref{sec:infinitePlayerSol},
which are also of independent interest; see
Propositions~\ref{prop:FBSDE1},~\ref{prop:FBSDE2},~\ref{prop:FBSDE3},~\ref{prop:FBSDE4} below.

\paragraph{Notation.} In the following, for any matrix $A
= (a_{ij})_{1 \leq i \leq n, 1 \leq j \leq m} \in \mathbb{R}^{n \times
m}$ we denote $|A| = (|a_{ij}|)_{1 \leq i \leq n, 1 \leq j \leq m}$.

\subsection{The finite player aggregated FBSDE in
  \eqref{eq:mean-FBSDE}}

We fix an integer $N \geq1$ and define the matrices
\begin{equation} \label{mat-def} 
\ol F_{11} = 
\begin{pmatrix}
0 & 0 \\
0 & -\rho
\end{pmatrix}, \;
\ol F_{12} = 
\begin{pmatrix}
-1 & 0 \\
\gamma & 0
\end{pmatrix}, \;
\ol F_{21} = 
\begin{pmatrix}
-\frac{\phi}{\lambda} & \frac{\kappa\rho}{2\lambda} \\
0 & 0
\end{pmatrix}, \;
\ol F^N_{22} = 
\begin{pmatrix}
-\frac{\gamma\kappa(N-1)}{2\lambda N} & \frac{\rho}{2\lambda} \\
\frac{\gamma\kappa}{N} & \rho
\end{pmatrix}.
\end{equation}
Recall that $\ol F_{00}$ was defined in \eqref{b-00}. 

As the first step, let $W^N \in C^1([0,T],\mathbb{R}^{2\times 2})$ be a solution of the following Riccati differential equation
\begin{equation} \label{eq:riccati1}
    \dot W^N_t - \ol F_{21} - \ol F_{22}^N W^N_t + W^N_t \ol F_{11} + W^N_t \ol F_{12} W^N_t = 0, \quad W^N_T = \ol F_{00}.
\end{equation}

\noindent In fact, the solution $W^N$ of \eqref{eq:riccati1} can be
characterized via the solution to following matrix-valued linear
system, where $P^N, Q^N \in C^1([0,T],\mathbb{R}^{2\times 2})$ are such that
\begin{equation} \label{eq:matrixODE1}
    \frac{d}{dt}\begin{pmatrix} Q^N_t \\ P^N_t \end{pmatrix} =
    - \begin{pmatrix} \ol F_{11} & \ol F_{12} \\ \ol F_{21} &
      \ol F^N_{22} \end{pmatrix} \begin{pmatrix} Q^N_t \\
      P^N_t \end{pmatrix} \qquad (0 \leq t \leq T)
\end{equation}
and
$\begin{pmatrix} Q^N_0 \\ P^N_0 \end{pmatrix}= \begin{pmatrix} I \\
  \ol F_{00} \end{pmatrix}$. Moreover, the matrix-valued linear system in
\eqref{eq:matrixODE1} can be solved by computing a matrix
exponential. We specify these results in the following two lemmas.

\begin{lemma} \label{lem:matrixODE1}
Let 
\begin{equation} \label{eq:barFN}
    \ol F^N = \begin{pmatrix} \ol F_{11} & \ol F_{12} \\ \ol F_{21} & \ol F^N_{22} \end{pmatrix} \in \mathbb{R}^{4\times 4}.
\end{equation}
Then
\begin{equation} \label{eq:solmatrixODE1}  
    \ol K^N(t) = \exp(-\ol F^N \cdot t) \cdot \begin{pmatrix} I \\ \ol F_{00} \end{pmatrix} \in \mathbb{R}^{4\times 2}
\end{equation}
solves \eqref{eq:matrixODE1}.
\end{lemma}

\begin{lemma} \label{lem:riccati1}
Let $P^N, Q^N$ satisfy \eqref{eq:matrixODE1} and assume that 
\begin{equation} \label{q-aasump1} 
\liminf_{N \geq 1}\inf_{t \in [0,T]} |\det(Q^N_t)| > 0. 
\end{equation} 
Then the following holds: 
\begin{itemize} 
\item[\textbf{(i)}] \begin{equation} \label{eq:solriccati1}
    W^N_t = P^N_{T-t}(Q^N_{T-t})^{-1}  \quad  (0\leq t \leq T)
\end{equation}
solves \eqref{eq:riccati1};  
\item[\textbf{(ii)}]\begin{equation} \label{unif-w1} 
\sup_{N \geq 1} \sup_{t \in [0,T]} |W^N_t | < \infty.
\end{equation} 
\end{itemize} 
\end{lemma}

\begin{proof} (i) For the sake of readability we suppress the
  dependence on $N$. Computing the derivative
  in~\eqref{eq:solriccati1}, we obtain   
\begin{equation*}
    \begin{aligned}
       \dot W_t = & \, \frac{d}{dt}\left(P_{T-t}Q_{T-t}^{-1}\right) = \left(\frac{d}{dt}P_{T-t}\right) Q_{T-t}^{-1} + P_{T-t}\left(\frac{d}{dt}Q_{T-t}^{-1}\right) \\
       = & \, \left(\frac{d}{dt}P_{T-t}\right) Q_{T-t}^{-1} -  P_{T-t} Q_{T-t}^{-1} \left(\frac{d}{dt}Q_{T-t}\right)Q_{T-t}^{-1} \\
       = & \, \left(\ol F_{21}Q_{T-t}+\ol F^N_{22}P_{T-t}\right) Q_{T-t}^{-1} - P_{T-t} Q_{T-t}^{-1} \left(\ol F_{11}Q_{T-t}+\ol F_{12}P_{T-t}\right)Q_{T-t}^{-1} \\ = & \, \ol F_{21} + \ol F^N_{22} W_t - W_t \ol F_{11} - W_t \ol F_{12} W_{t}
    \end{aligned}
\end{equation*}
and hence (i).

(ii) From \eqref{q-aasump1} and the explicit formula for the inverse of $2\times 2$ matrices, also known as the cofactor equation, it follows that 
$$
\sup_{N \geq 1} \sup_{t \in [0,T]} |(Q^N_t)^{-1} | < \infty. 
$$
Moreover, from Lemma \ref{lem:matrixODE1} together with \eqref{mat-def} we get
$$
\sup_{N \geq 1} \sup_{t \in [0,T]} |P^N_t | < \infty, 
$$
and thus, using with \eqref{eq:solriccati1}, we obtain (ii).  

\end{proof}

\begin{remark} 
Note that from Lemmas \ref{lem:matrixODE1} and \ref{lem:riccati1} the existence of a classical solution $(W^N_t)_{t\in [0,T]}$ to \eqref{eq:riccati1} follows under the assumption $\det(Q^N_t)\neq 0$ for all $t \in [0,T]$. 
\end{remark}

Next, 
let $(f^N_t)_{t \in [0,T]}$ be the $L^2$ solution to the two-dimensional linear BSDE 
\begin{equation} \label{eq:BSDE1}
    df^N_t = (\ol F_{22}^N - W^N_t \ol F_{12}) f^N_t dt + \begin{pmatrix}
    \frac{dP_t}{2\lambda} + dM^{1,N}_t \\
    dM^{2,N}_t
    \end{pmatrix}, \quad f^N _T = \begin{pmatrix}
    0 \\ 0
    \end{pmatrix},
\end{equation}
for some square integrable martingales $M^{1,N},M^{2,N}$. Using
$W^N$ satisfying~\eqref{eq:riccati1} and $f^N$
satisfying~\eqref{eq:BSDE1} define the process
\begin{equation} \label{def:Y1}
    \mathbb{Y}^N_t := W^N_t \mathbb{X}^N_t + f^N_t \qquad (0 \leq t
    \leq T),
\end{equation}
where $\mathbb{X}^N$ is the solution to the linear ODE with random coefficients 
\begin{equation} \label{eq:SDEX1}
    d\mathbb{X}^N_t = (\ol F_{11} \mathbb{X}^N_t + \ol F_{12} (W^N_t \mathbb{X}^N_t + f^N_t)) dt = (\ol F_{11} \mathbb{X}^N_t + \ol F_{12} \mathbb{Y}^N_t) dt, \quad \mathbb{X}^N_0 = \begin{pmatrix}
    \bar x^N \\ y
    \end{pmatrix}
\end{equation}
and initial value $\bar x^N$ as in \eqref{eq:mean-FBSDE}.

\begin{proposition} \label{prop:FBSDE1} 
The process $
\begin{pmatrix}
    \mathbb{X}^N , \mathbb{Y}^N
    \end{pmatrix}^\top
    $
     from \eqref{def:Y1}--\eqref{eq:SDEX1} satisfies the finite player aggregated FBSDE system in \eqref{eq:mean-FBSDE}, i.e., using the notation from \eqref{eq:mean-FBSDE} we have 
\begin{equation} \label{sol:FBSDE1}
\mathbb{Y}^N_t = 
\begin{pmatrix}
\bar u^N_t \\ \overline Z^{\bar u^N}_t
\end{pmatrix} \qquad \text{and} \qquad
\mathbb{X}^N_t = 
\begin{pmatrix}
\overline X^{\bar{u}^N}_t \\ \overline Y^{\bar{u}^N}_t
\end{pmatrix}.
\end{equation}
\end{proposition} 

\begin{proof} For the sake of readability we suppress the dependence on $N$ throughout the proof. 
Applying It\^o's formula in~\eqref{def:Y1}, using~\eqref{eq:SDEX1}, \eqref{eq:riccati1} and \eqref{eq:BSDE1}, we obtain for all $t \in [0,T]$, 
\begin{equation} \label{eq:SDEY}
    \begin{aligned}
d\mathbb{Y}_t = &  \, \dot W_t \mathbb{X}_t dt + W_t d\mathbb{X}_t + df_t = \, \dot W_t \mathbb{X}_t dt + W_t (\ol F_{11} \mathbb{X}_t + \ol F_{12} (W_t \mathbb{X}_t + f_t)) dt + df_t \\
= & \, (\dot W_t + W_t \ol F_{11} + W_t \ol F_{12} W_t) \mathbb{X}_t dt + W_t \ol F_{12} f_t dt + df_t \\
= & \, (\ol F_{21}+\ol F_{22}^N W_t) \mathbb{X}_t dt + W_t \ol F_{12} f_t dt + (\ol F_{22}^N - W_t \ol F_{12}) f_t dt + \begin{pmatrix}
    \frac{dP_t}{2\lambda} + dM^1_t \\
    dM^2_t
    \end{pmatrix} \\
    = & \, \ol F_{21} \mathbb{X}_t dt + \ol F_{22}^N (W_t \mathbb{X}_t + f_t) dt + \begin{pmatrix}
    \frac{dP_t}{2\lambda} + dM^1_t \\
    dM^2_t
    \end{pmatrix} \\
    = & \, \ol F_{21} \mathbb{X}_t dt 
    + \ol F_{22}^N \mathbb{Y}_t dt + \begin{pmatrix}
    \frac{dP_t}{2\lambda} + dM^1_t \\
    dM^2_t
    \end{pmatrix}.
    \end{aligned}
\end{equation}
Therefore, equations~\eqref{eq:SDEX1} and~\eqref{eq:SDEY} yield 
\begin{equation} \label{eq:SDEXY}
\begin{aligned}
    d\begin{pmatrix}
    \mathbb{X}_t \\ \mathbb{Y}_t
    \end{pmatrix} = & \, 
    \begin{pmatrix}
    \ol F_{11} &  \ol F_{12} \\
    \ol F_{21} &  \ol F_{22}^N
    \end{pmatrix} 
      \begin{pmatrix}
    \mathbb{X}_t \\ \mathbb{Y}_t
    \end{pmatrix} dt + \begin{pmatrix}
    0 \\ 0 \\ \frac{dP_t}{2\lambda} + dM^1_t \\
    dM^2_t
    \end{pmatrix} \\
    = & \, \begin{pmatrix}
    0 & 0 & -1 & 0 \\
    0 & -\rho & \gamma & 0 \\
    -\frac{\phi}{\lambda} & \frac{\kappa\rho}{2\lambda} & -\frac{\gamma\kappa(N-1)}{2\lambda N} & \frac{\rho}{2\lambda} \\
    0 & 0 & \frac{\kappa\gamma}{N} & \rho
    \end{pmatrix}
    \begin{pmatrix}
    \mathbb{X}^1_t \\
    \mathbb{X}^2_t \\
    \mathbb{Y}^1_t \\
    \mathbb{Y}^2_t
    \end{pmatrix} dt
    + \begin{pmatrix}
    0 \\ 0 \\ \frac{dP_t}{2\lambda} + dM^1_t \\
    dM^2_t
    \end{pmatrix},
    \end{aligned}
\end{equation}
where $\mathbb{Y}_T = W_T \mathbb{X}_T + f_T = \ol F_{00} \mathbb{X}_T = \begin{pmatrix}
\frac{\varrho}{\lambda} \mathbb{X}^1_T - \frac{\kappa}{2\lambda} \mathbb{X}^2_T \\ 0
\end{pmatrix}$. In other words, the process $(\mathbb{X}, \mathbb{Y})$ from~\eqref{eq:SDEX1} and~\eqref{def:Y1}, i.e.,
$
    \begin{pmatrix} \mathbb{X}_t, \mathbb{Y}_t \end{pmatrix}^\top
    =     \begin{pmatrix} \mathbb{X}^1_t , \mathbb{X}^2_t , \mathbb{Y}^1_t , \mathbb{Y}^2_t \end{pmatrix}^\top
$
satisfies the finite player aggregated FBSDE system in \eqref{eq:mean-FBSDE}.
\end{proof}



\subsection{Convergence of \eqref{eq:mean-FBSDE} to the Mean Field
  FBSDE \eqref{eq:MeanField_infAggregated}}


Now, we first repeat the analysis above to get an alternative
representation for the solution to the infinite-player mean field
FBSDE system in \eqref{eq:MeanField_infAggregated}. Then, using the
latter, we show that it is indeed the limit of the solution of
\eqref{eq:mean-FBSDE} derived in Proposition~\ref{prop:FBSDE1} above as $N$
goes to infinity.

To this end, introduce the matrix
\begin{equation*}
    \ol F^{\infty}_{22} := \lim_{N\rightarrow\infty} \ol F_{22}^N = \begin{pmatrix}
-\frac{\gamma\kappa}{2\lambda} & \frac{\rho}{2\lambda} \\
0 & \rho 
\end{pmatrix}
\end{equation*}

\noindent and, instead of~\eqref{eq:riccati1}, let $W \in C^1([0,T],\mathbb{R}^{2\times 2})$ be the solution of the Riccati differential equation
\begin{equation} \label{eq:riccati2}
    \dot{W}_t - \ol F_{21} - \ol F_{22}^{\infty} W_t + W_t \ol F_{11} + W_t \ol F_{12} W_t = 0, \quad W_T = \ol F_{00}.
\end{equation}

\noindent Similar to above, the solution $W$ of \eqref{eq:riccati2} can be characterized via the solution of a matrix-valued linear system.  

\begin{lemma} \label{lem:riccati2}
Let $ P, Q \in C^1([0,T],\mathbb{R}^{2\times 2})$ such that
\begin{equation} \label{eq:matrixODE2}
    \frac{d}{dt}\begin{pmatrix} Q_t \\ P_t \end{pmatrix} =
    - \begin{pmatrix} \ol F_{11} & \ol F_{12} \\ \ol F_{21} &
      \ol F^{\infty}_{22} \end{pmatrix} \begin{pmatrix} Q_t \\
      P_t \end{pmatrix} \qquad (0 \leq t \leq T)
\end{equation}
and $\begin{pmatrix} Q_0 \\ P_0 \end{pmatrix}= \begin{pmatrix} I \\ \ol F_{00} \end{pmatrix}$. Assume that 
\begin{equation} \label{q-aasump2} 
\inf_{t \in [0,T]} |\det(Q_t)| > 0.
\end{equation}   Then the following holds: 
\begin{itemize} 
\item[\textbf{(i)}] 
\begin{equation}
    W_t = P_{T-t} Q_{T-t}^{-1}, \quad 0\leq t \leq T, 
\end{equation}
solves \eqref{eq:riccati2};
\item[\textbf{(ii)}] 
$$
\sup_{t\in[0,T]} |W_t|<\infty. 
$$
\end{itemize} 
\end{lemma}

\begin{proof}
The proof follows the same lines as the proof of
Lemma~\ref{lem:riccati1} above.
\end{proof}

Again, as above, the matrix-valued linear system in \eqref{eq:matrixODE2} can be solved via a matrix exponential.

\begin{lemma} \label{lem:matrixODE2} 
Let 
\begin{equation} \label{eq:barF}
    \ol F = \begin{pmatrix} \ol F_{11} & \ol F_{12} \\ \ol F_{21} & \ol F^\infty_{22} \end{pmatrix} \in \mathbb{R}^{4\times 4}
\end{equation}
Then
\begin{equation} \label{eq:solmatrixODE2}
    \ol K(t) = \exp(-\ol F \cdot t) \cdot \begin{pmatrix} I \\ \ol F_{00} \end{pmatrix} \in \mathbb{R}^{4\times 2}
\end{equation}
solves \eqref{eq:matrixODE2}.
\end{lemma}

The convergence of $W^N$ solving~\eqref{eq:riccati1} to $W$
solving~\eqref{eq:riccati2} as $N$ goes to infinity is readily given by 

\begin{lemma} \label{lemma-con-w1} 
Assume \eqref{q-aasump1} and \eqref{q-aasump2}. Let $W^N$ solve~\eqref{eq:riccati1} and $W$ solve~\eqref{eq:riccati2}. Then 
\begin{equation}
\lim_{N\rightarrow \infty} \sup_{t\in [0,T]} |W^N_t -  W_t|=0.
\end{equation}
\end{lemma}

\begin{proof}
Note that for any $N\geq 1$ and $t \in [0,T]$ we have 
\begin{equation} \label{q-eq1} 
Q_t^{-1} - (Q_t^N)^{-1}  =  (Q_t^N)^{-1}\big(Q_t^N - Q_t \big) (Q_t)^{-1}.
\end{equation} 
From \eqref{q-aasump1} and \eqref{q-aasump2} it follows that 
\begin{equation} \label{q-eq2} 
\limsup_{N \geq 1}\sup_{t\in [0,T]}|(Q_t^N)^{-1}|<\infty, \quad \sup_{t\in [0,T]}|(Q_t)^{-1}|<\infty. 
\end{equation} 
Since $\ol F^N$ in~\eqref{eq:barFN} converges to $\ol F$ in~\eqref{eq:barF} as
$N \rightarrow \infty$, it follows that $\ol K^N$
in~\eqref{eq:solmatrixODE1} converges to $\ol K$ in~\eqref{eq:solmatrixODE2} as $N \rightarrow \infty$. But this implies that the solution $Q^N,P^N$ of~\eqref{eq:matrixODE1} converges to the solution $Q, P$ of~\eqref{eq:matrixODE2} uniformly on $[0,T]$. Moreover, 
we get from \eqref{q-eq1} and  \eqref{q-eq2} that 
$$
\lim_{N\rightarrow \infty} \sup_{t\in [0,T]}|Q_t^{-1} - (Q_t^N)^{-1}|=0.
$$
Together with Lemma \ref{lem:riccati1}(i)  and Lemma \ref{lem:riccati2}(i) we get the result. 
\end{proof}

Next, 
let $(f_t)_{t\in [0,T]}$ be the $L^2$ solution to the two-dimensional linear BSDE 
\begin{equation} \label{eq:BSDE2}
    df_t = (\ol F_{22}^{\infty} - W_t \ol F_{12}) f_t dt + \begin{pmatrix}
    \frac{dP_t}{2\lambda} + dM^1_t \\
    dM^2_t
    \end{pmatrix}, \quad f_T = \begin{pmatrix}
    0 \\ 0
    \end{pmatrix}
\end{equation}
for some square integrable martingales $M^1, M^2$. The convergence of $f^N$ solving \eqref{eq:BSDE1} to $f$ solving
\eqref{eq:BSDE2} as $N$ tends to infinity is established in

\begin{proposition} \label{prop-f-con1} 
Let $f^N$ be the solution to \eqref{eq:BSDE1} and $f$ be the solution to \eqref{eq:BSDE2}. Then we have 
$$
\lim_{N\rightarrow \infty} \sup_{t\in [0,T]} |f^N_t- f_t|  =0, \quad P-\textrm{a.s.} 
$$
\end{proposition}

\begin{proof} 
From \eqref{mat-def} and Lemma \ref{lemma-con-w1} it follows that 
$$
\lim_{N\rightarrow \infty} \sup_{t\in [0,T]} |\ol F_{22}^{N} - W^N_t \ol F_{12}- (\ol F_{22}^{\infty} - W_t \ol F_{12}) | =0.  
$$
As a consequence, using standard stability results for linear BDSEs (see, e.g., Theorem 2 in \cite{BAHLALI}), we get 
\begin{equation}  \label{bla} 
\lim_{N\rightarrow \infty}E\left[\sup_{t\in [0,T]} |f^N_t- f_t| \right]=0. 
\end{equation}
The result then follows by applying Fatou's lemma and using \eqref{bla}. 
 \end{proof} 
 
 Finally, using $W$ satisfying~\eqref{eq:riccati2} and $f$
 satisfying~\eqref{eq:BSDE2}, define similar to~\eqref{def:Y1} above
 the process
\begin{equation} \label{def:Y2}
    \mathbb{Y}_t := W_t \mathbb{X}_t + f_t \qquad (0 \leq t \leq T),
\end{equation}
where $\mathbb{X}$ is the solution to the linear ODE with random coefficients
\begin{equation} \label{eq:SDEX2}
    d\mathbb{X}_t = (\ol F_{11} \mathbb{X}_t + \ol F_{12} (W_t \mathbb{X}_t + f_t)) dt = (\ol F_{11} \mathbb{X}_t + \ol F_{12} \mathbb{Y}_t) dt, \quad \mathbb{X}_0 = \begin{pmatrix}
    \tilde x \\ y
    \end{pmatrix},
\end{equation}
and initial value $\tilde x$ as in \eqref{eq:MeanField_infAggregated}. 

\begin{proposition} \label{prop:FBSDE2} 
The process $
\begin{pmatrix} 
    \mathbb{X} , \mathbb{Y}
    \end{pmatrix}^\top
    $ in \eqref{def:Y2}--\eqref{eq:SDEX2}
     satisfies the infinite player mean field FBSDE system
     in~\eqref{eq:MeanField_infAggregated}, i.e., using the notation
     from \eqref{eq:MeanField_infAggregated} we have
\begin{equation} \label{sol:FBSDE2}
\mathbb{Y}_t = 
\begin{pmatrix}
\tilde \nu_t \\ 0
\end{pmatrix} \qquad \text{and} \qquad
\mathbb{X}_t = 
\begin{pmatrix}
\tilde X^{\tilde{\nu}}_t \\ \tilde Y^{\tilde{\nu}}_t. 
\end{pmatrix}.
\end{equation}
 \end{proposition} 

 The proof of Proposition \ref{prop:FBSDE2} is the
  same as the proof of Proposition \ref{prop:FBSDE1} and hence omitted.


\subsection{Proof of Theorem \ref{thm-con-long}(i)}

We introduce the following norms: For any integer $d \geq 1$ and any
vector $\mathbf{x}=(x_1,...,x_d) \in \mathbb{R}^d$ we define
$$
\|\mathbf x\|_{\ell_1} = \sum_{i=1}^d |x_i|;
$$ 
for any matrix $C=(c_{ij})_{1 \leq i,j \leq d} \in
\mathbb{R}^{d \times d}$ we define 
$$
\|C\|_{\textrm{max}} = \max_{1\leq i,j \leq d} |c_{ij}|.
$$
\begin{proof}[Proof of Theorem \ref{thm-con-long}(i)] 
Recall that from assumptions \eqref{init-assum} and
\eqref{eq:initPosLimit} we have $X^{\hat{u}^{i,N}}_0 =X^{\hat v^i}_0
= x^{i}$ for all $i=1,\ldots,N$, $N\geq 1$, and $\bar x^N \rightarrow
\tilde x$ as $N \rr \infty$. One can easily verify from \eqref{eq:matrixODE1}--\eqref{eq:solmatrixODE1}, \eqref{eq:matrixODE2}, \eqref{eq:solmatrixODE2} 
that \eqref{assump-k} implies \eqref{q-aasump1} and \eqref{q-aasump2}. We will show that 
 \be \label{l1-conv} 
\lim_{N \rightarrow \infty}  \sup_{t\in [0,T]} \big(|\mathbb{X}^N_t - \mathbb{X}_t|  + |\mathbb{Y}^N_t - \mathbb{Y}_t| \big)=0 \quad \textrm{a.s.}, 
 \ee
 which together with \eqref{sol:FBSDE1} and \eqref{sol:FBSDE2} will imply the claim in Theorem \ref{thm-con-long}(i). 
 
 From \eqref{eq:SDEX1} and \eqref{eq:SDEX2} we have 
\begin{equation} \label{gr0} 
\begin{aligned} 
 |\mathbb{X}^N_t-  \mathbb{X}_t| & \leq |\bar x^N -\tilde x| + \int_0^T |\ol F_{11}| | \mathbb{X}^N_s  - \mathbb{X}_s | ds  \\
 & \quad + \int_{0}^t |\ol F_{12}||  W^N_s \mathbb{X}^N_s -W_s \mathbb{X}_s| ds + \int_0^t|\ol F_{12}| |f_s^N- f_s | ds  \\
&\leq   |\bar x^N -\tilde x| + \int_0^T |\ol F_{11}|| \mathbb{X}^N_s  - \mathbb{X}_s | ds  \\
 & \quad + \int_{0}^t |\ol F_{12}||  W^N_s \mathbb{X}^N_s -  W^N_s \mathbb{X}_s| ds+ \int_{0}^t |\ol F_{12}| |  W^N_s \mathbb{X}_s - W_s \mathbb{X}_s| ds \\
 &\quad + \int_0^t|\ol F_{12}| |f_s^N- f_s | ds.  
 \end{aligned} 
\end{equation} 
Multiplying both sides of \eqref{gr0} by $(1,1)$ from the left, we get 
\begin{equation} \label{gr1} 
\begin{aligned} 
 \|\mathbb{X}^{N}_t -  \mathbb{X}_t\|_{\ell_1}   
&\leq  \bigg( \|\bar x^N -\tilde x\|_{\ell_1} +2 \|\ol F_{11}\|_{\textrm{max}}  \int_0^T \| \mathbb{X}^N_s  - \mathbb{X}_s \|_{\ell_1} ds  \\
 & \quad +4\|\ol F_{12}\|_{\textrm{max}} \sup_{N \geq 1}\sup_{ t\in [0,T]} \|W^N_t\|_{\textrm{max}}  \int_{0}^t \|  \mathbb{X}^N_s -  \mathbb{X}_s\|_{\ell_1} ds \\
 &\quad + 8\|\ol F_{12}\|_{\textrm{max}}  \sup_{t \in [0,T]} \| \mathbb{X}_t \|_{\ell_1}\int_{0}^t  \| W^N_s - W_s \|_{\textrm{max}} ds \\
 &\quad +2\|\ol F_{12}\|_{\textrm{max}} \int_0^t \|f_s^N- f_s \|_{\ell_1} ds.  
 \end{aligned} 
\end{equation} 

Now, let $\varepsilon>0$ be arbitrary small. From
Lemma~\ref{lemma-con-w1}, Proposition~\ref{prop-f-con1} and
\eqref{gr1} it follows that for all $N$ sufficiently large  
\begin{equation} \label{gr2} 
\begin{aligned} 
 \sup_{s \in [0,t]}  \|\mathbb{X}^{N}_s-  \mathbb{X}^{}_s\|_{\ell_1}   
&\leq  \varepsilon +C_1   \int_0^T  \sup_{r\in [0,s] }\| \mathbb{X}^N_r  - \mathbb{X}_r \|_{\ell_1} ds  \\
  &\quad +  T C_2 \varepsilon, 
 \end{aligned} 
\end{equation} 
where $C_i>0$ are constants not depending on $N$ and $t$. It then follows from Gronwall's lemma that there exists a constant $C(T)>0$ such that 
\begin{equation} \label{gr3} 
 \sup_{s \in [0,T]}  \|\mathbb{X}^{N}_s-  \mathbb{X}^{}_s\|_{\ell_1}   \leq C(T)\varepsilon. 
\end{equation} 
Since $\varepsilon>0$ is arbitrary small, the first part in~\eqref{l1-conv} follows from \eqref{gr3}. 

Finally, using Lemma~\ref{lemma-con-w1}, Proposition~\ref{prop-f-con1}, \eqref{def:Y1}, \eqref{def:Y2} and \eqref{gr3} we get  
\begin{equation} \label{gr4} 
\lim_{N \rightarrow \infty}  \sup_{s \in [0,T]}  \|\mathbb{Y}^{N}_s-  \mathbb{Y}^{}_s\|_{\ell_1}    =0, 
\end{equation} 
which together with \eqref{gr3} completes the proof.
\end{proof}



\subsection{Proof of Theorem \ref{thm-con-long}(ii)} 

The proof of Theorem \ref{thm-con-long}(ii) follows the same rationale
as the proof of part (i) presented above. Therefore, we only provide
the outline in order to avoid repetition. Specifically, the idea is
once more to derive alternative representations for the solutions to
the FBSDEs in \eqref{eq:NASHFBSDE*} and
\eqref{eq:MeanFieldGameFBSDE_inf}, and then to argue that the former
converges to the latter as $N$ goes to infinity.

We start with fixing an integer $N \geq 1$ and defining the matrices
\begin{equation*}
\hat F_{12} = 
\begin{pmatrix}
-1 & 0 \\
0 & 0
\end{pmatrix}, \;
\hat F_{21} = 
\begin{pmatrix}
-\frac{\phi}{\lambda} & 0 \\
0 & 0
\end{pmatrix}, \;
\hat F^N_{22} = 
\begin{pmatrix}
\frac{\gamma\kappa}{2\lambda N} & \frac{\rho}{2\lambda} \\
\frac{\gamma\kappa}{N} & \rho
\end{pmatrix}.
\end{equation*}
Recall that $\hat F_{00}$ was defined in \eqref{b-00}. 
Let $\wt W^N \in C([0,T],\mathbb{R}^{2\times 2})$ be a solution to the following Riccati differential equation
\begin{equation} \label{eq:riccati3}
    \dot{ \wt W}^N_t - \hat F_{21} - \hat F_{22}^N \wt W^N_t + \wt W^N_t \hat F_{12} \wt W^N_t = 0, \quad \wt W^N_T = \hat F_{00}.
\end{equation}

\noindent Similar to above, the solution $\wt W^N$ of \eqref{eq:riccati3} can be
characterized via the solution to following matrix-valued linear
system, where $\wt P^N,\wt Q^N \in C^1([0,T],\mathbb{R}^{2\times 2})$
are such that
\begin{equation} \label{eq:matrixODE3}
    \frac{d}{dt}\begin{pmatrix} \wt Q^N_t \\ \wt P^N_t \end{pmatrix} =
    - \begin{pmatrix} 0 & \hat F_{12} \\ \hat F_{21} &
      \hat F^N_{22} \end{pmatrix} \begin{pmatrix} \wt Q^N_t \\ \wt
      P^N_t \end{pmatrix} \qquad (0 \leq t \leq T)
\end{equation}
and $\begin{pmatrix} \wt Q^N_0 \\ \wt
  P^N_0 \end{pmatrix}= \begin{pmatrix} I \\
  \hat F_{00} \end{pmatrix}$. Moreover, the matrix-valued linear system in
\eqref{eq:matrixODE3} can again be
solved via a matrix exponential. We collect these results in the
following two lemmas which are the counterparts to Lemmas
\ref{lem:riccati2} and \ref{lem:matrixODE2} above.

\begin{lemma} \label{lem:matrixODE3}
Let 
\begin{equation} \label{eq:GN}
    \hat F^N = \begin{pmatrix} 0 & \hat F_{12} \\\hat F_{21} & \hat F^N_{22} \end{pmatrix} \in \mathbb{R}^{4\times 4}.
\end{equation}
Then
\begin{equation} \label{eq:solmatrixODE3}
    \hat K^N(t) = \exp(-\hat F^N \cdot t) \cdot \begin{pmatrix} I \\ \hat F_{00} \end{pmatrix} \in \mathbb{R}^{4\times 2}
\end{equation}
solves \eqref{eq:matrixODE3}.
\end{lemma}

\begin{lemma} \label{lem:riccati3}
Let $\wt P^N, \wt  Q^N$ satisfy \eqref{eq:matrixODE3} and assume that 
\begin{equation} \label{q-aasump3} 
\liminf_{N \geq 1}\inf_{t \in [0,T]} |\det(\wt Q^N_t)| > 0. 
\end{equation} 
Then the following holds: 
\begin{itemize} 
\item[\textbf{(i)}] \begin{equation} \label{eq:solriccati3}
    \wt W^N_t = \wt P^N_{T-t}(\wt Q^N_{T-t})^{-1}  \quad  (0\leq t \leq T),
\end{equation}
solves \eqref{eq:riccati3};  
\item[\textbf{(ii)}]\begin{equation*} \label{unif-w2} 
\sup_{N \geq 1} \sup_{t \in [0,T]} |\wt W^N_t | < \infty.
\end{equation*} 
\end{itemize} 
\end{lemma}

Next,  let $\wt Y^N := \mathbb{X}^{N,2}$ denote the second component
of the solution to the linear SDE in~\eqref{eq:SDEX1}, which is given
in~\eqref{sol:FBSDE1}. In addition, 
let $(g^N_t)_{t \in [0,T]}$ be the $L^2$ solution to the two-dimensional linear BSDE 
\begin{equation} \label{eq:BSDE3}
    dg^N_t = (\hat F_{22}^N - \wt W^N_t \hat F_{12}) g^N_t dt + \begin{pmatrix}
    \frac{d(P_t-\kappa \wt Y^N_t)}{2\lambda} + d \wt M^{N}_t \\
    0
    \end{pmatrix}, \quad g^N _T = \begin{pmatrix}
    -\frac{\kappa}{2\lambda} \wt Y^N_T \\ 0
    \end{pmatrix}
\end{equation}
for some square integrable martingale $\wt M$. Using $\wt W^N$ satisfying~\eqref{eq:riccati3} and $g^N$ satisfying~\eqref{eq:BSDE3} define for each $i \in \{1, \ldots, N\}$ the process
\begin{equation} \label{def:Y3}
  \wt{\mathbb{Y}}^{i,N}_t := \wt W^N_t \wt{\mathbb{X}}^{i,N}_t + g^N_t \qquad (0\leq t \leq T),
\end{equation}
where $\wt{\mathbb{X}}^{i,N}$ is the solution to the linear ODE with
random coefficients
\begin{equation} \label{eq:SDEX3}
    d\wt{\mathbb{X}}^{i,N}_t = (\hat F_{12} (\wt{W}^N_t \wt{\mathbb{X}}^{i,N}_t + g^N_t)) dt = (\hat F_{12}\wt{\mathbb{Y}}^{i,N}_t) dt, \quad \wt{\mathbb{X}}^{i,N}_0 = \begin{pmatrix}
    x^{i,N} \\ 0
    \end{pmatrix}
\end{equation}
and initial value $x^{i,N}$ as in \eqref{eq:NASHFBSDE*}. 

\begin{proposition} \label{prop:FBSDE3}
  The process
$
\begin{pmatrix} 
   \wt{\mathbb{X}}^{i,N} , \wt{\mathbb{Y}}^{i,N}
    \end{pmatrix}^\top
    $
     in~\eqref{def:Y3}--\eqref{eq:SDEX3} satisfies the $i$-th player FBSDE system in
     \eqref{eq:NASHFBSDE*}, i.e., using the notation from
     \eqref{eq:NASHFBSDE*} we have
\begin{equation*} \label{sol:FBSDE3}
\wt{\mathbb{Y}}^{i,N}_t = 
\begin{pmatrix}
u^{i,N}_t \\ Z^{u^{i,N}}_t
\end{pmatrix} \qquad \text{and} \qquad
\wt{\mathbb{X}}^{i,N}_t = 
\begin{pmatrix}
X^{u^{i,N}}_t \\ 0
\end{pmatrix}.
\end{equation*}
\end{proposition}

The proof of Proposition~\ref{prop:FBSDE3} follows the same lines as
the proof of Proposition~\ref{prop:FBSDE1} and is thus omitted.

Finally, as in the proof of Theorem \ref{thm-con-long}(i) above, we
also derive and alternative representation of the solution to
\eqref{eq:MeanFieldGameFBSDE_inf} in order to prove the desired
convergence. To achieve this, we introduce the matrix
\begin{equation}
    \hat F^{\infty}_{22} := \lim_{N\rightarrow\infty} \hat F_{22}^N = \begin{pmatrix}
0 & \frac{\rho}{2\lambda} \\
0 & \rho
\end{pmatrix}
\end{equation}
and, instead of~\eqref{eq:riccati3}, let $ \wt W \in C([0,T],\mathbb{R}^{2\times 2})$ be the solution to the following Riccati differential equation
\begin{equation} \label{eq:riccati4}
    \dot{\wt W}_t - \hat F_{21} - \hat F_{22}^\infty \wt W_t +\wt W_t \hat F_{12} \wt W_t = 0, \quad \wt W_T = \hat F_{00}.
\end{equation}
Again, $\wt W$ solving \eqref{eq:riccati4} can be characterized by
the solution of a matrix-valued linear system, which in turn can be 
solved via a matrix exponential.

\begin{lemma} \label{lem:riccati4}
Let $\wt P,\wt Q \in C^1([0,T],\mathbb{R}^{2\times 2})$ such that
\begin{equation} \label{eq:matrixODE4}
    \frac{d}{dt}\begin{pmatrix} \wt Q_t \\ \wt P_t \end{pmatrix} = - \begin{pmatrix} 0 & \hat F_{12} \\ \hat F_{21} & \hat F^{\infty}_{22} \end{pmatrix} \begin{pmatrix} \wt Q_t \\ \wt P_t \end{pmatrix} \qquad \text{for all } t \in [0,T]
\end{equation}
and $\begin{pmatrix} \wt Q_0 \\ \wt P_0 \end{pmatrix}= \begin{pmatrix} I \\ \hat F_{00} \end{pmatrix}$. 
 Assume that 
\begin{equation} \label{q-aasump4} 
\inf_{t \in [0,T]} |\det(\wt Q_t)| > 0. 
\end{equation} 
Then
\begin{equation*}
    \wt W_t =\wt P_{T-t} \wt  Q_{T-t}^{-1}
\end{equation*}
solves \eqref{eq:riccati4}.
\end{lemma}
 

\begin{lemma} \label{lem:matrixODE4}
Let 
\begin{equation} \label{eq:F}
    \hat F = \begin{pmatrix} 0 & \hat F_{12} \\ \hat F_{21} & \hat F^\infty_{22} \end{pmatrix} \in \mathbb{R}^{4\times 4}
\end{equation}
Then
\begin{equation} \label{eq:solmatrixODE4}
   \hat K(t) = \exp(-\hat F \cdot t) \cdot \begin{pmatrix} I \\ \hat F_{00} \end{pmatrix} \in \mathbb{R}^{4\times 2}
\end{equation}
solves \eqref{eq:matrixODE4}.
\end{lemma}
 The convergence of $\wt W^N$ solving~\eqref{eq:riccati3} to $\wt W$ solving~\eqref{eq:riccati4} as $N\rightarrow \infty$ is given by 
\begin{lemma} \label{lemma-con-w2} 
Assume \eqref{q-aasump3} and \eqref{q-aasump4}. Let $\wt W^N$
solve~\eqref{eq:riccati3} and $\wt W$ solve~\eqref{eq:riccati4}. Then 
\begin{equation}
\lim_{N\rightarrow \infty} \sup_{t\in [0,T]} |\wt W^N_t - \wt W_t|=0.
\end{equation}
\end{lemma}

\begin{proof}
Follows as in the proof of Lemma~\ref{lemma-con-w1}.
\end{proof}

 Let $\wt  Y := \mathbb{X}^2$ be the second component of the solution
 of the linear SDE in~\eqref{eq:SDEX2}, which is given
 in~\eqref{sol:FBSDE2}, 
and let $(g_t)_{t \in [0,T]}$ be the $L^2$ solution to the two-dimensional linear BSDE 
\begin{equation} \label{eq:BSDE4}
    dg_t = (\hat F_{22}^\infty -\wt  W_t \hat F_{12}) g_t dt + \begin{pmatrix}
    \frac{d(P_t - \kappa \wt  Y_t)}{2\lambda} + d\wt M_t \\
    0
    \end{pmatrix}, \quad g_T = \begin{pmatrix}
    -\frac{\kappa}{2\lambda} \wt Y_T \\ 0
    \end{pmatrix},
\end{equation}
for some square integrable martingale $\wt M$.  

\begin{proposition} \label{prop-f-con2} 
Let $g^N$ be the solution to \eqref{eq:BSDE3} and $g$ be the solution to \eqref{eq:BSDE4}. Then we have 
$$
\lim_{N\rightarrow \infty} \sup_{t\in [0,T]} |g^N_t- g_t|  =0, \quad P-\textrm{a.s.} 
$$
\end{proposition} 
The proof of Proposition \ref{prop-f-con2} follows the same reasoning
as the proof of Proposition~\ref{prop-f-con1}.
 
 Lastly, using $\wt W$ satisfying~\eqref{eq:riccati4} and $g$ satisfying~\eqref{eq:BSDE4}, define as above in~\eqref{def:Y3} for each $i \in \mathbb{N}$ the process
\begin{equation} \label{def:Y4}
   \wt{\mathbb{Y}}^i_t := \wt W_t \wt{\mathbb{X}}^i_t + g_t \qquad (0\leq t \leq T),
\end{equation}
where $\wt{\mathbb{X}}^i$ is the solution to the linear ODE with
random coefficients
\begin{equation} \label{eq:SDEX4}
    d\wt{\mathbb{X}}^i_t = \hat F_{12} (\wt W_t \wt{\mathbb{X}}^i_t + g_t) dt = \hat F_{12} \wt{\mathbb{Y}}^i_t dt, \quad \wt{\mathbb{X}}_0 = \begin{pmatrix}
    x^{i} \\ 0
    \end{pmatrix}
\end{equation}
and initial value $x^{i}$ as in \eqref{eq:MeanFieldGameFBSDE_inf}.


\begin{proposition} \label{prop:FBSDE4} 
The process $
\begin{pmatrix} 
    \wt{\mathbb{X}}^i , \wt{\mathbb{Y}}^i
    \end{pmatrix}^\top
    $
     in~\eqref{def:Y4}--\eqref{eq:SDEX4} satisfies the $i$-th player FBSDE system in \eqref{eq:MeanFieldGameFBSDE_inf}, i.e., using the notation from equation  \eqref{eq:MeanFieldGameFBSDE_inf},
\begin{equation} \label{sol:FBSDE4}
\wt{\mathbb{Y}}^i_t = 
\begin{pmatrix}
\hat v^i_t \\ 0
\end{pmatrix} \qquad \text{and} \qquad
\wt{\mathbb{X}}^i_t = 
\begin{pmatrix}
{X}^{\hat{v}^i}_t \\ 0
\end{pmatrix}.
\end{equation}
 \end{proposition}

We are now ready to proof Theorem \ref{thm-con-long}(ii). 
\begin{proof}[Proof of Theorem \ref{thm-con-long}(ii)] The proof of
  Theorem \ref{thm-con-long}(ii) is similar to part (i).  One can easily verify that \eqref{assump-k} implies \eqref{q-aasump3} and \eqref{q-aasump4}. Using
  Lemma \eqref{lemma-con-w2}, Proposition \ref{prop-f-con2} and
  Gronwall's Lemma we arrive at
 $$
\lim_{N \rightarrow \infty}  \sup_{t\in [0,T]} \big(|\wt{\mathbb{Y}}^{i,N}_t -\wt{ \mathbb{Y}}^i_t|  + |\wt{\mathbb{X}}^{i,N}_t - \wt{\mathbb{X}}^i_t| \big)=0 \quad \textrm{a.s.}
 $$
 and the result follows. 
 \end{proof}

\section{Proof of Theorem \ref{thm-eps-nash}}  \label{sec-eps-pf} 
We first introduce the following auxiliary lemma. 
\begin{lemma}  \label{lemma-con-mf} 
Let $\hat \nu$ and $\{\hat v^i, \, i\in \mathbb{N}\}$ be the equilibrium strategies of the mean field game as in Definition \ref{def:Nash_infMFG} and assume \eqref{init-assum} and \eqref{a-const}. 
Then we have  
$$
\sup_{t\in [0,T]} E\left[\left(\hat \nu_{t}-\frac{1}{N}\sum_{i=1}^N \hat v_t^i\right)^{2}\right]= O\left(\frac{1}{N^2}\right). 
$$
\end{lemma} 
\begin{proof}
 For such $\tilde \nu$, let $\{\hat v^i, i \in \mathbb{N}\}$ be the solution of \eqref{eq:MeanFieldGameFBSDE_inf} and let $\{\hat u^{i,N}, 1 \leq i \leq N\}$ be the solutions of the $N$-player system \eqref{eq:NASHFBSDE}, with similar initial conditions as the mean field game. From the proof of Theorem \ref{thm-strat-con} it follows that there exists $\hat C>0$ not depending on $N$ and $s\in [0,T]$ such that  
\bd  
\begin{aligned}
&E\left[\left(\tilde \nu_{s}-\frac{1}{N}\sum_{i=1}^N \hat v_s^i\right)^{2}\right] \\
&\leq 2E\left[\left(\tilde \nu_{s}-\frac{1}{N}\sum_{i=1}^N \hat u_s^{i,N}\right)^{2} \right] +2E\left[\left(\frac{1}{N}\sum_{i=1}^N \hat u_s^{i,N}-\frac{1}{N}\sum_{i=1}^N \hat v_s^i\right)^2\right] \\
&\leq C\frac{1}{N^2}+ C\frac{1}{N^2} N \sum_{i=1}^N  E\left[\left( \hat u_s^{i,N}- \hat v_s^i\right)^2\right] \\
&\leq  \hat C\frac{1}{N^2}, 
\end{aligned}
\ed
where we used \eqref{2-p} in the second inequality. 
 \end{proof}

In order to prove Theorem \ref{thm-eps-nash} we will need the
following lemma that bounds the difference between the performance
functional of the mean field game, $J^{i,\infty}$ in
\eqref{def:objective_infMFG} and the $N$-player game's performance functional, $J^{i,N}$ in \eqref{def:FPGobjective}. 

\begin{lemma}\label{lem-j-dif} Let $\hat \nu$ and $\{\hat v^i, \, i\in \mathbb{N}\}$ be the equilibrium strategies of the mean field game as in Definition \ref{def:Nash_infMFG}. For any $N \in \mathbb{N}$ and $1\leq i \leq N$, define  $\hat{v}^{-i}=(\hat{v}^{1},...\hat{v}^{i-1},\hat{v}^{i+1},...,\hat{v}^{N})$. Then under assumptions \eqref{init-assum} and \eqref{a-const} there exists $C>0$ independent from $N$ such that for all $u \in \mathcal A$ and $i \in \mathbb{N}$ we have   
$$
\left| J^{i,N}(u,\hat{v }^{-i}) -  {J}^{i,\infty}(u,\hat \nu)\right| \leq C \|u\|_{2,T}(1+\|u\|_{2,T}) \left(\frac{1}{N} \right). 
$$
\end{lemma} 
\begin{proof} 
Let $u\in \mathcal A$. For $(u,\hat v^{-i})$ let $S^{(u,\hat v^{-i})}$ be as in \eqref{def:S} and let $S^{\hat \nu}$ be as in \eqref{def:S_infMFG}. Using \eqref{def:FPGobjective} and \eqref{def:objective_infMFG} and then \eqref{def:Y} and  \eqref{def:Y_infMFG} we get 
\be \label{j-sub}
\begin{aligned} 
\left| J^{i,N}(u,\hat{v}^{-i}) -  J^{i,\infty}(u,\hat \nu)\right|&=
\left|E \Bigg[\int_0^T S^{(u,\hat{v}^{-i})}_t u_t dt - \int_0^T S^{\hat \nu}_t u_t dt    \Bigg] \right|  \\
&=\Bigg|E \Bigg[ \int_0^T \kappa\left(\int_0^t e^{-\rho (t-s)}\Big(\frac{1}{N}\big(\sum_{j\not = i} \hat{v}_s^{j}+u_s\big)- \hat{\nu}_s \Big) ds \right) u_t dt   \Bigg] \Bigg| \\ 
&\leq  \kappa \Bigg|E \Bigg[ \int_0^T \left(\int_0^t e^{-\rho (t-s)}\Big(\frac{1}{N}\sum_{j =1}^N  \hat{v}_s^{j}- \hat \nu_s \Big) ds \right) u_t dt \Bigg] \Bigg| \\
&\quad+ \kappa \Bigg|E \Bigg[ \int_0^T \left(\int_0^t e^{-\rho (t-s)}\frac{1}{N}\big( \hat{v}^i_s -u_s\big) ds \right) u_t dt \Bigg] \Bigg|\\
&=: \kappa \sum_{k=1}^2 I_k.
\end{aligned} 
\ee
Using Fubini Theorem, Jensen and H\"older inequalities we get for $I_1$:
\be\label{i-1}
\begin{aligned} I_1 &\leq    \int_0^T  \int_0^t e^{-\rho (t-s)}\Big|E\Big[\big(\frac{1}{N}\sum_{j =1}^N  \hat{v}_s^{j}- \hat\nu_s \big) u_t \Big]  \Big| ds  dt     \\
&\leq \int_0^T  \int_0^T E\Big[\Big| \big(\frac{1}{N}\sum_{j =1}^N  \hat{v}_s^{j}- \hat \nu_s \big) u_t\Big| \Big]   ds  dt     \\
&\leq C(T) \bigg(\int_0^T  E\Big[\big(\frac{1}{N}\sum_{j =1}^N  \hat{v}_s^{j}- \hat \nu_s \big)^2 \Big] ds\bigg)^{1/2} \Big(\int_0^TE[u_t^2]dt  \Big)^{1/2}        \\
&\leq C(T)\|u\|_{2,T} N^{-1},  
\end{aligned} 
\ee
where we used Lemma \ref{lemma-con-mf} and the fact that $u\in \mathcal A$ (see \eqref{def:admissset}) in the last inequality. 
 
Using Fubini's Theorem and H\"older inequality we get for $I_2$:
\be \label{i-2}
\begin{aligned} 
I_2 & \leq  \int_0^T\int_0^T \frac{1}{N}  E[|(\hat{v}^i_s -u_s)u_t|] ds  dt  \\
&\leq C(T) \frac{1}{N} \Big(\int_0^T E[(\hat{v}^i_s -u_s)^2] ds\Big)^{1/2}   \Big(\int_0^T E[u_t^2] dt\Big)^{1/2}   \\
&\leq C(T) \frac{1}{N} \left(  \|\hat{v}^i\|_{2,T}+ \|u\|_{2,T}   \right) \|u\|_{2,T}  \\
&\leq C(T)\frac{1}{N}(1+\|u\|_{2,T}) \|u\|_{2,T}, 
\end{aligned} 
\ee
where we used Proposition \ref{lem-bnd-ui}(ii) in the last inequality. 

By plugging in \eqref{i-1} and \eqref{i-2} to \eqref{j-sub} we get the result. 
\end{proof} 

\begin{proof} [Proof of Theorem \ref{thm-eps-nash}]
First note that the inequality 
$$
J^{i,N}(\hat{v}^{i}; \hat{v}^{-i}) \leq \sup_{u \in \mathcal A} J^{i,N}(u; \hat{v}^{-i}), 
$$
holds trivially by the definition of the supremum. 

Using Lemma \ref{lem-j-dif} we get for any $u\in \mathcal A$,
\be \label{rt1}
\begin{aligned}
J^{i,N}(u; \hat v^{-i})& \leq  J^{i,\infty}(u,\hat \nu) +C\|u\|_{2,T}(1+\|u\|_{2,T})\frac{1}{N} \\ 
&\leq   J^{i,\infty}(\hat v^i,\hat\nu) +C\|u\|_{2,T}(1+\|u\|_{2,T})\frac{1}{N}, 
\end{aligned}
\ee
where $C>0$ is a constant not depending on $u$ or $N$. We used $ J^{i,\infty}(\hat u^i,\hat \nu)= \sup_{u\in \mathcal A} J^{i,\infty}(u,\hat \nu)$ in the second inequality.

Using Lemma \ref{lem-j-dif} again we get for some constant $\tilde C>0$ that 
\be \label{rt2} 
\begin{aligned}
  J^{i,\infty}(\hat v^i,\hat \nu )  &\leq   J^{i,N}(\hat v^i,\hat v^{-i}) + C\|\hat v_i\|_{2,T}(1+\|\hat v_i\|_{2,T})\frac{1}{N} \\
  &\leq J^{i,N}(\hat v^i,\hat v^{-i}) +\tilde C\frac{1}{N},
\end{aligned}
\ee
where we used Proposition \ref{lem-bnd-ui}(ii) in the second inequality. 

Using \eqref{rt2} to bound the right hand side of \eqref{rt1}, we get 
$$
 J^{i,N}(u; \hat v^{-i})  \leq  J^i(\hat v^i,\hat v^{-i}) +C \|u\|^2_{2,T} (1\vee  \|u\|^2_{2,T})\frac{1}{N}, \quad \textrm{for all }  u \in \mathcal A,
$$
which completes the proof. 
\end{proof} 


\section{Proofs for the finite player
  game} \label{subsec:proofs-finite}

We start with establishing a strict concavity property of each player's objective  functional in~\eqref{def:FPGobjective} in the following
\begin{lemma} \label{lem:concave_finite}
Let $i \in \{1,\ldots, N\}$. The functional $u^{i,N} \mapsto J^{i,N}(u^{i,N};u^{-i,N})$ in~\eqref{def:FPGobjective} is strictly concave in $u^{i,N} \in \mathcal{A}$.
\end{lemma}
\begin{proof}
First, observe that we can decompose the aggregated transient price impact $Y^{u^N}$ in~\eqref{def:Y} into $Y^{u^N} = Y^{u^{i,N}} + Y^{u^{-i,N}}$ where
\begin{equation} \label{def:Yui}
    Y_t^{u^{i,N}} \triangleq e^{-\rho t} y + \frac{\gamma}{N} \int_0^t e^{-\rho (t-s)} u^{i,N}_s ds, \quad
    Y_t^{u^{-i,N}} \triangleq \frac{\gamma}{N} \int_0^t e^{-\rho (t-s)} \left( \sum_{j\neq i} u^{j}_s \right) ds
\end{equation}
for all $t \in [0,T]$. As a consequence, we can write
\begin{equation} \label{rep:S}
S_t^{u^N} = P_t - \kappa Y_t^{u^{-i,N}} - \kappa Y_t^{u^{i,N}}
\end{equation}
in agent $i$'s performance functional $u^{i,N} \mapsto J^{i,N}(u^{i,N};u^{-i,N})$ in~\eqref{def:FPGobjective}. Next, using the product rule in the expression of the terminal
liquidation value
$X_T^{u^{i,N}} (P_T - \varrho X^{u^{i,N}}_T)$ we can rewrite the
functional as
\begin{equation*}
  J^{i,N}(u^{i,N};u^{-i,N}) = X_0^{u^{i,N}} (P_0 - \varrho X_0^{u^{i,N}}) + \kappa J_1^{i,N}(u^{i,N}) + J_2^{i,N}(u^{i,N};u^{-i,N}),
\end{equation*}
where
\begin{align*}
J_1^{i,N}(u^{i,N}) \triangleq & \; E \left[ \int_0^T Y^{u^{i,N}}_t dX^{u^{i,N}}_t \right], \\
  J_2^{i,N}(u^{i,N};u^{-i,N}) \triangleq & \; E \left[ \int_0^T X^{u^{i,N}}_t (2\varrho u^{i,N}_t - \phi X_t^{u^{i,N}} ) dt - \lambda \int_0^T
    \left( u^{i,N}_t \right)^2 dt \right. \\
    & \left. \hspace{15pt} + \int_0^T X^{u^{i,N}}_t dP_t - \kappa \int_0^T Y^{u^{-i,N}}_t u^{i,N}_t dt \right].
\end{align*}
Regarding the mapping $u^{i,N} \mapsto J_1^{i,N}(u^{i,N})$ we can deduce from the arguments in~\citet[Proof of Theorem 2.3]{Lehalle-Neum18} that this functional is strictly concave in $u^{i,N}$. Hence, it is left to show that 
\begin{equation} \label{proof:concave:eq1}
  J_2^{i,N}(\varepsilon u^{i,N} + (1-\varepsilon) w^{i,N};u^{-i,N}) - \varepsilon J_2^{i,N}( u^{i,N} ;u^{-i,N}) - (1-\varepsilon) J_2^{i,N}(w^{i,N};u^{-i,N}) > 0
\end{equation}
for all $\varepsilon \in (0,1)$ and $u^{i,N}, w^{i,N} \in \mathcal{A}$ such that $u^{i,N} \neq w^{i,N}$  $d\P \otimes ds\textrm{-a.e. on } \Omega \times
[0,T]$. Using that $X_t^{\varepsilon u^{i,N} + (1-\varepsilon) w^{i,N}} =
\varepsilon X^{u^{i,N}}_t + (1-\varepsilon) X^{w^{i,N}}_t$  
with $X^{u^{i,N}}_0 = X^{w^{i,N}}_0$, a
straightforward computation reveals that
\begin{equation} \label{proof:concave:eq2}
  \begin{aligned}
    & J_2^{i,N}(\varepsilon u^{i,N} + (1-\varepsilon) w^{i,N};u^{-i,N}) - \varepsilon J_2^{i,N}( u^{i,N} ;u^{-i,N}) - (1-\varepsilon) J_2^{i,N}(w^{i,N};u^{-i,N}) \\
    & = \varepsilon (1-\varepsilon) E \bigg[ \int_0^T \Big(
    2\varrho (X^{u^{i,N}}_t- X^{w^{i,N}}_t) (w^{i,N}_t - u^{i,N}_t)
    \Big. \bigg. \\
    & \bigg. \Big. \hspace{100pt} + \phi (X_t^{u^{i,N}}-X_t^{w^{i,N}})^2 + \lambda
    (u^{i,N}_t - w^{i,N}_t)^2 \Big) dt \bigg].
  \end{aligned}
\end{equation}
Obviously, the last two terms in~\eqref{proof:concave:eq2} are always
strictly positive. Moreover, 
regarding the first term in~\eqref{proof:concave:eq2} integration by
parts yields
\begin{equation*}
2\varrho \int_0^T  (X^{u^{i,N}}_t- X^{w^{i,N}}_t) (w^{i,N}_t - u^{i,N}_t) dt =
\varrho (X_T^{u^{i,N}}-X_T^{w^{i,N}})^2 > 0.
\end{equation*}
Putting all together, we obtain~\eqref{proof:concave:eq1} as desired and the claim.
\end{proof}

From Lemma \ref{lem:concave_finite} we obtain the following important consequence.
\begin{lemma} \label{lem:uniqueNE_finite}
There exists at most one Nash equilibrium in the sense of Definition~\ref{def:Nash}.
\end{lemma}
\begin{proof}
 This follows from Lemma~\ref{lem:concave_finite} by adopting the same argumentation via contradiction from \citet[Proposition 4.8]{SchiedStrehelZhang:17} to our finite player game.
\end{proof}

\begin{proof}[Proof of Lemma~\ref{thm:NASHFBSDE}] Let us first consider a single fixed agent $i \in \{1,\ldots,N\}$ and
  characterize her best response to the other agents' given fixed
  strategies $u^{-i,N} \in \mathcal{A}^{N-1}$ by maximizing her
  performance functional $u^{i,N} \mapsto J^{i,N}(u^{i,N};u^{-i,N})$ 
  in~\eqref{def:FPGobjective}. To this end, note that we can decompose the jointly aggregated transient price impact $Y^{u^N}$ in~\eqref{def:Y} into $Y^{u^N} = Y^{u^{i,N}} + Y^{u^{-i,N}}$ with $Y^{u^{i,N}}, Y^{u^{-i,N}}$ as in~\eqref{def:Yui} above, and we can write $S^{u^N} = P - \kappa Y^{u^{-i,N}} - \kappa Y^{u^{i,N}}$
in agent $i$'s performance functional
in~\eqref{def:FPGobjective}. Next, a computation very similar to the
proof of Lemma~5.2 in~\cite{N-V19} (with parameter $\gamma/N$ instead
of $\gamma$) shows that strategy $u^{i,N}$ determines agent $i$'s unique
best response in~\eqref{def:FPGoptimization} to a given set of
competitor strategies $u^{-i,N}$ if and only if
$(X^{u^{i,N}}, Y^{u^{i,N}},u^{i,N},Z^{u^{i,N}})$ satisfy the following coupled linear
FBSDE system
\begin{equation} \label{eq:iFBSDE}
\left\{
\begin{aligned}
    dX^{u^{i,N}}_t = & \, - u^{i,N}_t dt, \quad X^{u^{i,N}}_0 = x^{i,N},\\
    dY^{u^{i,N}}_t = & \, -\rho Y_t^{u^{i,N}} dt + \frac{\gamma}{N} u^{i,N}_t dt, \quad Y^{u^{i,N}}_0 = y, \\
    du^{i,N}_t = & \, \frac{1}{2\lambda} \left( dP_t - \kappa dY^{u^{-i,N}}_t \right) + \frac{\kappa\rho}{2\lambda} Y^{u^{i,N}}_t dt - \frac{\phi}{\lambda} X^{u^{i,N}}_t dt + \frac{\rho}{2\lambda} Z^{u^{i,N}}_t dt + dM^{i,N}_t,  \\
    & \hspace{160pt} u^{i,N}_T =\frac{\varrho}{\lambda} X^{u^{i,N}}_T - \frac{\kappa}{2\lambda} \left(Y^{u^{-i,N}}_T + Y^{u^{i,N}}_T \right), \\
    dZ^{u^{i,N}}_t = & \,\rho Z^{u^{i,N}}_t dt + \frac{\gamma\kappa}{N}u^{i,N}_t dt + dN^{i,N}_t, 
    \quad  Z^{u^{i,N}}_T = 0,
\end{aligned}
\right.
\end{equation}
for two suitable square integrable martingales $M^{i,N}=(M^{i,N}_t)_{0 \leq t
  \leq T}$ and $N^{i,N}=(M^{i,N}_t)_{0 \leq t \leq T}$. Moreover, from the
proof of Lemma 5.2 in~\cite[equation (5.7)]{N-V19}  we also know that
these martingales are given by
\begin{align}
    M^{i,N}_t = & \; \frac{1}{2\lambda} \tilde{M}^{i,N}_t - \frac{\gamma\kappa}{2\lambda N} \int_0^t e^{\rho s} d\tilde{N}^{i,N}_s, \label{eq:fpmart1} \\
    N^{i,N}_t = & \; - \frac{\gamma\kappa}{N} \int_0^t e^{\rho s} d\tilde{N}^{i,N}_s, \label{eq:fpmart2}
\end{align}
where
\begin{align}
  \tilde{N}^{i,N}_t \triangleq & \; E_t \left[ \int_0^T e^{-\rho s} u^{i,N}_s \,
                             ds
                             \right], \label{eq:fpmart3} \\
  \tilde{M}^{i,N}_t \triangleq & \; E_t\left[2\phi \int_0^T X^{u^{i,N}}_s ds + 2 \varrho X^{u^{i,N}}_T - P_T \right] ,\label{eq:fpmart4}
\end{align}
 for all $t \in [0,T]$.
Finally, in order for a set of controls $(u^{i,N})_{i \in \{1,\ldots, N\}} \subset \mathcal A^{N}$ to yield the unique Nash equilibrium in the sense of Definition~\ref{def:Nash}, the above FBSDE system must be satisfied simultaneously for all agents $i = 1,\ldots,N$. Now, rewriting for each $i \in \{1,\ldots,N\}$ the following expression in the BSDE for $u^{i,N}$ in~\eqref{eq:iFBSDE} as
\begin{equation*}
\begin{aligned}
    -\frac{\kappa}{2\lambda} dY^{u^{-i,N}}_t + \frac{\kappa\rho}{2\lambda} Y^{u^{i,N}}_t dt = & \, \frac{\kappa\rho}{2\lambda} \left( Y_t^{u^{i,N}} + Y_t^{u^{-i,N}} \right) dt -\frac{\gamma\kappa}{2\lambda N} \left( \sum_{j \neq i} u^{j,N}_t \right) dt \\
    = & \, \frac{\kappa\rho}{2\lambda} Y_t^{u^N} dt -\frac{\gamma\kappa}{2\lambda N} \left( \sum_{j \neq i} u^{j,N}_t \right) dt,
\end{aligned}
\end{equation*}
we note that the coupling of all systems in~\eqref{eq:iFBSDE} for all $i \in \{1,\ldots,N\}$ only depends on $Y^{u^N}$ but not on the processes $Y^{u^{i,N}}$ and $Y^{u^{-i,N}}$ separately. Therefore together with Lemma \ref{lem:uniqueNE_finite} we obtain the characterization of the unique Nash equilibrium as claimed in~\eqref{eq:NASHFBSDE}.
\end{proof} 

\begin{proof}[Proof of Proposition~\ref{prop:sol-mean-FBSDE}] The proof is similar to the proof of~\cite[Theorem 3.2]{N-V19}. Therefore we only sketch the main steps. To solve the linear FBSDE system in~\eqref{eq:mean-FBSDE} we conveniently rewrite it as   
\begin{equation} \label{eq:linODE}
    d\boldsymbol{\ol X}^N_t = \ol F^N \, \boldsymbol{\ol X}^N_t dt + d\boldsymbol{\ol M}^N_t \quad (0 \leq t \leq T),
\end{equation}
with $\ol F^N \in \mathbb{R}^{4 \times 4}$ introduced in~\eqref{def:Fbar} and
\begin{equation*}
    \boldsymbol{\ol X}_t^N \triangleq \begin{pmatrix}
   \ol X_t^{\bar u^N} \\ \ol Y_t^{\bar u^N} \\ \bar u_t^N \\ \ol Z_t^{\bar u^N}
    \end{pmatrix},   
    \, \quad  \boldsymbol{\ol M}_t^N \triangleq \begin{pmatrix}
    0 \\ 0 \\ \frac{1}{2\lambda} P_t + \ol M_t^N \\ \ol N_t^N
    \end{pmatrix} \quad (0 \leq t \leq T).
\end{equation*}
The corresponding initial and terminal conditions are given by $\boldsymbol{\ol X}^{N, 1}_0 =\bar x^N$, $\boldsymbol{\ol X}^{N,2}_0 = y$ and
\begin{equation} \label{def:termCond}
\left( \frac{\varrho}{\lambda},  - \frac{\kappa}{2\lambda}, -1,0
   \right) \boldsymbol{\ol X}_T^N = 0 \quad \text{and} \quad 
   \left(0,0,0,1 \right) \boldsymbol{\ol X}^N_T = 0.
\end{equation}
Note that the unique solution of the linear system in~\eqref{eq:linODE} can be expressed in terms of the matrix exponential defined in~\eqref{def:matrixExpEbar} via  
\begin{equation} \label{eq:X-T}
    \boldsymbol{\ol X}_T^N = \ol Q(T-t) \boldsymbol{\ol X}_t^N + \int_t^T \ol Q(T-s) d\boldsymbol{\ol M}_s^N \quad (0 \leq t \leq T).
\end{equation}
Multiplying~\eqref{eq:X-T} from the left with the row vector $\left(\frac{\varrho}{\lambda}, - \frac{\kappa}{2\lambda}, -1,0 \right)$, using the first terminal condition in~\eqref{def:termCond}, taking conditional expectations and solving for $\ol u^N_t$ gives us 
\begin{equation} \label{eq:u_Rep}
\begin{aligned} 
\ol u_{t}^N = & -\frac{\ol G_1(T-t)}{\ol G_3(T-t)}\ol X_{t}^{\bar u^N} -\frac{\ol G_2(T-t)}{\ol G_{3}(T-t)}\ol Y_{t}^{\bar u^N}  - \frac{\ol G_4(T-t)}{\ol G_3(T-t)} \ol Z_t^{\bar u^N} \\
 &  - \frac{1}{2\lambda}E_{t} \left[ \int_{t}^T\frac{\ol G_3(T-s)}{\ol G_{3}(T-t)}dA_{s} \right] \qquad (0 \leq t \leq T),
\end{aligned}
\end{equation}
with $\ol G$ as defined in~\eqref{def:Gbar}. Also note that $\ol G_{3}(T-t) \neq 0$ for all $t \in [0,T]$ by assumption. Moreover, repeating the same steps by using the second terminal condition in~\eqref{def:termCond} and solving the obtained identity for $\ol Z^{\bar u^N}$ yields
\begin{equation} \label{eq:Z_Rep}
\begin{aligned} 
\ol Z_t^{\bar u^N} = & - \frac{\ol H_1(T-t)}{\ol H_4(T-t)} \ol X_t^{\bar u^N}  - \frac{\ol H_2 (T-t)}{\ol H_4(T-t)} \ol Y_t^{\bar u^N} - \frac{\ol H_3(T-t)}{\ol H_4(T-t)} \bar u_t^N  \\
& -  \frac{1}{2\lambda} E_t \left[ \int_t^T \frac{\ol H_{3}(T-s)}{\ol H_{4}(T-t)} dA_s \right] \qquad (0 \leq t \leq T),
\end{aligned}
\end{equation}
where $\ol H$ is defined in~\eqref{def:Hbar} and $\ol H_{4}(T-t) \neq 0$ for all $t \in [0,T]$ by Assumption~\ref{assump:1}. Plugging~\eqref{eq:Z_Rep} into~\eqref{eq:u_Rep} and solving for $\ol u$ yields the claim in~\eqref{eq:opt_ubar}, where $\bar v_0$ is well-defined by Assumption~\ref{assump:1}. 


Finally, the claim that $\ol u^N$ in~\eqref{eq:opt_ubar} belongs to $\mathcal A$ follows from Assumption~\ref{assump:1}, which guarantees the boundedness from above of the functions $\ol G_i(t), \ol H_i(t), i \in \{1,\ldots,4\}$, in~\eqref{def:Gbar1}--\eqref{def:Hbar4}, and by employing a similar Gronwall-type argument as in the proof of~\cite[Theorem 3.2]{N-V19}, step 2.
\end{proof} 

\begin{proof}[Proof of Theorem~\ref{thm:main-finite}] In view of Corollary~\ref{cor:NASHFBSDE} we have to solve for each $i \in \{1,\ldots,N\}$ the linear FBSDE in~\eqref{eq:NASHFBSDE*}. Therefore, the proof of Theorem~\ref{thm:main-finite} follows the same reasoning as the proof of Proposition~\ref{prop:sol-mean-FBSDE} above. Indeed, observe that for each $i \in \{1,\ldots,N\}$ the system in~\eqref{eq:NASHFBSDE*} can be written as 
\begin{equation} \label{eq:linODENash}
    d\boldsymbol{X}^{i,N}_t = F^N \boldsymbol{X}^{i,N}_t dt + d\boldsymbol{M}^{i,N}_t \quad (0 \leq t \leq T)
\end{equation}
where 
\begin{equation*}
    \boldsymbol{X}^{i,N}_t \triangleq \begin{pmatrix}
   X^{u^{i,N}}_t \\ u^{i,N}_t \\ Z^{u^{i,N}}_t
    \end{pmatrix},   
    \, \quad  \boldsymbol{M}^{i,N}_t \triangleq \begin{pmatrix}
    0 \\ \frac{1}{2\lambda} (P_t - \kappa \ol Y^{\bar u^N}_t) +  M^{i,N}_t  \\ N^{i,N}_t
    \end{pmatrix} \qquad (0 \leq t \leq T)
\end{equation*}
and $F^N \in \mathbb{R}^{3 \times 3}$ defined in~\eqref{def:F}; initial and terminal condition are given by $\boldsymbol{X}^{i,N,1}_0 = x^{i,N}$ and
\begin{equation} \label{def:termCondNashFBSDE}
\left( \frac{\varrho}{\lambda}, -1,0 \right) 
\boldsymbol{X}^{i,N}_T = \frac{\kappa}{2\lambda} \ol Y^{\bar u^N}_T, \quad \left(0,0,1 \right) \boldsymbol{X}^{i,N}_T = 0.
\end{equation}
That is, as in the proof of Proposition~\ref{prop:sol-mean-FBSDE} we can write the unique solution of the linear system in~\eqref{eq:linODENash} as
\begin{equation*}
    \boldsymbol{X}^{i,N}_T = Q(T-t) \boldsymbol{X}^{i,N}_t + \int_t^T Q(T-s) d\boldsymbol{M}^{i,N}_s \quad (0 \leq t \leq T),
\end{equation*}
with the matrix exponential introduced in~\eqref{def:matrixExpE} and follow the same steps. That is, via the terminal conditions in~\eqref{def:termCondNashFBSDE} we eventually get the identities
\begin{equation} \label{eq:u_RepNash}
\begin{aligned} 
u^{i,N}_{t} = & -\frac{G_1(T-t)}{G_2(T-t)} X^{u^{i,N}}_{t} - \frac{G_3(T-t)}{G_2(T-t)} Z^{u^{i,N}}_t \\
 &  - \frac{1}{2\lambda}E_{t} \left[
   \int_{t}^T\frac{G_2(T-s)}{G_{2}(T-t)} \left(dA_{s} - \kappa d\ol
     Y^{\bar u^N}_s\right) -\frac{\kappa}{G_2(T-t)} \ol Y^{\bar u^N}_T  \right] \quad (0 \leq t \leq T),
\end{aligned}
\end{equation}
as well as
\begin{equation} \label{eq:Z_RepNash}
\begin{aligned} 
Z^{u^{i,N}}_t = & - \frac{H_1(T-t)}{H_3(T-t)} X^{u^{i,N}}_t - \frac{H_2 (T-t)}{H_3(T-t)} u^{i,N}_t\\
& - \frac{1}{2\lambda} E_t \left[ \int_t^T \frac{H_{2}(T-s)}{H_{3}(T-t)} \left(dA_{s} - \kappa d\ol Y^{\bar u^N}_s\right) \right] \qquad (0 \leq t \leq T),
\end{aligned}
\end{equation}
with $G$ and $H$ defined in~\eqref{def:G} and~\eqref{def:H}. Note that we have $G_2(T-t) \neq 0$ and $H_3(T-t) \neq 0$ for all $t \in [0,T]$ by Assumption \ref{assump:2}. Plugging~\eqref{eq:Z_RepNash} into~\eqref{eq:u_RepNash} and solving for $u^{i,N}$ yields the claim in~\eqref{eq:opt_ui}, where $v_0$ is well-defined by Assumption~\ref{assump:2}.

Finally, the claim that $u^{i,N}$ in~\eqref{eq:opt_ui} belongs to $\mathcal A$ follows from Assumption \ref{assump:2} and the fact that the functions $G_i(t), H_i(t), i \in\{1,\ldots,3\}$ in~\eqref{def:G1}--\eqref{def:H3} are bounded from above, using again a similar Gronwall-type argument as in the proof of~\cite[Theorem 3.2]{N-V19}, step 2.

\end{proof} 



\section{Proofs for the infinite player game} \label{subsec:proofs-infinite}
As in the finite player game, we first provide that each player's objective functional in~\eqref{def:objective_infMFG} is strictly concave.
\begin{lemma} \label{lem:concave_infinite}
Let $i \in \mathbb{N}$. The functional $v^i \mapsto J^{i,\infty}(v^i;\nu)$ in~\eqref{def:objective_infMFG} is strictly concave in $v^i \in \mathcal{A}$.
\end{lemma}
\begin{proof}
The computations are similar to the proof of Lemma~\ref{lem:concave_finite} above. Indeed, using again the product rule in the expression of the terminal
liquidation value $X_T^{v^i} (P_T - \varrho X^{v^i}_T)$
in agent $i$'s performance functional $v^i \mapsto J^{i,\infty}(v^i;\nu)$ in~\eqref{def:objective_infMFG}, we obtain
\begin{equation*}
  J^{i,\infty}(v^i;\nu) = X_0^{v^i} (P_0 - \varrho X_0^{v^i}) + J_1^{i,\infty}(v^i;\nu),
\end{equation*}
where
\begin{align*}
  J_1^{i,\infty}(v^i;\nu) \triangleq & \; E \left[ \int_0^T X^{v^i}_t (2\varrho v^i_t - \phi X_t^{v^i} ) dt - \lambda \int_0^T
    \left( v^i_t \right)^2 dt \right. \\
    & \left. \hspace{15pt} + \int_0^T X^{v^i}_t dP_t - \kappa \int_0^T Y^{\nu}_t v^i_t dt \right].
\end{align*}
Strict concavity of the mapping $v^i \mapsto J_1^{i,\infty}(v^i;\nu)$ then follows as in the proof of Lemma~\ref{lem:concave_finite}.
\end{proof}
Consequently, again similar to the finite player game, we can establish
\begin{lemma} \label{lem:uniqueNE_infinite}
There exists at most one Nash equilibrium in the sense of Definition~\ref{def:Nash_infMFG}.
\end{lemma}
\begin{proof}
 This follows from Lemma~\ref{lem:concave_infinite} by adopting  once more the same argumentation via contradiction from \citet[Proposition 4.8]{SchiedStrehelZhang:17} to our infinite player game.
\end{proof}

\begin{proof}[Proof of Lemma~\ref{thm:MeanFieldFBSDE_inf}]
  First, observe that for a given and fixed net trading flow
  $\nu \in \mathcal{A}$ the optimisation problem of an individual
  agent $i \in \mathbb{N}$ in~\eqref{def:optimization_infMFG} is very
  similar to the single-agent optimization problem studied
  in~\citet{BMO:19}, where the agent is only facing temporary price
  impact and the unaffected price process is given by $P-\kappa Y^\nu$. The only difference is that the value of the terminal inventory $X^{v^i}_T$ in~\eqref{def:objective_infMFG} is expressed in terms of $P_T$ and not $S^{\nu}_T$. Therefore, it follows along
  the lines of the proof of~\cite[Theorem 3.1]{BMO:19} that the unique
  solution $\hat v^i \in \mathcal{A}$
  of~\eqref{def:optimization_infMFG} satisfies the linear FBSDE system
\begin{equation} \label{eq:BMOFBSDE}
\left\{
\begin{aligned}
    dX^{\hat v^i}_t = & \, - \hat v^i_t dt, \quad X^{\hat v^i}_0 = x^i,\\
    dY^{\nu}_t  = & \, -\rho Y^{\nu}_t  dt + \gamma  \nu_{t} dt, \quad Y^{\nu}_0 = y, \\
    d\hat v^i_t = & \, \frac{1}{2\lambda} \left( dP_t - \kappa dY^{
        \nu}_t \right) - \frac{\phi X^{\hat v^{i}}_t}{\lambda} dt +
    dL^i_t,  \qquad \hat v^i_T = \frac{\varrho}{\lambda} X^{\hat
      v^i}_T - \frac{\kappa}{2\lambda} Y^{\nu}_T,
\end{aligned}
\right.
\end{equation}
with a square integrable martingale $L^i=(L^i_t)_{0 \leq t \leq T}$ given by
\begin{equation}
    L^i_t \triangleq \frac{1}{2\lambda} E_t \left[2\phi \int_0^T X^{\hat v^i}_s ds + 2 \varrho X^{\hat v^i}_T - P_T \right] \qquad (0 \leq t \leq T).  \label{eq:ipmart}
\end{equation}
Consequently, in order for a collection of controls $(\hat v ^{i})_{i \in \mathbb{N}} \subset \mathcal A$ to yield the unique Nash equilibrium in the sense of Definition~\ref{def:Nash_infMFG}, the above FBSDE system must be satisfied simultaneously for all agents $i \in \mathbb{N}$. The uniqueness of the solution to \eqref{eq:BMOFBSDE} then follows from Lemma \ref{lem:uniqueNE_infinite}.

In order to complete the proof, we need to show that we can find an admissible $\hat \nu$ such that the consistency condition \eqref{MeanFieldFBSDE_inf_consist} is satisfied. We will show that $\tilde \nu$ that solves \eqref{eq:MeanField_infAggregated} is the right candidate for that. 

For such $\tilde \nu$, let $\{\hat v^i, i \in \mathbb{N}\}$ be the solution of \eqref{eq:MeanFieldGameFBSDE_inf}. Then from \eqref{eq:MeanField_infAggregated}, \eqref{eq:MeanFieldGameFBSDE_inf} and \eqref{eq:opt_ustar_alt}  we have,  
\be \label{new1} 
\begin{aligned} 
\Big|\frac{1}{N}\sum_{i=1}^N\hat{v}^i_{t} -\tilde{\nu}_t \Big|&= \left | \frac{R'(T-t)}{R(T-t)} \left( \tilde{X}_t^{\tilde \nu} - \frac{1}{N}\sum_{i=1}^N X^{\hat{v}^{i}}_t \right) \right| \\ 
&\leq  \sup_{s\in [0,T]}\left| \frac{R'(T-s)}{R(T-s)} \right| \left( \Big|\frac{1}{N}\sum_{i=1}^{N}x_i - \tilde x \Big| +  \Big|\int_0^t\Big(\frac{1}{N}\sum_{i=1}^N\hat v^i_s - \tilde \nu_s\Big)  ds \Big|  \right) \\ 
\end{aligned}
\ee
From \eqref{def:R} it follows that $\sup_{s\in [0,T]}\left| \frac{R'(T-s)}{R(T-s)} \right|<\infty$. Together with \eqref{eq:initPosLimit} Jensen inequality and Fubini theorem we get for any $0\leq t\leq T$ that there exist constants $C_1(N), C_2>0$ such that  
\be \label{new1} 
\begin{aligned} 
E\left[\sup_{s\in [0,t]}\Big(\frac{1}{N}\sum_{i=1}^N\hat{v}^i_{s} -\tilde{\nu}_s \Big)^2 \right] \leq C_1(N) + C_2\int_0^t E\Big[ \sup_{r\in[0,s]} \Big(\frac{1}{N}\sum_{i=1}^N\hat v^i_r - \tilde \nu_r\Big)^2 \Big]  ds,
 \end{aligned}
\ee
where $C_1(N) \rr 0$ as $N\rr \infty$. Then from Gronwall's lemma we get 
\be \label{consis-l2}
\begin{aligned}
\lim_{N\rr \infty}E\left[\sup_{s\in [0,T]}\left(\tilde \nu_{s}-\frac{1}{N}\sum_{i=1}^N \hat v_s^i\right)^{2}\right]  =0.
\end{aligned}
\ee
Note that the convergence rate in \eqref{consis-l2} is determined by \eqref{eq:initPosLimit}. 

By Fatou's lemma and \eqref{consis-l2} we have
\be \label{fat-1}
E\left[\left(\tilde \nu_{s}-\lim_{N\rr \infty}\frac{1}{N}\sum_{i=1}^N \hat v_s^i\right)^{2}\right] \leq \liminf_N E\left[\left(\tilde \nu_{s}-\frac{1}{N}\sum_{i=1}^N \hat v_s^i\right)^{2}\right] =0.
\ee
Hence it holds that 
$$
\tilde \nu_{s}=\lim_{N\rr \infty}\frac{1}{N}\sum_{i=1}^N \hat v_s^i, \quad \textrm{for all } 0\leq s \leq T, \ P-\textrm{a.s.}, 
$$
and this completes the proof.
\end{proof} 

\begin{proof}[Proof of Corollary \ref{cor:MeanFieldFBSDE_inf}] 
The proof of Corollary \ref{cor:MeanFieldFBSDE_inf} follows immediately from \eqref{fat-1}. 
\end{proof} 

\begin{proof}[Proof of Proposition~\ref{prop:sol-MeanField_infAggregated}] Solving the linear  system in~\eqref{eq:MeanField_infAggregated} follows along the same lines as solving the linear FBSDE systems in~\eqref{eq:mean-FBSDE} and~\eqref{eq:NASHFBSDE*} in Proposition~\ref{prop:sol-mean-FBSDE} and Theorem~\ref{thm:main-finite}, respectively. That is, the system in~\eqref{eq:MeanField_infAggregated} can be written as
\begin{equation} \label{eq:linODENashNew}
    d\boldsymbol{\tilde X}_t = \tilde B \boldsymbol{\tilde{X}}_t dt + d\boldsymbol{\tilde{M}}_t \quad (0 \leq t \leq T),
\end{equation}
where 
\begin{equation*}
    \boldsymbol{\tilde X}_t \triangleq \begin{pmatrix}
  \tilde{X}_t^{\tilde \nu} \\ \tilde{Y}_t^{\tilde \nu} \\ \tilde{\nu}_t
    \end{pmatrix},   
    \, \quad  \boldsymbol{\tilde M}_t \triangleq \begin{pmatrix}
    0 \\ 0 \\ \frac{1}{2\lambda} P_t + \tilde{L}_t
    \end{pmatrix} \qquad (0 \leq t \leq T),
\end{equation*}
and $\tilde B \in \mathbb{R}^{3 \times 3}$ defined in~\eqref{def:Btilde}; initial and terminal conditions are given by $\boldsymbol{\tilde X}^{1}_0 = \tilde x$, $\boldsymbol{\tilde X}^{2}_0 = y$ and
\begin{equation} \label{def:termCondMFGode}
 \left( \frac{\varrho}{\lambda}, -\frac{\kappa}{2\lambda}, -1 \right) \boldsymbol{\tilde X}_T = 0.
\end{equation}
As a consequence, together with the matrix exponential $\tilde R$ defined in~\eqref{def:matrixExpBtilde} we can write the unique solution to the linear system in~\eqref{eq:linODENashNew} as
\begin{equation*}
    \boldsymbol{\tilde X}_T = \tilde R(T-t) \boldsymbol{\tilde{X}}_t + \int_t^T \tilde R(T-s) d\boldsymbol{\tilde{M}}_s \quad (0 \leq t \leq T).
\end{equation*}
Lastly, using the terminal condition in~\eqref{def:termCondMFGode} together with $\tilde K$ defined in~\eqref{def:Ktilde} yields the identity
\begin{equation*} 
 \tilde{\nu}_{t} = -\frac{\tilde K_1(T-t)}{\tilde K_3(T-t)} \tilde{X}_{t}^{\tilde \nu} - \frac{\tilde K_2(T-t)}{\tilde K_3(T-t)} \tilde{Y}_t^{\tilde \nu}
 - \frac{1}{2\lambda} E_t\left[\int_{t}^T\frac{\tilde K_3(T-s)}{\tilde K_{3}(T-t)} dA_{s} \right] \quad (0 \leq t \leq T),
\end{equation*}
where we recall that $\tilde K_3(T-t) \neq 0$ for all $t \in [0,T]$ by assumption.

Finally, the claim that $\tilde\nu$ belongs to $\mathcal A$ can be deduced from our assumption $\inf_{t \in [0,T]}|\tilde K_3(t)| > 0$ and the fact that $\tilde{K}_1,\tilde{K}_2,\tilde{K}_3$ in~\eqref{def:Ktilde1}--\eqref{def:Ktilde3} are bounded from above, following once more a similar Gronwall-type argument as in the proof of~\cite[Theorem 3.2]{N-V19}, step 2.

\end{proof}

\begin{proof} [Proof of Theorem \ref{thm:main-infMFG}]  
 The linear FBSDE system in~\eqref{eq:MeanFieldGameFBSDE_inf} is 
almost the same as the one derived in~\citet{BMO:19} with
 signal process $A-\kappa \tilde Y^{\tilde \nu}$. The only difference is the
   terminal condition for $\hat{v}^i_T$. However, the computations
   in~\cite[Theorem 3.1]{BMO:19} can be easily adapted to yield our
   claim in~\eqref{eq:opt_uiInf}.
\end{proof}




\section{Computing the matrix exponentials} \label{subsec:proofs-matrixexponentials} 

\subsection{Finite player game} \label{sec-mat-fin} 

We start with computing the matrix exponential $\ol Q (t) = \exp(\ol F^N \cdot t) \in \mathbb{R}^{4 \times 4}$ for all $t \in [0, \infty)$ in \eqref{def:matrixExpEbar} by decomposing the matrix $\ol F^N = \ol U \, \ol D \, \ol U^{-1}$ from~\eqref{def:Fbar} into a diagonal matrix $\ol D\in\mathbb{R}^{4\times 4}$ and an invertible matrix $\ol U \in\mathbb{R}^{4\times 4}$. The eigenvalues $\bar\nu_1, \bar\nu_2, \bar\nu_3, \bar\nu_4$ of $\ol F^N$ are the \ roots of the equation
\begin{equation*}
    x^4 + \frac{(N-1)\kappa\gamma}{2N\lambda} x^3 - \left( \frac{\kappa\gamma\rho (N+1)}{2N\lambda} +\rho^2 +\frac{\phi}{\lambda} \right) x^2 + \frac{\phi}{\lambda} \rho^2 = 0.
\end{equation*}
Recall that by Assumption~\ref{assump:1} we assume that these eigenvalues are real-valued and distinct. Since $\det(\ol F^N) = \rho^2 \phi/\lambda > 0$ we can deduce that they are different from zero.

\noindent The corresponding eigenvectors are given by 
\begin{equation*} \label{def:eigenvectorF}
  \bar v_i \triangleq
  \begin{pmatrix}
    -\frac{N(\bar\nu_i-\rho)}{\kappa\gamma\bar\nu_i} \\
     \frac{N(\bar\nu_i - \rho)}{\kappa(\bar\nu_i+\rho)} \\
     \frac{N(\bar\nu_i - \rho)}{\kappa\gamma} \\
     1
    \end{pmatrix} \qquad (i=1,2,3,4).
\end{equation*}
Hence, we have
\begin{equation*}
    \ol D = \begin{pmatrix} 
    \bar\nu_1 & 0 & 0 & 0 \\
    0 & \bar\nu_2 & 0 & 0 \\
    0 & 0 & \bar\nu_3 & 0 \\
    0 & 0 & 0 & \bar\nu_4 
    \end{pmatrix}, \quad 
    \ol U =
    \begin{pmatrix}
    -\frac{N(\bar\nu_1-\rho)}{\kappa\gamma\bar\nu_1} &  -\frac{N(\bar\nu_2-\rho)}{\kappa\gamma\bar\nu_2} &  -\frac{N(\bar\nu_3-\rho)}{\kappa\gamma\bar\nu_3} &  -\frac{N(\bar\nu_4-\rho)}{\kappa\gamma\bar\nu_4} \\
     \frac{N(\bar\nu_1 - \rho)}{\kappa(\bar\nu_1+\rho)} & \frac{N(\bar\nu_2 - \rho)}{\kappa(\bar\nu_2+\rho)} & \frac{N(\bar\nu_3 - \rho)}{\kappa(\bar\nu_3+\rho)} & \frac{N(\bar\nu_4 - \rho)}{\kappa(\bar\nu_4+\rho)} \\
     \frac{N(\bar\nu_1 - \rho)}{\kappa\gamma} & \frac{N(\bar\nu_2 - \rho)}{\kappa\gamma} & \frac{N(\bar\nu_3 - \rho)}{\kappa\gamma} & \frac{N(\bar\nu_4 - \rho)}{\kappa\gamma}  \\
     1 & 1 & 1 & 1
    \end{pmatrix},
\end{equation*}  
and, introducing the differences
\begin{equation*}
    \bar\nu_{i,j} = \bar\nu_{i} - \bar\nu_{j} \qquad (i,j \in \{1,2,3,4\}),
\end{equation*}
we obtain
\begin{align*} 
  & \ol U^{-1} = \frac{1}{2 N \rho^2} \\
  &
  {\footnotesize
  \begin{pmatrix}
   -\frac{2\rho^2 \phi \gamma\kappa (\bar\nu_1+\rho) }{\lambda \bar\nu_{1,2}\bar\nu_{1,3}\bar\nu_{1,4}} & -\frac{\kappa \bar\nu_1 (\bar\nu_1+\rho)(\bar\nu_2+\rho)(\bar\nu_3+\rho) (\bar\nu_4+\rho)}{\bar\nu_{1,2}\bar\nu_{1,3}\bar\nu_{1,4}} & \frac{2\rho^2\gamma\kappa\bar\nu_1(\bar\nu_1+\rho)}{\bar\nu_{1,2}\bar\nu_{1,3}\bar\nu_{1,4}} & -\frac{N \bar\nu_1 (\bar\nu_1+\rho) (\bar\nu_2-\rho)(\bar\nu_3-\rho)(\bar\nu_4-\rho)}{\bar\nu_{1,2}\bar\nu_{1,3}\bar\nu_{1,4}} \\
   \frac{2\rho^2 \phi \gamma\kappa (\bar\nu_2+\rho) }{\lambda \bar\nu_{1,2}\bar\nu_{2,3}\bar\nu_{2,4}} & \frac{\kappa \bar\nu_2 (\bar\nu_1+\rho)(\bar\nu_2+\rho)(\bar\nu_3+\rho) (\bar\nu_4+\rho)}{\bar\nu_{1,2}\bar\nu_{2,3}\bar\nu_{2,4}} & -\frac{2\rho^2\gamma\kappa\bar\nu_2(\bar\nu_2+\rho)}{\bar\nu_{1,2}\bar\nu_{2,3}\bar\nu_{2,4}} & \frac{N \bar\nu_2 (\bar\nu_1-\rho) (\bar\nu_2+\rho)(\bar\nu_3-\rho)(\bar\nu_4-\rho)}{\bar\nu_{1,2}\bar\nu_{2,3}\bar\nu_{2,4}}\\
   -\frac{2\rho^2 \phi \gamma\kappa (\bar\nu_3+\rho) }{\lambda \bar\nu_{1,3}\bar\nu_{2,3}\bar\nu_{3,4}} & -\frac{\kappa \bar\nu_3 (\bar\nu_1+\rho)(\bar\nu_2+\rho)(\bar\nu_3+\rho) (\bar\nu_4+\rho)}{\bar\nu_{1,3}\bar\nu_{2,3}\bar\nu_{3,4}} & \frac{2\rho^2\gamma\kappa\bar\nu_3(\bar\nu_3+\rho)}{\bar\nu_{1,3}\bar\nu_{2,3}\bar\nu_{3,4}} & -\frac{N \bar\nu_3 (\bar\nu_1-\rho) (\bar\nu_2-\rho)(\bar\nu_3+\rho)(\bar\nu_4-\rho)}{\bar\nu_{1,3}\bar\nu_{2,3}\bar\nu_{3,4}}\\
   \frac{2\rho^2 \phi \gamma \kappa (\bar\nu_4 + \rho) }{\lambda \bar\nu_{1,4}\bar\nu_{2,4}\bar\nu_{3,4}} & \frac{\kappa \bar\nu_4 (\bar\nu_1+\rho)(\bar\nu_2+\rho)(\bar\nu_3+\rho) (\bar\nu_4+\rho)}{\bar\nu_{1,4}\bar\nu_{2,4}\bar\nu_{3,4}} & -\frac{2\rho^2\gamma\kappa\bar\nu_4(\bar\nu_4+\rho)}{\bar\nu_{1,4}\bar\nu_{2,4}\bar\nu_{3,4}} & \frac{N \bar\nu_4 (\bar\nu_1-\rho) (\bar\nu_2-\rho)(\bar\nu_3-\rho)(\bar\nu_4+\rho)}{\bar\nu_{1,4}\bar\nu_{2,4}\bar\nu_{3,4}}
  \end{pmatrix}.
  }
\end{align*}
Consequently, the matrix exponential $\ol Q(t) = (\ol Q_{ij}(t))_{1 \leq i,j \leq 4}$ is given by
\begin{equation}
  \ol Q(t) = \ol U \begin{pmatrix}
    e^{\bar\nu_1 t} & 0 & 0 & 0 \\
    0 & e^{\bar\nu_2 t} & 0 & 0 \\
    0 & 0 & e^{\bar\nu_3 t} & 0 \\
    0 & 0 & 0 & e^{\bar\nu_4 t}\\
  \end{pmatrix} \ol U^{-1} \qquad (t \geq 0)
\end{equation}
and it follows that $\ol G$ defined in~\eqref{def:Gbar} can be computed as
\begin{align} 
    \ol G_{1}(T-t) = & \, \frac{\phi}{2 \lambda^2} \, \left( \frac{(\bar\nu_1-\rho)(2\varrho(\bar\nu_1+\rho)+\gamma\kappa \bar\nu_1 + 2\lambda\bar\nu_1 (\bar\nu_1+\rho))}{\bar\nu_1 \bar\nu_{1,2} \bar\nu_{1,3} \bar\nu_{1,4}} e^{\bar\nu_1(T-t)} \right. \nonumber\\ 
    & \hspace{35pt} - \frac{(\bar\nu_2-\rho)(2\varrho(\bar\nu_2+\rho)+\gamma\kappa \bar\nu_2+2\lambda \bar\nu_2 (\bar\nu_2+\rho))}{\bar\nu_2 \bar\nu_{1,2} \bar\nu_{2,3} \bar\nu_{2,4}} e^{\bar\nu_2(T-t)} \nonumber \\
    & \hspace{35pt} + \frac{(\bar\nu_3-\rho)(2\varrho (\bar\nu_3+\rho)+\gamma\kappa \bar\nu_3 + 2\lambda \bar\nu_3(\bar\nu_3+\rho))}{\bar\nu_3 \bar\nu_{1,3} \bar\nu_{2,3} \bar\nu_{3,4}} e^{\bar\nu_3(T-t)} \nonumber \\
    & \hspace{35pt} \left. - \frac{(\bar\nu_4-\rho)(2\varrho(\bar\nu_4+\rho)+\gamma\kappa \bar\nu_4 + 2\lambda \bar\nu_4 (\bar\nu_4+\rho))}{\bar\nu_4 \bar\nu_{1,4} \bar\nu_{2,4} \bar\nu_{3,4}} e^{\bar\nu_4(T-t)} \right), \label{def:Gbar1}
\end{align}
\begin{align} 
    \ol G_{2}(T-t) = & \, \frac{1}{4 \gamma \lambda \rho^2}  \label{def:Gbar2} \\
    & \left( \frac{(2\varrho(\bar\nu_1+\rho)+\gamma\kappa \bar\nu_1 + 2 \lambda\bar\nu_1(\bar\nu_1+\rho))(\bar\nu_1-\rho)(\bar\nu_2+\rho)(\bar\nu_3+\rho)(\bar\nu_4+\rho)}{\bar\nu_{1,2} \bar\nu_{1,3} \bar\nu_{1,4}} e^{\bar\nu_1(T-t)} \right. \nonumber \\ 
    & \hspace{12pt} - \frac{(2\varrho(\bar\nu_2+\rho)+\gamma\kappa \bar\nu_2 + 2 \lambda\bar\nu_2(\bar\nu_2+\rho))(\bar\nu_1+\rho)(\bar\nu_2-\rho)(\bar\nu_3+\rho)(\bar\nu_4+\rho)}{\bar\nu_{1,2} \bar\nu_{2,3} \bar\nu_{2,4}} e^{\bar\nu_2(T-t)} \nonumber \\
    & \hspace{12pt} + \frac{(2\varrho(\bar\nu_3+\rho)+\gamma\kappa \bar\nu_3 + 2 \lambda\bar\nu_3(\bar\nu_3+\rho))(\bar\nu_1+\rho)(\bar\nu_2+\rho)(\bar\nu_3-\rho)(\bar\nu_4+\rho)}{\bar\nu_{1,3} \bar\nu_{2,3} \bar\nu_{3,4}} e^{\bar\nu_3(T-t)} \nonumber \\
    & \hspace{12pt} \left. - \frac{(2\varrho(\bar\nu_4+\rho)+\gamma\kappa \bar\nu_4 + 2 \lambda\bar\nu_4(\bar\nu_4+\rho))(\bar\nu_1+\rho)(\bar\nu_2+\rho)(\bar\nu_3+\rho)(\bar\nu_4-\rho)}{\bar\nu_{1,4} \bar\nu_{2,4} \bar\nu_{3,4}} e^{\bar\nu_4(T-t)} \right), \nonumber
\end{align}
\begin{align} 
    \ol G_{3}(T-t) = & \, \frac{1}{2\lambda} \left( -\frac{(\bar\nu_1-\rho)(2\varrho(\bar\nu_1+\rho)+\gamma\kappa\bar\nu_1+2\lambda\bar\nu_1(\bar\nu_1+\rho))}{\bar\nu_{1,2} \bar\nu_{1,3} \bar\nu_{1,4}} e^{\bar\nu_1(T-t)} \right. \nonumber \\ 
    & \hspace{26pt} + \frac{(\bar\nu_2-\rho)(2\varrho(\bar\nu_2+\rho)+\gamma\kappa\bar\nu_2+2\lambda \bar\nu_2(\bar\nu_2+\rho))}{\bar\nu_{1,2} \bar\nu_{2,3} \bar\nu_{2,4}} e^{\bar\nu_2(T-t)} \nonumber \\
    & \hspace{26pt} - \frac{(\bar\nu_3-\rho)(2\varrho(\bar\nu_3+\rho)-\gamma\kappa\bar\nu_3+2\lambda \bar\nu_3(\bar\nu_3+\rho))}{\bar\nu_{1,3} \bar\nu_{2,3} \bar\nu_{3,4}} e^{\bar\nu_3(T-t)} \nonumber \\
    & \hspace{26pt} \left. + \frac{(\bar\nu_4-\rho)(2\varrho(\bar\nu_4+\rho)+\gamma\kappa\bar\nu_4+2\lambda \bar\nu_4(\bar\nu_4+\rho))}{\bar\nu_{1,4} \bar\nu_{2,4} \bar\nu_{3,4}} e^{\bar\nu_4(T-t)} \right), \label{def:Gbar3}
\end{align}
\begin{align} 
    \ol G_{4}(T-t) = & \, \frac{N (\bar\nu_1-\rho)(\bar\nu_2-\rho)(\bar\nu_3-\rho)(\bar\nu_4
    -\rho)}{4 \gamma\kappa\lambda\rho^2 } \nonumber \\
    & \left( \frac{(2\varrho(\bar\nu_1+\rho)+\gamma\kappa\bar\nu_1+2\lambda\bar\nu_1(\bar\nu_1+\rho))}{\bar\nu_{1,2} \bar\nu_{1,3} \bar\nu_{1,4}} e^{\bar\nu_1(T-t)} \right. \nonumber \\ 
    & \hspace{12pt} - \frac{(2\varrho(\bar\nu_2+\rho)+\gamma\kappa\bar\nu_2+2\lambda\bar\nu_2(\bar\nu_2+\rho))}{\bar\nu_{1,2} \bar\nu_{2,3} \bar\nu_{2,4}} e^{\bar\nu_2(T-t)} \nonumber\\
    & \hspace{12pt} + \frac{(2\varrho(\bar\nu_3+\rho)+\gamma\kappa\bar\nu_3+2\lambda\bar\nu_3(\bar\nu_3+\rho))}{\bar\nu_{1,3} \bar\nu_{2,3} \bar\nu_{3,4}} e^{\bar\nu_3(T-t)} \nonumber \\
    & \hspace{12pt} \left. - \frac{(2\varrho(\bar\nu_4+\rho)+\gamma\kappa\bar\nu_4+2\lambda\bar\nu_4(\bar\nu_4+\rho))}{\bar\nu_{1,4} \bar\nu_{2,4} \bar\nu_{3,4}} e^{\bar\nu_4(T-t)} \right) \label{def:Gbar4}
\end{align}
for all $t \in [0,T]$. Moreover, $\ol H$ introduced in~\eqref{def:Hbar} is given by
\begin{align}
    \ol H_1(T-t) = & \, \frac{\gamma\kappa\phi}{N \lambda} \left( -\frac{(\bar\nu_1+\rho)e^{\bar\nu_1(T-t)}}{\bar\nu_{1,2} \bar\nu_{1,3} \bar\nu_{1,4}}  +\frac{(\bar\nu_2+\rho)e^{\bar\nu_2(T-t)}}{\bar\nu_{1,2} \bar\nu_{2,3} \bar\nu_{2,4}} \right. \nonumber \\
    & \hspace{38pt} \left. -\frac{(\bar\nu_3+\rho)e^{\bar\nu_3(T-t)}}{\bar\nu_{1,3} \bar\nu_{2,3} \bar\nu_{3,4}} +\frac{(\bar\nu_4+\rho) e^{\bar\nu_4(T-t)} }{\bar\nu_{1,4} \bar\nu_{2,4} \bar\nu_{3,4}} \right), \label{def:Hbar1}
\end{align}
\begin{align}
    \ol H_2(T-t) = & \, \frac{\kappa (\bar\nu_1+\rho)(\bar\nu_2+\rho)(\bar\nu_3+\rho)(\bar\nu_4+\rho)}{2 N \rho^2} \nonumber \\
    & \left( -\frac{\bar\nu_1 e^{\bar\nu_1(T-t)}}{\bar\nu_{1,2} \bar\nu_{1,3} \bar\nu_{1,4}}  +\frac{\bar\nu_2 e^{\bar\nu_2(T-t)}}{\bar\nu_{1,2} \bar\nu_{2,3} \bar\nu_{2,4}}  -\frac{\bar\nu_3 e^{\bar\nu_3(T-t)}}{\bar\nu_{1,3} \bar\nu_{2,3} \bar\nu_{3,4}}  +\frac{\bar\nu_4 e^{\bar\nu_4(T-t)}}{\bar\nu_{1,4} \bar\nu_{2,4} \bar\nu_{3,4}}  \right), \label{def:Hbar2}
\end{align}
\begin{align}
    & \ol H_3(T-t) = \frac{\gamma\kappa}{N}  \label{def:Hbar3} \\
    & \left( \frac{\bar\nu_1(\bar\nu_1+\rho)e^{\bar\nu_1(T-t)}}{\bar\nu_{1,2} \bar\nu_{1,3} \bar\nu_{1,4}} - \frac{\bar\nu_2(\bar\nu_2+\rho)e^{\bar\nu_2(T-t)}}{\bar\nu_{1,2} \bar\nu_{2,3} \bar\nu_{2,4}} + \frac{\bar\nu_3(\bar\nu_3+\rho)e^{\bar\nu_3(T-t)}}{\bar\nu_{1,3} \bar\nu_{2,3} \bar\nu_{3,4}} - \frac{\bar\nu_4(\bar\nu_4+\rho)e^{\bar\nu_4(T-t)}}{\bar\nu_{1,4} \bar\nu_{2,4} \bar\nu_{3,4}}  \right), \nonumber
\end{align}
\begin{align}
    \ol H_4(T-t) = & \, \frac{1}{2\rho^2} \left( -\frac{\bar\nu_1(\bar\nu_1+\rho)(\bar\nu_2-\rho)(\bar\nu_3-\rho)(\bar\nu_4-\rho)e^{\bar\nu_1(T-t)}}{\bar\nu_{1,2} \bar\nu_{1,3} \bar\nu_{1,4}} \right. \nonumber \\
    & \hspace{26pt} + \frac{\bar\nu_2(\bar\nu_1-\rho)(\bar\nu_2+\rho)(\bar\nu_3-\rho)(\bar\nu_4-\rho)e^{\bar\nu_2(T-t)}}{\bar\nu_{1,2} \bar\nu_{2,3} \bar\nu_{2,4}} \nonumber \\
    & \hspace{26pt}- \frac{\bar\nu_3(\bar\nu_1-\rho)(\bar\nu_2-\rho)(\bar\nu_3+\rho)(\bar\nu_4-\rho)e^{\bar\nu_3(T-t)}}{\bar\nu_{1,3} \bar\nu_{2,3} \bar\nu_{3,4}} \nonumber \\
    & \hspace{26pt} \left. + \frac{\bar\nu_4(\bar\nu_1-\rho)(\bar\nu_2-\rho)(\bar\nu_3-\rho)(\bar\nu_4+\rho)e^{\bar\nu_4(T-t)}}{\bar\nu_{1,4} \bar\nu_{2,4} \bar\nu_{3,4}}  \right), \label{def:Hbar4}
\end{align}
for all $t\in[0,T]$.

Next, we compute the matrix exponential $Q (t) = \exp(F^N \cdot t) \in \mathbb{R}^{3 \times 3}$, introduced in \eqref{def:matrixExpE}, for all $t \in [0, \infty)$,  by diagonalizing the matrix $F = UDU^{-1}$ in~\eqref{def:F} with diagonal matrix $D \in \mathbb{R}^{3 \times 3}$ and invertible matrix $U\in \mathbb{R}^{3 \times 3}$. The eigenvalues $\nu_1,\nu_2,\nu_3$ are the roots of the equation
\begin{equation} \label{eq:cubicroot1}
    x^3 - \frac{2N\lambda\rho+\gamma\kappa}{2N\lambda} x^2 - \frac{\phi}{\lambda} x + \frac{\phi\rho}{\lambda} = 0.
\end{equation}
Introducing the constants
\begin{equation*}
a = - \frac{2N\lambda\rho+\gamma\kappa}{2\lambda N}, \quad b = - \frac{\phi}{\lambda}, \quad c = \frac{\phi\rho}{\lambda}, \quad p = b -\frac{a^2}{3}, \quad q = \frac{2a^3}{27}-\frac{ab}{3}+c,
\end{equation*}
and substituting $x = z-a/3$ in~\eqref{eq:cubicroot1} yields the equivalent equation $z^3+p z + q =0$ with discriminant $(q/2)^2 + (p/3)^3 < 0$. This implies that the latter has three real-valued distinct roots allowing for the analytical representations    
\begin{equation}
    \begin{aligned}
    \nu_1 = & \, - \sqrt{-\frac{4}{3} p} \cdot \cos\left(\frac{1}{3} \arccos \left(-\frac{q}{2} \cdot \sqrt{-\frac{27}{p^3}}\right) +\frac{\pi}{3}\right) - \frac{a}{3}, \\
    \nu_2 = & \, \sqrt{-\frac{4}{3} p} \cdot \cos\left(\frac{1}{3} \arccos \left(-\frac{q}{2} \cdot \sqrt{-\frac{27}{p^3}}\right) \right) - \frac{a}{3}, \\
    \nu_3 = & \, - \sqrt{-\frac{4}{3} p} \cdot \cos\left(\frac{1}{3} \arccos \left(-\frac{q}{2} \cdot \sqrt{-\frac{27}{p^3}}\right) -\frac{\pi}{3}\right) - \frac{a}{3}.
    \end{aligned}
\end{equation}
The corresponding eigenvectors are given by 
\begin{equation*} \label{def:eigenvectorbarF}
  v_i \triangleq
  \begin{pmatrix}
    -\frac{N(\nu_i-\rho)}{\kappa\gamma\nu_i} \\
     \frac{N(\nu_i - \rho)}{\kappa\gamma} \\
     1
    \end{pmatrix} \qquad (i=1,2,3).
\end{equation*}
Consequently, we have
\begin{equation*}
    D = \begin{pmatrix} 
    \nu_1 & 0 & 0 \\
    0 & \nu_2 & 0 \\
    0 & 0 & \nu_3
    \end{pmatrix}, \quad 
    U =
    \begin{pmatrix}
    -\frac{N(\nu_1-\rho)}{\kappa\gamma\nu_1} &  -\frac{N(\nu_2-\rho)}{\kappa\gamma\nu_2} &  -\frac{N(\nu_3-\rho)}{\kappa\gamma\nu_3} \\
     \frac{N(\nu_1 - \rho)}{\kappa\gamma} & \frac{N(\nu_2 - \rho)}{\kappa\gamma} & \frac{N(\nu_3 - \rho)}{\kappa\gamma} \\
     1 & 1 & 1
    \end{pmatrix},
\end{equation*}  
and, introducing the differences
\begin{equation*}
    \nu_{i,j} = \nu_{i} - \nu_{j} \qquad (i,j \in \{1,2,3\}),
\end{equation*}
we obtain
\begin{equation*}
    U^{-1} = \frac{1}{N \rho} \begin{pmatrix}
   -\frac{\gamma\kappa\phi\rho}{\lambda\nu_{1,2}\nu_{1,3}} & \frac{\gamma\kappa\rho\nu_1}{\nu_{1,2}\nu_{1,3}} & \frac{N\nu_1(\nu_2-\rho)(\nu_3-\rho)}{\nu_{1,2}\nu_{1,3}} \\
   \frac{\gamma\kappa\phi\rho}{\lambda\nu_{1,2}\nu_{2,3}} & -\frac{\gamma\kappa\rho\nu_2}{\nu_{1,2}\nu_{2,3}} & -\frac{N\nu_2(\nu_1-\rho)(\nu_3-\rho)}{\nu_{1,2}\nu_{2,3}}\\
   -\frac{\gamma\kappa\phi\rho}{\lambda\nu_{1,3}\nu_{2,3}} & \frac{\gamma\kappa\rho\nu_3}{\nu_{1,3}\nu_{2,3}} & \frac{N\nu_3(\nu_1-\rho)(\nu_2-\rho)}{\nu_{1,3}\nu_{2,3}}
  \end{pmatrix}.
\end{equation*} 
The matrix exponential $Q(t) = (Q_{ij}(t))_{1 \leq i,j \leq 3}$ is thus given by
\begin{equation}
  Q(t) = U \begin{pmatrix}
    e^{\nu_1 t} & 0 & 0 \\
    0 & e^{\nu_2 t} & 0 \\
    0 & 0 & e^{\nu_3 t} \\
  \end{pmatrix} U^{-1} \qquad (t \geq 0).
\end{equation}
Moreover, it follows that $G$ defined in~\eqref{def:G} can be computed as
\begin{align}
    G_1(T-t) = & \, \frac{\phi}{\lambda^2} \left( \frac{(\nu_1-\rho)(\varrho+\lambda\nu_1)}{\nu_1\nu_{1,2}\nu_{1,3}} e^{\nu_1(T-t)} \right. \nonumber \\
    & \hspace{28pt} - \frac{(\nu_2-\rho)(\varrho+\lambda\nu_2)}{\nu_2\nu_{1,2}\nu_{2,3}} e^{\nu_2(T-t)} \nonumber \\
    & \hspace{28pt} + \left. \frac{(\nu_3-\rho)(\varrho+\lambda\nu_3)}{\nu_3\nu_{1,3}\nu_{2,3}} e^{\nu_3(T-t)} \right), \label{def:G1}
\end{align}
\begin{align}
    G_2(T-t) = & \, \frac{1}{\lambda} \left( -\frac{(\nu_1-\rho)(\varrho+\lambda\nu_1)}{\nu_{1,2}\nu_{1,3}} e^{\nu_1(T-t)} \right. \nonumber \\
    & \hspace{22pt} + \frac{(\nu_2-\rho)(\varrho+\lambda\nu_2)}{\nu_{1,2}\nu_{2,3}} e^{\nu_2(T-t)} \nonumber \\
    & \hspace{22pt} \left. - \frac{(\nu_3-\rho)(\varrho+\lambda\nu_3)}{\nu_{1,3}\nu_{2,3}} e^{\nu_3(T-t)} \right), \label{def:G2}
\end{align}
\begin{align}
    G_3(T-t) = & \, \frac{(\nu_1-\rho)(\nu_2-\rho)(\nu_3-\rho)}{\gamma\kappa\lambda\rho} \nonumber \\
    & \left( -\frac{\varrho+\lambda\nu_1}{\nu_{1,2}\nu_{1,3}} e^{\nu_1(T-t)} + \frac{\varrho+\lambda\nu_2}{\nu_{1,2}\nu_{2,3}} e^{\nu_2(T-t)} -\frac{\varrho+\lambda\nu_3}{\nu_{1,3}\nu_{2,3}} e^{\nu_3(T-t)} \right), \label{def:G3}
\end{align}
for all $t \in [0,T]$. Similarly, we obtain for $H$ defined in~\eqref{def:H}
\begin{align}
    H_1(T-t) = \frac{\gamma\kappa\phi}{N\lambda} \left( -\frac{e^{\nu_1(T-t)}}{\nu_{1,2}\nu_{1,3}} + \frac{e^{\nu_2(T-t)}}{\nu_{1,2}\nu_{2,3}} -\frac{e^{\nu_3(T-t)}}{\nu_{1,3}\nu_{2,3}}\right), \label{def:H1}
\end{align}
\begin{align}
    H_2(T-t) = \frac{\gamma\kappa}{N} \left( \frac{\nu_1 e^{\nu_1(T-t)}}{\nu_{1,2}\nu_{1,3}} - \frac{\nu_2 e^{\nu_2(T-t)}}{\nu_{1,2}\nu_{2,3}} +\frac{\nu_3 e^{\nu_3(T-t)}}{\nu_{1,3}\nu_{2,3}}\right), \label{def:H2}
\end{align}
\begin{align}
    H_3(T-t) = & \, \frac{\nu_1 (\nu_2-\rho)(\nu_3-\rho)e^{\nu_1(T-t)}}{\nu_{1,2}\nu_{1,3}} - \frac{\nu_2 (\nu_1-\rho)(\nu_3-\rho)e^{\nu_2(T-t)}}{\nu_{1,2}\nu_{2,3}} \nonumber \\ 
    & +\frac{\nu_3 (\nu_1-\rho)(\nu_2-\rho)e^{\nu_3(T-t)}}{\nu_{1,3}\nu_{2,3}}, \label{def:H3}
\end{align}
for all $t \in [0,T]$.


\subsection{Infinite player game}
\label{sec-mat-infin} 

Computing the matrix exponential $\tilde R (t) = \exp(\tilde B \cdot t) \in \mathbb{R}^{3 \times 3}$, introduced in \eqref{def:matrixExpBtilde}, for all $t \in [0, \infty)$,  is very similar to the computations in Section \ref{sec-mat-fin}. Again, we decompose the matrix $\tilde B = \tilde{U}\tilde{D}\tilde{U}^{-1}$ in~\eqref{def:Btilde} into a diagonal matrix $\tilde{D} \in \mathbb{R}^{3 \times 3}$ and an invertible matrix $\tilde{U}\in \mathbb{R}^{3 \times 3}$. The eigenvalues $\tilde{\nu}_1,\tilde{\nu}_2,\tilde{\nu}_3$ are the roots of the equation
\begin{equation} \label{eq:cubicroot2}
    x^3 + \frac{2\lambda\rho+\gamma\kappa}{2\lambda} x^2 - \frac{\phi}{\lambda} x - \frac{\rho\phi}{\lambda} = 0.
\end{equation}
Introducing the constants
\begin{equation*}
\tilde a = \frac{2\lambda\rho+\gamma\kappa}{2\lambda}, \quad \tilde b = - \frac{\phi}{\lambda}, \quad \tilde c = -\frac{\phi\rho}{\lambda}, \quad \tilde p = \tilde b -\frac{\tilde{a}^2}{3}, \quad \tilde q = \frac{2\tilde{a}^3}{27}-\frac{\tilde{a}\tilde{b}}{3}+\tilde{c},
\end{equation*}
and substituting $x = z-\tilde{a}/3$ in~\eqref{eq:cubicroot2} yields the equivalent equation $z^3+\tilde p z + \tilde q =0$ with discriminant $(\tilde{q}/2)^2 + (\tilde{p}/3)^3 < 0$. As above this implies that the latter has three real-valued distinct roots allowing for the analytical representations    
\begin{equation}
    \begin{aligned}
    \tilde\nu_1 = & \, - \sqrt{-\frac{4}{3} \tilde p} \cdot \cos\left(\frac{1}{3} \arccos \left(-\frac{\tilde q}{2} \cdot \sqrt{-\frac{27}{\tilde{p}^3}}\right) +\frac{\pi}{3}\right) - \frac{\tilde a}{3}, \\
    \tilde\nu_2 = & \, \sqrt{-\frac{4}{3} \tilde p} \cdot \cos\left(\frac{1}{3} \arccos \left(-\frac{\tilde q}{2} \cdot \sqrt{-\frac{27}{\tilde{p}^3}}\right) \right) - \frac{\tilde a}{3}, \\
    \tilde\nu_3 = & \, - \sqrt{-\frac{4}{3} \tilde p} \cdot \cos\left(\frac{1}{3} \arccos \left(-\frac{\tilde q}{2} \cdot \sqrt{-\frac{27}{\tilde{p}^3}}\right) -\frac{\pi}{3}\right) - \frac{\tilde a}{3}.
    \end{aligned}
\end{equation}
The corresponding eigenvectors are given by 
\begin{equation*} \label{def:eigenvectortildeB}
  \tilde{v}_i \triangleq
  \begin{pmatrix}
    -\frac{1}{\tilde{\nu}_i} \\
     \frac{\gamma}{\tilde{\nu}_i + \rho} \\
     1
    \end{pmatrix} \qquad (i=1,2,3).
\end{equation*}
That is, we obtain
\begin{equation*}
    \tilde D = \begin{pmatrix} 
    \tilde{\nu}_1 & 0 & 0 \\
    0 & \tilde{\nu}_2 & 0 \\
    0 & 0 & \tilde{\nu}_3
    \end{pmatrix}, \quad 
    \tilde U =
    \begin{pmatrix}
    -\frac{1}{\tilde{\nu}_1} & -\frac{1}{\tilde{\nu}_2} & -\frac{1}{\tilde{\nu}_3}\\
     \frac{\gamma}{\tilde{\nu}_1 + \rho} & \frac{\gamma}{\tilde{\nu}_2 + \rho} & \frac{\gamma}{\tilde{\nu}_3 + \rho} \\
     1 & 1 & 1
    \end{pmatrix},
\end{equation*}  
and, introducing the differences
\begin{equation*}
    \tilde\nu_{i,j} = \tilde\nu_{i} - \tilde\nu_{j} \qquad (i,j \in \{1,2,3\}),
\end{equation*}
we have that
\begin{equation*}
    \tilde U^{-1} = \frac{1}{\gamma\rho} \begin{pmatrix}
   -\frac{\gamma\rho\phi(\tilde\nu_1+\rho)}{\lambda\tilde\nu_{1,2}\tilde\nu_{1,3}} & -\frac{\tilde\nu_1(\tilde\nu_1+\rho)(\tilde\nu_2+\rho)(\tilde\nu_3+\rho)}{\tilde\nu_{1,2}\tilde\nu_{1,3}} & \frac{\gamma\rho\tilde\nu_1(\tilde\nu_1+\rho)}{\tilde\nu_{1,2}\tilde\nu_{1,3}} \\
   \frac{\gamma\rho\phi(\tilde\nu_2+\rho)}{\lambda\tilde\nu_{1,2}\tilde\nu_{2,3}} & \frac{\tilde\nu_2(\tilde\nu_1+\rho)(\tilde\nu_2+\rho)(\tilde\nu_3+\rho)}{\tilde\nu_{1,2}\tilde\nu_{2,3}} & -\frac{\gamma\rho\tilde\nu_2(\tilde\nu_2+\rho)}{\tilde\nu_{1,2}\tilde\nu_{2,3}}\\
   -\frac{\gamma\rho\phi(\tilde\nu_3+\rho)}{\lambda\tilde\nu_{1,3}\tilde\nu_{2,3}} & -\frac{\tilde\nu_3(\tilde\nu_1+\rho)(\tilde\nu_2+\rho)(\tilde\nu_3+\rho)}{\tilde\nu_{1,3}\tilde\nu_{2,3}} & \frac{\gamma\rho\tilde\nu_3(\tilde\nu_3+\rho)}{\tilde\nu_{1,3}\tilde\nu_{2,3}}
  \end{pmatrix}.
\end{equation*} 
Therefore, the matrix exponential $\tilde R(t) = (\tilde R_{ij}(t))_{1 \leq i,j \leq 3}$ can be computed as
\begin{equation}
  \tilde R(t) = \tilde U \begin{pmatrix}
    e^{\tilde\nu_1 t} & 0 & 0 \\
    0 & e^{\tilde\nu_2 t} & 0 \\
    0 & 0 & e^{\tilde\nu_3 t} \\
  \end{pmatrix} \tilde U^{-1} \qquad (t \geq 0),
\end{equation}
and we obtain that $\tilde K$ defined in~\eqref{def:Ktilde} is given by
\begin{align}
    \tilde K_1(T-t) = & \, \frac{\phi}{2\lambda^2} \left( \frac{2\varrho(\tilde{\nu}_1 + \rho)+\kappa\gamma\tilde{\nu}_1+2\lambda\tilde{\nu}_1(\tilde{\nu}_1+\rho)}{\tilde\nu_1\tilde{\nu}_{1,2}\tilde{\nu}_{1,3}} e^{\tilde\nu_1(T-t)} \right. \nonumber \\ 
    & \, \hspace{30pt} \left. - \frac{2\varrho(\tilde{\nu}_2+\rho)+\kappa\gamma\tilde{\nu}_2 + 2\lambda\tilde{\nu}_2 (\tilde{\nu}_2+\rho) }{\tilde\nu_2\tilde{\nu}_{1,2}\tilde{\nu}_{2,3}} e^{\tilde\nu_2(T-t)} \right. \nonumber \\
    & \, \left. \hspace{30pt} + \frac{2\varrho(\tilde{\nu}_3+\rho)+\kappa\gamma\tilde{\nu}_3+2\lambda\tilde{\nu}_3(\tilde{\nu}_3+\rho) }{\tilde\nu_3\tilde{\nu}_{1,3}\tilde{\nu}_{2,3}} e^{\tilde\nu_3(T-t)} \right), \label{def:Ktilde1}
\end{align}
\begin{align}
    \tilde K_2(T-t) = & \, \frac{1}{2\gamma\lambda\rho} \label{def:Ktilde2} \\
    & \left(\frac{(2\varrho(\tilde{\nu}_1 + \rho)+\kappa\gamma\tilde\nu_1+2\lambda\tilde{\nu}_1(\tilde{\nu}_1 + \rho))(\tilde{\nu}_2 + \rho)(\tilde{\nu}_3 + \rho)}{\tilde{\nu}_{1,2}\tilde{\nu}_{1,3}} e^{\tilde\nu_1(T-t)} \right. \nonumber \\
    & \hspace{10pt} - \frac{(2\varrho(\tilde{\nu}_2 + \rho)+\kappa\gamma\tilde\nu_2+2\lambda\tilde{\nu}_2(\tilde{\nu}_2 + \rho))(\tilde{\nu}_1 + \rho)(\tilde{\nu}_3 + \rho)}{\tilde{\nu}_{1,2}\tilde{\nu}_{2,3}} e^{\tilde\nu_2(T-t)} \nonumber \\
    & \left. \hspace{10pt} + \frac{(2\varrho(\tilde{\nu}_3 + \rho)+\kappa\gamma\tilde\nu_3+2\lambda\tilde{\nu}_3(\tilde{\nu}_3 + \rho))(\tilde{\nu}_1 + \rho)(\tilde{\nu}_2 + \rho)}{\tilde{\nu}_{1,3}\tilde{\nu}_{2,3}} e^{\tilde\nu_3(T-t)} \right), \nonumber
\end{align}
\begin{align}
    \tilde K_3(T-t) = & \, \frac{1}{2\lambda} \left( -\frac{2\varrho(\tilde{\nu}_1 + \rho)+\kappa\gamma\tilde\nu_1+2\lambda\tilde{\nu}_1(\tilde{\nu}_1 + \rho)}{\tilde{\nu}_{1,2}\tilde{\nu}_{1,3}} e^{\tilde\nu_1(T-t)} \right. \nonumber \\
    & \, \hspace{24pt} + \frac{2\varrho(\tilde{\nu}_2 + \rho)+\kappa\gamma\tilde\nu_2 + 2\lambda\tilde{\nu}_2 (\tilde{\nu}_2 + \rho)}{\tilde{\nu}_{1,2}\tilde{\nu}_{2,3}} e^{\tilde\nu_2(T-t)} \nonumber \\
    & \, \left. \hspace{24pt} - \frac{2\varrho(\tilde{\nu}_3 + \rho)+\kappa\gamma\tilde\nu_3 + 2\lambda\tilde{\nu}_3 (\tilde{\nu}_3 + \rho)}{\tilde{\nu}_{1,3}\tilde{\nu}_{2,3}} e^{\tilde\nu_3(T-t)} \right), \label{def:Ktilde3}
\end{align}
for all $t \in [0,T]$.

\section*{Acknowledgments}
We are very grateful to the Associate Editor and to the anonymous referees for careful reading of the manuscript and for a number of useful comments and suggestions that significantly improved this paper.

\section*{Data availability statement}
No data was used in this article.


\end{document}